\documentclass[a4paper,12pt]{article}
\usepackage[english]{babel}
\usepackage{amsmath, amsfonts, amsthm, mathrsfs, latexsym, multicol}  
\usepackage{srcltx}
\usepackage{pgf,epsfig,multicol,graphics,pdfsync,url}

\newcommand{\abs}[1]{\lvert#1\rvert}
\newcommand{\norm}[1]{\lVert#1\rVert}
\newcommand{\field}[2]{\mathbb{#1}^{#2}}

\newcommand{\hilbert}[1]{\mathscr{#1}}
\newcommand{\fock}{\hilbert{F}}

\newcommand{\epsi}{\varepsilon}

\newcommand{\E}{{\mathrm{e}}}
\newcommand{\I}{\mathrm{i}}
\newcommand{\Hf}{H_{\mathrm{f}}}
\newcommand{\hepsiop}{{H}^{\epsi}}
\newcommand{\hzeroepsiop}{ {H}_{0}^{\epsi}}

\newcommand{\hfree}{{H}_{\mathrm{free}}^{\epsi}}

\newcommand{\vrm}[1]{V_{\mathrm{#1}}}

\newcommand{\me}{m_{\mathrm{el}}}
\newcommand{\mn}{m_{\mathrm{nuc}}}
\newcommand{\Or}{{\mathcal{O}}}
\newcommand{\R}{{\mathbb{R}}}
\newcommand{\C}{{\mathbb{C}}}
\newcommand{\N}{{\mathbb{N}}}

\newcommand{\D}{{\mathrm{d}}}
\newcommand{\Hi}{ \mathcal{H} }

\newcommand{\rom}{\renewcommand{\labelenumi}{{\rm(\roman{enumi})}}}

\newcommand{\brom}{\begin{enumerate}\rom}
\newcommand{\erom}{\end{enumerate}}

\newtheorem{theorem}{Theorem}
\newtheorem*{theorem*}{Theorem}
\newtheorem{lemma}{Lemma}
\newtheorem*{lemma*}{Lemma}

\newtheorem{proposition}{Proposition}
\newtheorem{corollary}{Corollary}

\title{Spontaneous decay of resonant energy levels for molecules with moving nuclei}
\author{  Stefan Teufel and Jakob Wachsmuth\\[4mm] \small Mathematisches Institut der Universit\"at T\"ubingen\\  \small  Auf der Morgenstelle 10, 72076 T\"ubingen, Germany\\}
 \date{  September 2, 2011}
\begin{document}
\maketitle

\begin{abstract}\small
We consider the Pauli-Fierz Hamiltonian with dynamical nuclei and investigate the transitions between the resonant electronic energy levels under the assumption that there are no free photons in the beginning. Coupling the limits of small fine structure constant and of heavy nuclei  allows us to 
prove the validity of the Born-Oppenheimer approximation at leading order  and to provide a simple formula for the rate of spontaneous decay.  
\end{abstract}

\tableofcontents

\newpage

\section{Introduction}

In the quantum mechanical description of atoms and molecules one usually neglects the coupling to the radiation field and thus the possibility of emission or absorption of photons. The charged nuclei and electrons only interact via the static Coulomb interaction. Still the predictions for the spectra of atoms and molecules are in very good agreement with experimental data usually gathered through interaction with light. Also the predictions for the dynamical behavior of molecules agree with the motion observed e.g.\ in chemical reactions.
The reason for the good agreement lies in the smallness of the fine structure constant $\alpha\approx\frac{1}{137}$ that determines the strength of the coupling to the radiation field.

It is by now well understood even on a mathematical level how the coupling to the quantized radiation field changes the spectrum of the Hamiltonian operator describing an atom or a {\em static} molecule, e.g.\ \cite{BFS,HHH, AFFS,Fa,HaSe}. 
The quantum mechanical eigenstates become resonances, with energies   close  to the original eigenvalues, that decay nearly exponentially with a rate that can be computed perturbatively.

The quantum mechanical understanding of the {\em dynamics} of molecules is  based on  the Born-Oppenheimer approximation. Roughly speaking one assumes that if the electrons are initially in a certain eigenstate relative to the nuclear positions (e.g.\ in the ground state), they will remain in the ``same'' eigenstate relative to the nuclear positions  even when the latter change. The electronic state is ``slaved'' in this sense, but by energy conservation the electronic energy level serves as an effective potential for the motion of the  nuclei. 
The validity of the approximation was proved in various versions \cite{HaJo1,MaSo1,SpTe,PST2,MaSo2}. It is an adiabatic approximation relying on the fact that due to their large mass the nuclei move slowly compared to the lighter electrons. While transitions between different electronic levels (so-called non-adiabatic transitions) are possible even without coupling to the radiation field, the probabilities   for such  transitions are usually exponentially small in the adiabatic parameter and thus negligible.  

The content of this work is a mathematical analysis of molecular dynamics with the coupling to the quantized radiation field taken into account. 
Our first result is  the validity of the Born-Oppenheimer approximation at leading order. This is of course expected, since the validity of the Born-Oppenheimer approximation has been confirmed  experimentally in countless situations. Again the reason is the smallness of $\alpha$ which leads to small decay rates on the time scale set by the nuclear motion. This result is a rather straightforward consequence of combining the known quantum mechanical results on the Born-Oppenheimer approximation  with standard time-dependent perturbation theory.

The main mathematical and physical problem solved in this paper is the determination of the rates of spontaneous emission for dynamical molecules, i.e.\ for situations where the nuclei undergo a nontrivial dynamics. Let us briefly discuss an example where these rates are relevant. In Figure~\ref{figure} some electronic energy levels for a di-atomic molecule are schematically plotted as a function of the nuclear separation $R$. The ground state energy actually behaves like $R^{-6}$ for large $R$ and thus leads to a rather small attractive force for separated atoms in the ground state, the so-called van-der-Waals force. One strategy to accelerate the production of  dimers  is to excite one of the atoms, so that the molecular system is in the first excited state that behaves like $R^{-3}$ and thus leads to a stronger attractive force. Once the nuclei come close,  the system goes either into the ground state by   spontaneous 
emission of a photon or the nuclei will only scatter and separate again. One is thus interested in the probability for spontaneous emission within a finite time interval while the nuclei are sufficiently close. However, this probability is not governed by a fixed decay rate since the electronic state and thus the lifetime of the resonance changes with the location of the nuclei. In particular no exponential decay law can be expected.

\begin{figure}[h]\label{figure}
\begin{center}
\includegraphics[height=9cm]{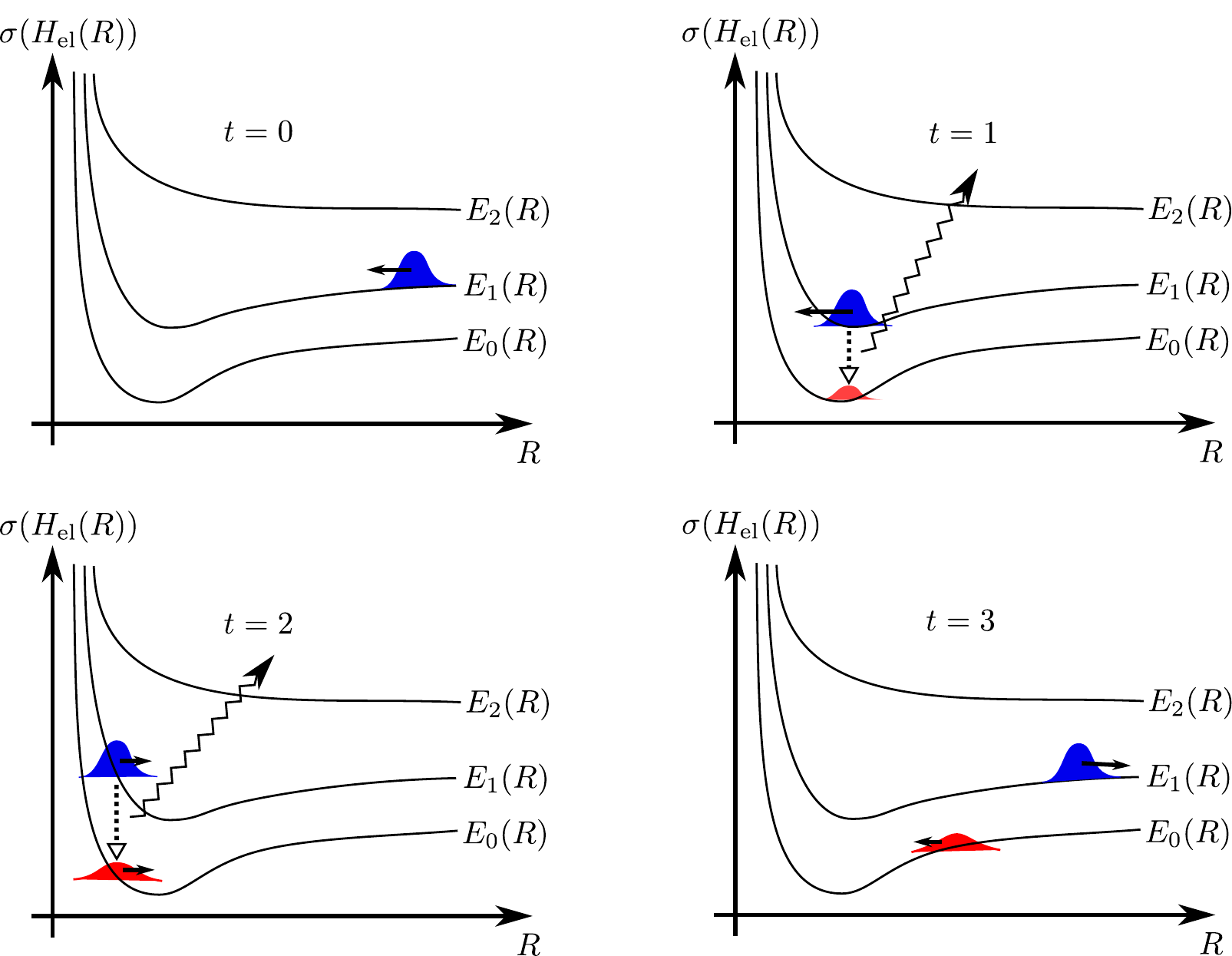}
\end{center}
\caption{While moving in the electronic surface $E_1$, there is a configuration dependent probability for a transition to the ground state surface $E_0$ through spontaneous emission of a photon.}
\end{figure}

Our main result is an explicit time-dependent formula for the probability of spontaneous  decay of a dynamical molecule through emission of a photon, for finite times on the natural time scale of molecular dynamics. Since on this time scale the probability for spontaneous emission 
is quite small, it is far from straightforward to determine its leading order expression and to show that the remainder terms are even smaller. In particular we need to carefully separate the three time-scales given by the slow nuclear motion, the intermediate electronic motion and the fast photons. 
The main idea of our proof is the construction of subspaces that correspond to specific electronic states relative to the nucleonic configuration and momentum   that are dressed by a cloud of virtual photons. The restriction of the full dynamics to these subspaces is the Born-Oppenheimer approximation. A transition between two subspaces corresponds to a change in the electronic state with simultaneous emission or absorption of a free photon. 
 
Before we can make these ideas more precise, we have to explain the mathematical model. For a molecule with $l$ nuclei and $r$ electrons the Hamiltonian is (in atomic units where  $\hbar=1$ and $c=1$)
\[
H_{\rm mol} = -\frac{1}{2m_{\rm nuc}} \sum_{j=1}^l \Delta_{x_j} -\frac{1}{2m_{\rm el}} \sum_{j=1}^r \Delta_{y_j} +  \alpha[V_{\mathrm{ee}}(y) + V_{\mathrm{en}}(x, y) + V_{\mathrm{nn}}(x)]
\]
acting on the state space $\Hi_{\rm mol} := \Hi_{\rm nuc}\otimes \Hi_{\rm el} := L^2(\R_x^{3l}) \otimes L^2_{\rm as}(\R_y^{3r})$.
Here $x = (x_{1}, \ldots, x_{l})$ denotes the configuration of the $l$ nuclei, $y = (y_{1}, \ldots, y_{r})$ the configuration of the $r$ electrons and $\alpha$ is the fine structure constant
\[ 
\alpha := \frac{e^{2}}{4\pi}\approx \frac{1}{137}\,.
\]
For notational simplicity we assume that all the nuclei have the same mass $m_{\rm nuc}$ and denote the electron mass by $m_{\rm el}$. We also disregard spin as it  would only complicate notation and not change the results.
For the moment let $\vrm{ee}, V_{\rm en}$ and $\vrm{nn}$ be the Coulomb potentials between electrons, electrons and nuclei, and nuclei respectively.

 Taking  into account also the coupling to the quantized radiation field, the Hamiltonian for the system becomes
\begin{eqnarray*}
\tilde{H}
 & := &
  \frac{1}{2\mn}\sum_{j=1}^{l}\big(p_{j, x} - 2\pi^{\frac{1}{2}}\alpha^{\frac{1}{2}}Z_{j}A_\Lambda(x_{j})\big)^{2} + \frac{1}{2\me}\sum_{j=1}^{r}\big(p_{j, y} - 2\pi^{\frac{1}{2}}\alpha^{\frac{1}{2}}A_{\Lambda}(y_{j})\big)^{2}\nonumber \\[1mm]
  &&+\, H_{\rm f} \;\;+\;\;  \alpha[V_{\mathrm{ee}}(y)\,+\, V_{\mathrm{en}}(x, y) \,+ \,V_{\mathrm{nn}}(x)]
 \end{eqnarray*}
where $p_{j,x} := -\I\nabla_{x_j}$ and $p_{j,y} := -\I\nabla_{y_j}$.
It acts on the Hilbert space
\[
\Hi := \Hi_{\rm mol} \otimes\fock,
\]
where $\fock$ is the photonic Fock space.

The nuclear charge in multiples of the electron charge is denoted by $Z_j$, $H_{\rm f}$ is the Hamiltonian of the free field and $A_{\Lambda}$ is the quantized transverse vector potential in the Coulomb gauge, with a sharp ultraviolet cutoff $\Lambda$, needed to make the Hamiltonian $\tilde{H}$ a well-defined self-adjoint operator. More explicitly
\[
A_{\Lambda}(q) =\frac{1}{(2\pi)^\frac{3}{2}} \sum_{\lambda=1}^{2}\int_{|k|\leq \Lambda} \frac{\D^3 k}{\sqrt{2\abs{k}}}\, e_{\lambda}(k) \left(\E^{\I k\cdot q}\,a(k, \lambda) + \E^{-\I k\cdot q}\,a^{*}(k, \lambda)\right), \,\, q\in\field{R}{3},
\]
where $a^{*}(k, \lambda)$ and $a(k, \lambda)$ are the standard creation and annihilation operators and $\{e_{\lambda}\}_{\lambda=1, 2}$ are the photon polarization vectors. Note that we use a sharp ultraviolet cutoff just to simplify notation. All our proofs work without any changes for smooth cutoffs.
Physically the cutoff is irrelevant for the problem at hand, as long as it is large compared to the energy of the emitted photons.

As explained before, the validity of the Born-Oppenheimer approximation rests on the fact that the coupling to the field is small and we will be interested in the asymptotics for small $\alpha$. However, since also the Coulomb interaction depends on $\alpha$ (it is a consequence of coupling to the field after all) the size of an atom or molecule as well as the electronic energy levels depend on $\alpha$.  To understand atoms and molecules by perturbation theory in $\alpha$ one thus switches to $\alpha$-dependent units where the typical sizes and the typical energies are independent of $\alpha$: introducing the Bohr-radius $\eta$ and the Rydberg-energy $\mu$ as
\[
 \eta = \frac{1}{2\me\alpha}\,, \quad \mu = 2\me\alpha^{2}\,,
\]
one implements the change of units on the one-particle configuration space as
\[
U_\eta :  L^2(\R^3_z,\D z) \to  L^2(\R^3_z,\D z)\,,\quad (U_\eta\psi)(z)= \eta^{3/2} \psi(\eta  z)
\]
and on the one-photon momentum space as
\[
U_\mu :  L^2(\R^3_k,\D k) \to  L^2(\R^3_k,\D k) \,,\quad (U_\mu \phi)(k) = \mu^{3/2} \phi(\mu  k)\,.
\]
These transformations are canonically lifted to a unitary $U_\alpha$ on the full  Hilbert space~$\Hi$. 
Finally we fix the  ultraviolet cutoff in units of Rydberg to some finite value $\Lambda_{0}<\infty$,  
\[ 
\Lambda = \mu\,\Lambda_{0}\,.
\]
A straightforward computation shows that in the new units the Hamiltonian becomes
\begin{eqnarray}
H^{\epsi, \alpha} \;\;:=\;\; \tfrac{1}{\mu}\,U_\alpha\, \tilde{H}\,U_\alpha^{*}&=&
 \epsi^{2}\,\sum_{j=1}^{l}\left(p_{j, x} - 2\pi^{1/2}\alpha^{3/2}Z_{j}A_{\Lambda_0}(\alpha x_{j})\right)^{2} \nonumber\\
&&+\,\sum_{j=1}^{r}\big(p_{j, y} - 2\pi^{1/2}\alpha^{3/2}A_{\Lambda_0}(\alpha y_{j})\big)^{2} \nonumber\\[2mm] &&+\,H_{\rm f}\;+\; \vrm{ee}(y) + \vrm{en}(x, y) + \vrm{nn}(x) \,,\label{hepsialpha}
\end{eqnarray}
where we abbreviate
\[ 
\epsi := \bigg(\frac{\me}{\mn}\bigg)^{1/2}\,.
\]
Note that even for the lightest nuclei $\epsi$ is already rather small,
\[
\epsi  \;\approx\;\begin{cases} \frac{1}{43} & \mn = m_{\mathrm{p}},\\
                                                         \frac{1}{136} & \mn = 10m_{\mathrm{p}},
                                           \end{cases}
					   \qquad m_{\mathrm{p}}: \,\textrm{proton mass}\,.
\]
The physical Hamiltonian \eqref{hepsialpha} depends on the two small dimensionless parameters $\epsi$ and $\alpha$. The smallness of $\epsi$ is the basis for the Born-Oppenheimer approximation in molecular dynamics and the smallness of $\alpha$ allows for a perturbative understanding of electronic resonances.

Our aim is to construct a Born-Oppenheimer expansion for the resonances of the molecular system. To explain exactly what we mean by this statement we first recall   some known results about the two limit cases which are contained in the Hamiltonian $H^{\epsi, \alpha}$, the case $\epsi=0, \alpha\neq 0$ (a molecule coupled to electromagnetic field with clamped nuclei) and the case $\epsi\neq 0, \alpha=0$ (a molecule with dynamical nuclei but no coupling to the field).  

\subsection[Electronic resonances for fixed nuclei ]{Electronic resonances for fixed nuclei ($\epsi=0 ,\,\alpha >0$)}

When $\epsi=0$, the Hamiltonian depends parametrically on the nuclear configuration $x$ and becomes
\begin{equation*}
H^{\epsi=0,\alpha}(x) =: H_{{\rm el}}(x) + \Hf + W(\alpha, \Lambda_0),
\end{equation*}
where the electronic Hamiltonian
\begin{equation}\label{hel}
H_{{\rm el}}(x) := \sum_{j=1}^{r}p_{j, y}^{2} + \vrm{ee}(y) + V_{{\rm en}}(x, y) + \vrm{nn}(x)
\end{equation}
is for every fixed $x$ a self-adjoint operator on $\Hi_{\rm el} = L^{2}(\field{R}{3r}_{y})$. We assume that the spectrum of $H_{{\rm el}}(x)$ is of the form
\begin{equation*}
\sigma(H_{{\rm el}}(x)) = \{E_{0}(x), E_{1}(x), \ldots\}\cup [\Sigma(x), \infty),
\end{equation*}
where $E_{0}(x)<E_{1}(x)<E_{2}(x)<\ldots\leq\Sigma(x)$ are eigenvalues of finite multiplicity below $\Sigma(x)$, possibly with an accumulation point at $\Sigma(x)$ and absolutely continuous spectrum in $[\Sigma(x), +\infty)$.
As shown in \cite{Zi}, if 
\begin{equation}\label{zishlin}
r\leq\sum_{j=1}^{l}Z_{j}
\end{equation}
then $H_{\mathrm{el}}(x)$ has an infinite number of eigenvalues below the threshold $\Sigma(x)$.

Under the same hypothesis \eqref{zishlin} it was shown in \cite{LiLo}, using a binding condition introduced in \cite{GLL}, that $H^{\epsi = 0, \alpha}$ has a ground state $E(x)$ for every $\alpha>0$ (using a smooth ultraviolet cutoff). The existence of the ground state for small values of the fine structure constant $\alpha$ has   been shown before in \cite{BFS}.

It is expected that the electronic eigenvalues $E_1(x), E_2(x),\ldots$ turn into resonances and that apart from the ground state the spectrum  of $H^{\epsi = 0, \alpha}$ is absolutely continuous.
It was shown in \cite{BFS} (see also \cite{AFFS, HHH}) for the case $l=1$ (an atom)
that the eigenvalues $\{E_{j}\}_{j>0}$ become resonances   in the sense of the Aguilar-Balslev-Combes-Simon   theory (\cite{HiSi} chapters 16--18, \cite{ReSi4} sections XII.6, XIII.10, \cite{Si}). 
We quote a typical result (cf.\ e.g.\
 Corollary~2 from \cite{HHH}) on the exponential decay of resonant states  without giving technical details.\\[2mm]
\noindent{\bf Almost exponential decay of atomic resonances.} {\em
Let $P_{j}$ be the spectral projection of $H_{\rm el}= H_{\rm el}(0)$ corresponding to 
the eigenvalue  $E_j$, $j\not=0$, and let $Q_0$ be the projection on the Fock-vacuum.  
For $\Psi\in\,{\rm Ran}P_j\otimes Q_0$ normalized to one it holds that
\[
\| (P_j\otimes Q_0) \,\E^{-\I t H^{\epsi=0,\alpha} }  \Psi \|^2= \E^{-t \alpha^3 \gamma} + b(\alpha, t)\,,
\]
where $b(\alpha, t) = \Or(\alpha^\frac{1}{2})$ uniformly in time and $\gamma>0$.
It follows that  the lifetime of the resonance is of order $\alpha^{-3}$.
}\\[2mm]
The difficult part in proving such results is to control the error term uniformly in time. For short times $t\ll \alpha^{-3}$ the decay rate into any other state $E_i<E_j$ can be easily computed by a perturbative argument known as Fermi's golden rule. \\[2mm]
\noindent{\bf Fermi's golden rule.} {\em 
Let $P_{j}$ and $P_i$ be the spectral projections of $H_{\rm el}= H_{\rm el}(0)$ corresponding to 
the eigenvalues  $E_i<E_j$  and let $Q_0$ be the projection on the Fock-vacuum. Then for $\Psi\in\,{\rm Ran}P_j\otimes Q_0$ normalized to one it holds that
\begin{equation}\label{FGR}
\| (P_i\otimes {  1}) \,\E^{-\I t H^{\epsi=0,\alpha} }\Psi \|^2= 
\tfrac{4}{3} \,\alpha^3 (E_j-E_i)^3 |D_{ij}|^2 \,  t \;+\;o(\alpha^3)
\end{equation}
uniformly on bounded time intervals, 
where $|D_{ij}|$ is the dipole-matrix element, c.f.\ Section~\ref{regime}.
}\\[2mm]
Since the natural time scale for nuclear dynamics is short in this sense, we will not be interested in results on exponential decay on long time scales for dynamical molecules (it is not even clear what this would exactly mean), but in explicit decay rates in the form of Fermi's golden rule. However, as will be explained in Section~\ref{regime}, for moving nuclei the decay rate depends on the configuration of the nuclei, which in turn changes quickly on the time scale of the resonance. As a consequence
one can not just adapt the usual perturbative argument in order to compute decay rates for dynamical molecules. But before coming to the full problem, let us first recall some basic facts about the case $\alpha=0$ and $\epsi>0$.

\subsection[Dynamical nuclei without coupling to the field]{Dynamical nuclei without coupling to the field\\ ($\epsi>0,\,\alpha =0$)}

This case is the setting of the standard time-dependent Born-Oppenheimer approximation (e.g.\ \cite{HaJo1,MaSo1,SpTe,Teu2,PST1,MaSo2}). The Hamiltonian has  the form
\begin{equation*}
H^{\epsi, \alpha=0} = -\epsi^{2}\Delta_{x} + H_{\mathrm{el}}(x)\,,
\end{equation*}  
where we omit the field Hamiltonian $\Hf$ because it commutes with the rest and is therefore irrelevant. 

For kinetic energies of order one (in units of Rydberg!)
 the nuclei have  velocities of order $ \epsi $. The   time scale on which the nuclei move distances of order one (in units of Bohr radii) are thus times of order $ \epsi^{-1} $. Hence it is natural to change the unit of time as well and to solve 
 the time-dependent  Schr\"odinger equation 
\begin{equation*}
\I\epsi\tfrac{\D}{\D t}\psi(t) = H^{\epsi, \alpha=0}\,\psi(t)\, .
\end{equation*} 
The long time-scale will be reflected in the following by the fact, that we evaluate unitary groups at times 
$\frac{t}{\epsi}$, i.e.\ we consider $\E^{-\I H \frac{t}{\epsi}}$ for $t$ of order one.

To avoid additional technicalities, one assumes that $H_{\mathrm{el}}$ is in a suitable sense a smooth function of $x$. This requires to introduce a smearing of the nuclear charge distribution $\varphi\in C_0^\infty(\R^3)$ with $\varphi \geq 0$ and $\int\varphi=1$.
The electronic repulsion remains unchanged
\begin{equation*}
V_{\mathrm{ee}}(y) = \sum_{n=1}^{r-1}\sum_{m=n+1}^{r}\frac{1}{\abs{y_{n}-y_{m}}},
\end{equation*}
while the electron-nucleon attraction and the nuclear repulsion become
\begin{equation*}
V_{\mathrm{en}}(x, y) = -\sum_{n=1}^{r} \sum_{m=1}^{l}  \left(\tfrac{2}{\pi}\right)^{1/2}Z_m \int_{\field{R}{3}}\D k\, \frac{\hat{\varphi}(k)}{\abs{k}^{2}}\E^{\I k\cdot(y_{n}-x_{m})}
\end{equation*}
and
\begin{equation*}
\vrm{nn}(x) = \sum_{n=1}^{l-1}\sum_{m=n+1}^{l}4\pi Z_nZ_m\int_{\field{R}{3}}\D k\, \frac{\abs{\hat{\varphi}(k)}^{2}}{\abs{k}^{2}}\E^{\I k\cdot(x_{n}-x_{m})} \,.
\end{equation*}

In \cite{MaSo2} a ``twisted'' pseudo-differential calculus is introduced, which, generalizing Hunziker's distortion analyticity method \cite{Hu}, allows to treat also the case of the unsmeared Coulomb potential.

Let $E_{j}(x)$ be an eigenvalue of the electronic Hamiltonian $H_{\mathrm{el}}(x)$ which is globally isolated by a gap from the rest of the spectrum. \\[2mm]
\noindent{\bf Definition of isolated electronic eigenvalues.} {\em 
Let for all $x\in\R^{3l}$ be 
  $ E_{j}(x) $  an eigenvalue of the electronic Hamiltonian $H_{\mathrm{el}}(x)$. The family $E_j(x)$   is called isolated, if  there exist two functions $f_{\pm}\in C_{\mathrm{b}}(\field{R}{3l}_{x}, \field{R}{})$ defining an interval $I(x)=[f_{-}(x), f_{+}(x)]$ such that 
\begin{equation*}
\sigma(H_{\mathrm{el}}(x)) \cap I(x) = E_{j}(x) \quad\textrm{and}\quad \inf_{x\in\field{R}{3l}}\mathrm{dist}\big(E_{j}(x), \sigma(H_{\mathrm{el}}(x))\setminus E_{j}(x)\big) >0 \, .
\end{equation*}
 }

This condition implies that $E_j(x)$ and the spectral projection $P_{j}(x)$ onto the eigen\-space of $E_{j}(x)$ are smooth functions of $x$, c.f.\ Lemma~\ref{derivPj}. We denote by $P_{j}$ the direct integral 
\begin{equation*}
P_{j} := \int^{\oplus}_{\field{R}{3l}}\D x\, P_{j}(x)
\end{equation*}
which acts on $\Hi_{\rm nuc} \otimes \Hi_{\rm el} \cong  L^{2}(\field{R}{3l}, \Hi_{\rm el})$.

The Born-Oppenheimer approximation rests on the observation that  the electronic state adjusts adiabatically to the slow motion of the nuclei, i.e., that the subspace $P_j \Hi_{\rm mol}$ is approxiamtely invariant under the time evolution.

 A rigorous version of this statement is the following theorem from \cite{SpTe}, which is also a special case of Proposition~\ref{TheoBOnoField} proven below. \\[3mm]
\noindent {\bf Leading order Born-Oppenheimer approximation.} {\em
The operator
\[
 \hepsiop_{j } :=   P_j H^{\epsi,\alpha=0} P_j   + (1-P_j) H^{\epsi,\alpha=0} (1-P_j  )
\] 
 is self-adjoint  on the domain $D$ of $H^{\epsi,\alpha=0}$ and satisfies
\begin{equation}\label{leadingBO0}
\left\|  \E^{-\I \frac{t}{\epsi}     \,H^{\epsi,\alpha=0}}  -\E^{-\I \frac{t}{\epsi} \hepsiop_{j }}       \right\|_{\mathcal{L}(D,\Hi_{\rm mol})} = \Or(\epsi|t|+\epsi)\,.
\end{equation}
}\\
This result is optimal in the sense that the difference is not smaller than order $\epsi$. However, the overlap 
\[
\| P_i \, \E^{-\I \frac{t}{\epsi}     \,H^{\epsi,\alpha=0}} \,P_j \| =\Or(\epsi)\quad \mbox{for $i\not= j$}
\]
that the true time evolution introduces between the different electronic subspaces does not correspond to actual transitions between electronic states. Indeed, the subspaces $P_j\Hi_{\rm mol}$ can be replaced by slightly deformed super\-adiabatic subspaces $P_j^\epsi\Hi_{\rm mol}$ that are invariant to higher order in $\epsi$. Physically  in $P_j^\epsi\Hi_{\rm mol}$ the electronic state now depends also on the velocity of the nuclei.

We construct such superadiabatic projections in Proposition~\ref{BOprop}. A straightforward consequence is   the following statement. \\[3mm]
\noindent {\bf Second order Born-Oppenheimer approximation.} {\em
The operator
\[
\tilde H^\epsi_{j } :=   P^\epsi_j H^{\epsi,\alpha=0} P^\epsi_j   + (1-P^\epsi_j) H^{\epsi,\alpha=0} (1-P^\epsi_j  )
\] 
 is self-adjoint  on the domain $D$ of $H^{\epsi,\alpha=0}$ and satisfies
\begin{equation}\label{leadingBO1}
\left\|\left( \E^{-\I \frac{t}{\epsi}     \,H^{\epsi,\alpha=0}}  - \E^{-\I \frac{t}{\epsi} \tilde H^\epsi_{j }}   \right)  {\bf 1}_{(-\infty,E)}(H^{\epsi,\alpha=0})   \right\|_{\mathcal{L}(\Hi_{\rm mol})} = \Or(\epsi^2|t|)\,.
\end{equation}
 }\\
Here an energy cutoff at an arbitrary but fixed energy $E$ is needed.
This improved approximation of the dynamics  is necessary for obtaining error terms smaller than the effect we are interested in, namely transitions between different electronic levels due to spontaneous emission of photons. But it turns out that a rigorous  control of these error terms requires  to prove (\ref{leadingBO0}) and (\ref{leadingBO1}) with respect to more general energy norms, which is the main new content of Propositions~\ref{TheoBOnoField} and~\ref{BOprop}. 

On the other hand, (\ref{leadingBO1})  can be shown with an error of order $\epsi^N$ for any $N\in\N$. Martinez and Sordoni even prove exponential error bounds without assuming a regularization on the nuclear charges, \cite{MaSo2}. However, the task of computing the exponentially small transition probabilities between superadiabatic subspaces (transitions that happen without emission of photons) is extremely difficult   even on a heuristic level, see \cite{HaJo2,BGT,BeGo}.

Another question is, whether one can dispose with the gap condition. At crossings of electronic eigenvalues the Born-Oppenheimer approximation breaks down and transitions between the levels occur at order $\epsi^0$, cf.\ \cite{LaTe} and references therein. 
For eigenvalues embedded into or at the threshold to continuous spectrum the rate of transition depends on the details of the model (see e.g.\ \cite{Ten} and \cite{TeTe}). 

Finally we remark that the importance of the Born-Oppenheimer approximation lies in the observation that  the diagonal Hamiltonian $  H^\epsi_{j }$,  when acting on states in the range of $P_j$, has an asymptotic expansion  starting with very simple terms,
\[
 H^\epsi_{j } P _j =  \left(-  (\epsi\nabla_x^{\rm Berry})^2 + E_j(x)\right) \,P_j  +\Or(\epsi^2)\,.
\]
Here $\nabla_x^{\rm Berry} := P_j\nabla_x P_j$ is the so-called Berry connection. Note that the electronic eigenvalue $E_j(x)$ appears as an effective potential for the motion of the nuclei.
To get the correct higher order terms, one needs to expand $\tilde H^\epsi_j$ on the range of $P^\epsi_j$ 
instead. While at zeroth and first order one obtains the same expansion as for $H^\epsi_j P_j$, starting at second order additional terms appear, see for example in \cite{PST2}.
Thus the unitary groups $\E^{-\I \frac{t}{\epsi}   H^\epsi_{j }} P_j $ resp.\ $\E^{-\I \frac{t}{\epsi} \tilde H^\epsi_{j }} P_j^\epsi$ can be computed by solving a Schr\"odinger equation for the nuclei only. Note that (\ref{leadingBO0}) and (\ref{leadingBO1}) are indeed the nontrivial mathematical statements to prove for justifying the time-dependent Born-Oppenheimer approximation.

\subsection[Dynamical nuclei with coupling to the field]{Dynamical nuclei with coupling to the field\\ ($\epsi,\alpha >0$)}\label{regime}

The coupling to the quantized radiation field presumably turns all electronic eigenvalues except for the ground state into resonances.
Our first aim is to prove that the Born-Oppenheimer approximation for a molecule described by $H^{\epsi, \alpha}$ remains valid. This makes sense only if the lifetime of the resonance, given according to the above discussion by $\alpha^{-3}$, is bigger than the time scale of molecular dynamics, given by $\epsi^{-1}$. 
To control the relation between the two scales we thus choose $\alpha$ to be a function of $\epsi$ such that $\alpha(\epsi)^{-3}>\epsi^{-1}$. Assuming that 
\begin{equation}\label{alphaepsi}
\alpha(\epsi) = \epsi^{\beta}, \qquad \beta>0,
\end{equation}
this condition implies that $\beta> \frac{1}{3}$. 
This is always true for realistic nuclei because
\[
m_{\rm p} \leq m_{\rm nuc} \leq 250 m_{\rm p} \quad\mbox{corresponds to} \quad   \epsi_{\rm min} = \frac{1}{680}  \leq \epsi \leq \frac{1}{43}= \epsi_{\rm max}\,,
\]
where $m_{\rm p}$ is the proton mass.
Thus
\[
\beta_{\rm min} = \frac{\ln \alpha}{\ln \epsi_{\rm min}} \approx 0.75 \quad \mbox{and}\quad  \beta_{\rm max} = \frac{\ln \alpha}{\ln \epsi_{\rm max}} \approx 1.31\,,
\]
which suggests to consider $\frac{3}{4}<\beta<\frac{4}{3}$. For some results we are able to cover even the range $\frac{2}{3}<\beta<\frac{4}{3}$, while for others we have to restrict to  $\frac{5}{6}<\beta<\frac{4}{3}$ which corresponds to $m_{\rm p} \leq m_{\rm nuc} \leq 72 m_{\rm p}$.
 
Inserting \eqref{alphaepsi} into $H^{\epsi,\alpha}$ as given in \eqref{hepsialpha} and expanding in powers of $\epsi$, we get a  Hamiltonian which depends just on $\epsi$.
Setting 
\begin{eqnarray}
H_0^\epsi &:= & \epsi^2\sum_{j=1}^{l}p_{j, x}^{2} + H_{\rm el}(x) + \Hf \nonumber\\
H_1^\epsi &:=& - \, 4\pi^{1/2} \sum_{j=1}^{r}A (\epsi^\beta y_{j})\cdot p_{j, y} \nonumber\\
H_2^\epsi &:=& - \,4\pi^{1/2}\sum_{j=1}^{l}Z_{j}A (\epsi^\beta x_{j})\cdot\epsi p_{j,x}  \nonumber \\
&&+\, \epsi^{\frac{3}{2}\beta-1}4\pi\sum_{j=1}^{r}:\hspace{-3pt}A (\epsi^\beta y_{j})^{2}\hspace{-3pt}: \;+\; \epsi^{\frac{3}{2}\beta+1}4\pi\sum_{j=1}^{l}Z_{j}^{2}:\hspace{-3pt}A (\epsi^\beta x_{j})^{2}\hspace{-3pt}: \,,\nonumber
 \end{eqnarray}
 where we normal ordered the quadratic terms,
we can write $H^\epsi := H^{\epsi,\alpha(\epsi)}$ as
\begin{equation}\label{hepsiop}
H^\epsi = \hzeroepsiop + \epsi^{\frac{3}{2}\beta}\hepsiop_{1} + \epsi^{\frac{3}{2}\beta +1}\hepsiop_{2} \,.
\end{equation}
Note that we think of $\epsi p_{j,x}$ being of order $\epsi^0$, since we want to look at states with nuclear kinetic energy of order $\epsi^0$. 
The leading order term $H_0^\epsi$ contains no coupling to the field at all. The first order term $H^\epsi_1$ describes the linear coupling of the electrons to the field and will be the relevant term for understanding spontaneous emission of  photons. Contributions from $H^\epsi_2$ will always be of lower order and contribute only to our error terms.

  Lemma~\ref{SALemma} below asserts that $H^\epsi$ is a well-defined self-adjoint operator for $\epsi$ sufficiently small and that the expansion (\ref{hepsiop}) makes actually sense, since the coefficients $H^\epsi_1$ and $H^\epsi_2$ are relatively $H_0^\epsi$-bounded with relative bounds independent of $\epsi$.

We come now to  an informal statement of our main results. 
Let 
\[
\hepsiop_{j,\rm field } := H_{j }^\epsi \otimes 1 \;+\; 1\otimes H_{\rm f}\,.
\]
In Corollary~\ref{TheoBOField} we show that, up  to a worse error estimate, the statement of  (\ref{leadingBO0}) 
remains valid.  \\[3mm]
\noindent {\bf Leading order BO-approximation with coupling to the field.} \\{\em
For $\frac{2}{3}<\beta\leq\frac{4}{3}$ it holds that
\[ 
\bigg\lVert\E^{-\I \frac{t}{\epsi} \hepsiop } -   \E^{-\I \frac{t}{\epsi} \hepsiop_{j,\rm field } } \bigg\rVert_{\mathcal{L}(D_0 ,\Hi)}
 \;=\;\Or(\epsi^{\frac{3}{2}\beta-1}\abs{t}+\epsi)\,.
\]
}\\
Technically this is a straightforward perturbative consequence of (\ref{leadingBO0}), as
the contribution of $\epsi^{\frac{3}{2}\beta}H^\epsi_1$ is of order $\epsi^{\frac{3}{2}\beta-1}$ for times of order $\epsi^{-1}$. However, we still believe that this result is conceptually important. It shows  that, in the context where the Born-Oppenheimer approximation is usually applied, the coupling to the radiation field is negligible at leading order. To our knowledge this is the first mathematical result of this type.

Our main result concerns the failure of the Born-Oppenheimer approximation because of spontaneous emission of photons. However, as we will show, the probability for making a transition through spontaneous emission is of order $\epsi^{3\beta -1}|t|$, which  for $\beta \geq 1$ is smaller than $\epsi^2$, the square of the error in the standard Born-Oppenheimer approximation (\ref{leadingBO0}).
Hence for $\beta \geq 1$ we need to consider transitions between  the superadiabatic subspaces $P^\epsi_j$  in order to correctly separate transitions through spontaneous emission from errors in the adiabatic approximation.  
Our main result is then the following, cf.\ Theorem~\ref{decayformula}. \\[3mm]
{\bf Probability for spontaneous emission.} {\em
Let $E_j(x) > E_i(x)$ for all $x$ and let $\Psi = \psi\otimes \Omega$ with $\psi  \in$ {\rm Ran}$P_{j }^\epsi$ and $\Omega\in\fock$ the vacuum state. 
The probability for ending up in the $i$-th electronic state   after time $t$ when starting in $\Psi$   is
\begin{equation}\label{simpform}
 \| (P_i^\epsi\otimes 1)\,  \E^{-\I \frac{t}{\epsi} \hepsiop }  \,\Psi \|^2 = \epsi^{ 3\beta-1}\, \int_0^{t}  \tfrac{4}{3}\,\big\||D_{ij}| \Delta_E^{3/2} \,
\E^{-\I\frac{s}{\epsi} \hepsiop_{j}}P_j \psi \big\|^2\,\D s   + o(\epsi^{ 3\beta-1})\,.
\end{equation}
 Here $\Delta_E$ and $D_{ij}$ are real-valued multiplication operators, namely  $\Delta_E(x) =  E_j(x) - E_i(x)$ the difference in energy and  $D_{ij}(x):= \sum_{k=1}^r \langle \varphi_i(x) | y_k \varphi_j(x) \rangle_{\Hi_{\rm el}}$  the dipole coupling element. $\varphi_i(x)$ and $\varphi_j(x)$ are normalized electronic states in {\rm Ran}$P_i(x)$ and {\rm Ran}$P_j(x)$ respectively.
}\\[3mm]
Thus the decay probability can be computed by propagating the initial molecular wave function according to the standard Born-Oppenheimer approximation in the level $E_j$ and integrating the decay rate along this trajectory. The decay rate is given by $\tfrac{4}{3}\,\alpha^3\epsi^{-1}|D_{ij}(x)|^2 \Delta_E(x)^3$ as a function of the nuclear configuration $x$. Recall that $\alpha=\epsi^\beta$ and that we rescaled time by $\epsi^{-1}$. This is   the natural generalization of Fermi's golden rule for atoms, c.f.\ (\ref{FGR}),
to moving nuclei. 

We briefly sketch the strategy of our proof and comment on some difficulties.
The basic idea is to use time-dependent perturbation theory according to the splitting
\[
\hepsiop = H_{j,\rm field}^\epsi  + \epsi^{\frac{3}{2}\beta} H_1^\epsi + \Or( \epsi^{\frac{3}{2}\beta+1})\,.
\]
Note that the $\Or(\epsi^2)$ term contains the $H_2^\epsi$ term, i.e.\ higher order terms in the coupling to the field, and error terms from the Born-Oppenheimer approximation, i.e.\ $P_j^\epsi \hepsiop (1-P_j^\epsi)$ and its adjoint. Abbreviating $P_{j,0}^\epsi := P^\epsi_j \otimes Q_0$, first order time-dependent perturbation theory gives at least formally
 \begin{eqnarray*} 
P_i^\epsi \E^{-\I \frac{t}{\epsi} \hepsiop } P_{j,0}^\epsi   &=& \underbrace{P_i^\epsi \E^{-\I \frac{t}{\epsi} H_{j,\rm field}^\epsi} P_{j,0}^\epsi}_{=0} -\frac{\I\epsi^{\frac{3}{2}\beta}}{\epsi} P_i^\epsi \int\limits_0^t \E^{-\I \frac{t-s}{\epsi} H_{j,\rm field}^\epsi} 
H_1^\epsi \E^{-\I \frac{s}{\epsi} H_{j,\rm field}^\epsi} P_{j,0}^\epsi\D s + \Or(\epsi^{\frac{3}{2}\beta})\\ &=& 
-\I \epsi^{\frac{3}{2}\beta -1} \underbrace{\int\limits_0^t \E^{-\I \frac{t-s}{\epsi} H_{i,\rm field}^\epsi} 
 P_i^\epsi H_1^\epsi P_{j,0}^\epsi \,\E^{-\I \frac{s}{\epsi} H_{j,\rm field}^\epsi} \D s}_{(*)} + \Or(\epsi^{\frac{3}{2}\beta})\,.
\end{eqnarray*}
This integral expression  is certainly a correct formula for the leading order piece of the wave function that made a transition after time $t$. However, since $P_i^\epsi H_1^\epsi P_{j,0}^\epsi$ is of order one, it seems at first sight to be of order $ \epsi^{\frac{3}{2}\beta -1}$, giving a transition probability of order $ \epsi^{3\beta -2}$. This is by a factor of $\epsi^{-1}$ larger than the expected value of order $\alpha^3\epsi^{-1}= \epsi^{3\beta -1}$. Thus the integral $(*)$ must be of order $\epsi^\frac{1}{2}$ due to oscillations. We don't see any way, however, to evaluate  $(*)$ directly  in order to get the simple formula (\ref{simpform}).
This is because the ``unperturbed dynamics'' given by the Born-Oppenheimer approximation is still a highly nontrivial Schr\"odinger evolution for many interacting particles.  In order to obtain a perturbative integral expression for the leading order transitions that has less oscillations and is thus tractable, we  replace $P_{j,0}^\epsi$ by  dressed superadiabatic vacuum projections $P_{j,\rm vac}^\epsi$. 
Physically speaking, this is because the leading order effect of the coupling to the field is a dressing of the electrons of order $\epsi^{\frac{3}{2}\beta}=\alpha^\frac{3}{2}$. The rate of spontaneous emission is only of order $\epsi^{\frac{3}{2}\beta +\frac{1}{2}}$. However, since the dressing does not grow as a function of time, after times of order $\epsi^{-1}$ the spontaneous emission   of order $\epsi^{\frac{3}{2}\beta -\frac{1}{2}}$ dominates the effect of the dressing. Therefore we can neglect the dressing in the final statement and it appears only in the proof.

 \section{The main results}\label{spaceadiabexp}

In this section we only give the main theorems and explain their proofs. The more technical proofs of the propositions and the lemmas are provided in Sections~\ref{Props} \& \ref{Lemmas} respectively.
In the first subsection we state a result about the time-dependent Born-Oppenheimer approximation and show that it remains valid, when we switch on the coupling to the field. In the next subsection we verify that also the superadiabatic subspaces survive in the coupled case.  Our central results will be presented in the last two subsections. There we consider the transitions between different energy levels, when there are no free photons in the beginning. In Section \ref{spontemi} we derive an expression for the leading order of the transition operator and in Section \ref{decayrate} we provide a more explicit formula for the transition rate.

\subsection{Born-Oppenheimer approximation with field}

First we consider only the molecular Hamiltonian
\[
\hepsiop_{\rm mol} := -\epsi^2 \sum_{j=1}^l \Delta_{x_j} + H_{\rm el}(x) \quad\mbox{on}\quad \Hi_{\rm mol} :=  L^2(\R_x^{3l}) \otimes L^2_{\rm as}(\R_y^{3r})
\]
with the usual domain $D_{\rm mol}:=W^{2,2}(\R^{3l+3r})\cap \Hi_{\rm mol}$. We denote the infimum of the spectrum of $\hepsiop_{\rm mol} $ by $e$ and set $\Delta_x:=\sum_{j=1}^l \Delta_{x_j}$. Furthermore, we define $D_{\rm mol}^0:=\Hi_{\rm mol}$ and denote the maximal domain of $(\hepsiop_{\rm mol})^n$ equipped with the graph norm by  $D_{\rm mol}^n$ for $n\in\N$. 

\medskip

Let $E_j(x)$ be an isolated energy band with $H_{\rm el}(x) P_j(x) = E_j(x) P_j(x)$. 
We will make use of the following version of the leading order time-dependent Born-Oppenheimer approximation.

\begin{proposition}\label{TheoBOnoField}
Let $E_j$ be an isolated energy band and $P_j$ the corresponding band projection.
The operator
\[
 \hepsiop_{j } :=   P_j \hepsiop_{\rm mol} P_j   + (1-P_j) \hepsiop_{\rm mol} (1-P_j  )
\] 
with domain $D_{\rm mol}$ is self-adjoint and it holds for $n=0,1$ that
\begin{equation}\label{leadingBO}
\left\|  \E^{-\I \frac{t}{\epsi}     \,\hepsiop_{\rm mol}}      -\E^{-\I \frac{t}{\epsi} \hepsiop_{j }}   \right\|_{\mathcal{L}(D_{\rm mol}^{n+1},D_{\rm mol}^n)} = \Or(\epsi|t|+\epsi)\,.
\end{equation}
Moreover, $P_j\in\mathcal{L}(D_{\rm mol}^n)$ with norm bounded independently of $\epsi$ for all $n\in\N$.
\end{proposition}

This is a variant of a result in \cite{SpTe}. However, in \eqref{leadingBO} we have a slower growth of the bound as a function of time and better control on the domains, which is essential for the following.  The proof  given in Section~\ref{proofBOprop} is a streamlined and improved version  of the approach  in \cite{SpTe}. 
\medskip

Now we will take the radiation field into account. Recall that
\[
\hepsiop_0  := \hepsiop_{\rm mol}\otimes 1 + 1\otimes H_{\rm f}  
\]
on $D_0  = (D_{\rm mol} \otimes \fock) \cap (\Hi_{\rm mol} \otimes D(H_{\rm f}))$ and that
\[
H^\epsi = \hzeroepsiop + \epsi^{\frac{3}{2}\beta}\hepsiop_{1} + \epsi^{\frac{3}{2}\beta +1}\hepsiop_{2}
\]
was defined in (\ref{hepsiop}).
The following lemma  shows, in particular, that  $\hepsiop_0$ and $\hepsiop$ are self-adjoint on $D_0$. 

\begin{lemma}\label{SALemma}
The free Hamiltonian
\[
\hfree := \epsi^2 \sum_{j=1}^{l}p_{j,x}^{2} + \sum_{j=1}^{r}p_{j, y}^{2} + \Hf\,,
\]
with domain
\[
D_{0}:= D(\hfree) = (W^{2,2}(\R^{3l+3r})\otimes\fock) \cap (\Hi_{\rm mol} \otimes D(H_{\rm f}))
\]
is self-adjoint. 
The potentials $\vrm{ee}$, $V_{\rm en}$ and $\vrm{nn}$ are infinitesimally $\hfree$ bounded and 
$H^\epsi_1$ and $H^\epsi_2$ are $\hfree$-bounded with relative bounds independent of $\epsi$.\\[1mm]
Hence, by Kato-Rellich, $H_0^\epsi$ and $H^\epsi$ are self-adjoint on the domain $D_0$ for $\epsi$ small enough.
\end{lemma}

As before we define the ``diagonal'' part of $H^\epsi_0$  as
\begin{eqnarray*}
 \hepsiop_{j,{\rm field}}  &:=&  (P_j \otimes 1) \,\hepsiop_0 \, (P_j \otimes 1) + ((1-P_j) \otimes 1) \,\hepsiop_0 \, ((1-P_j) \otimes 1) \\
&=& \hepsiop_{j }\otimes 1 + 1\otimes H_{\rm f}. 
\end{eqnarray*}
It is now straightforward to prove the correctness of the leading order Born-Oppenheimer approximation also with field, by treating the coupling
to the field as a small perturbation. 

\begin{corollary}\label{TheoBOField}
For $n=0,1$ and $\frac{2}{3}<\beta\leq\frac{4}{3}$ it holds that
\begin{equation}\label{AdiF}
\bigg\lVert\E^{-\I \frac{t}{\epsi} \hepsiop } -   \E^{-\I \frac{t}{\epsi} \hepsiop_{j,{\rm field}} } \bigg\rVert_{\mathcal{L}(D_0^{n+1},D_0^n)}
 \;=\;\Or(\epsi^{\frac{3}{2}\beta-1}\abs{t}+\epsi)\,
\end{equation}
and 
\begin{equation}\label{AdiFP}
\bigg\lVert\Big(\E^{-\I \frac{t}{\epsi} \hepsiop } -   \E^{-\I \frac{t}{\epsi} \hepsiop_{j,{\rm field}} } \Big) (P_j\otimes 1)\bigg\rVert_{\mathcal{L}(D_0^{n+1},D_0^n)}
 \;=\;\Or( \epsi^{\frac{3}{2}\beta-1}\abs{t} +\epsi)\,.
\end{equation}

\begin{proof}
Standard time-dependent perturbation theory yields
\[
 \left\|\E^{-\I \frac{t}{\epsi} \hepsiop } -  \E^{-\I \frac{t}{\epsi} \hepsiop_{0}} \right\|_{\mathcal{L}(D_0^{n+1} ,D_0^n )} \;=\; \Or(  \epsi^{\frac{3}{2}\beta-1}\abs{t}) \,,
\]
because $\hepsiop - \hepsiop_{0} =  \epsi^{\frac{3}{2}\beta}$ in $\mathcal{L}(D_0^{n+1},D_0^n )$.
The last statement is not completely obvious for $n=1$, but it follows from Lemma~\ref{SALemma} and the fact that the commutator 
$
\epsi^{-\frac{3}{2}\beta}[\hepsiop - H^\epsi_0, \hfree]
$ is relatively bounded by $(H_{\rm free}^\epsi)^2$ uniformly in $\epsi$.
Now (\ref{AdiF}) follows from
 \[
\left(\E^{-\I \frac{t}{\epsi} \hepsiop_0 } -   \E^{-\I \frac{t}{\epsi} \hepsiop_{j,{\rm field}} } \right)  = 
\underbrace{\Big(\E^{-\I \frac{t}{\epsi} \hepsiop_{\rm mol} } -   \E^{-\I \frac{t}{\epsi} \hepsiop_{j } }\Big)}_{=\,\Or(\epsi|t|+\epsi)  \mbox{ in }  \mathcal{L}(D_{\rm mol}^{n+1},D_{\rm mol}^n )  }  \otimes \;\E^{-\I \frac{t}{\epsi} H_{\rm f} }\,,
\]
and the fact that $1\otimes\E^{-\I\frac{t}{\epsi}H_{\rm f}}$ is uniformly bounded in $\mathcal{L}(D_0^n)$
and from using the following technical lemma.
\begin{lemma}\label{TensorLemma}
Let $m,n\in\N_0$ with $m\leq n$. There is a constant $C<\infty$ such that if $\|B\|_{\mathcal{L}(D_{\rm mol}^{j+n-m}, D_{\rm mol}^j)  }\leq\delta$ for some $\delta>0$ and all $j=0,\dots,m$, then $\|B\otimes 1\|_{\mathcal{L}(D_0^n, D_0^m)  } \leq C\delta$.
\end{lemma}
By Proposition \ref{TheoBOnoField} we have $P_j \in \mathcal{L}(D_{\rm mol}^n)$ and thus, again by the previous lemma, $P_j\otimes 1 \in \mathcal{L}(D_0^n)$. This shows also  (\ref{AdiFP}).
\end{proof}
\end{corollary}

\subsection{Superadiabatic subspaces}

At first sight one might hope that the ``non-adiabatic matrix elements''
\begin{equation}\label{trans0}
  P_i \,\E^{-\I \frac{t}{\epsi} H^\epsi} \, P_j 
\end{equation}
from the $j$th to the $i$th electronic state are, at least for $E_i<E_j$,   dominated by spontaneous emission of photons. However, this is not the case, as the main contribution to (\ref{trans0})  comes from a velocity dependent deformation of the   electronic eigenstates. More precisely, the range of $P_j(x)$ is spanned by the eigenstate for a  {\em static} nucleonic configuration. But the slow movement of the nuclei will deform the electronic states at order~$\epsi$. This deformation is visible in the naive ``non-adiabatic matrix element'' (\ref{trans0}), but does not go along with emission of a photon.
The   projection  $P_j^\epsi = P_j +\Or(\epsi)$ onto the correctly modified electronic states 
 is the so-called superadiabatic projection  associated to $P_j$.
As the following proposition shows, the superadiabatic projection $P_j^\epsi$  commutes with $\hepsiop_{\rm mol}$ up to errors of order $\epsi^3$ uniformly on subspaces of bounded total energy. It is possible to go to even higher orders or exponentially small errors of the form $\E^{-\frac{c}{\epsi}}$, but this requires some highly technical pseudo-differential calculus, cf.\ \cite{MaSo2}, and is not needed for our analysis.

\begin{proposition} \label{BOprop}
Fix an arbitrary cutoff energy $E<\infty$.  
For any isolated energy band $E_j(x)$ there is $\epsi_0>0$ such that for $\epsi<\epsi_0$ there are operators $P_j^\epsi$    with the following properties:\\[1mm]
 $P_j^\epsi$ is an orthogonal projection and for any $n\in\N_0$ there is $\epsi_n>0$ and $C_n<\infty$  such that  for $\epsi<\epsi_n$ the operators
 $P_j^\epsi$   are bounded uniformly in    $\mathcal{L}(D_{\rm mol}^n)$ and
 \begin{equation}\label{PepsinaheP}
  \left\| P^\epsi_j - P_j\right\|_{\mathcal{L}(D_{\rm mol}^n)} \leq C_n \, \epsi\,.
\end{equation}
With ${\bf 1}_{E}$ denoting the characteristic function on the interval $(-\infty,E\,]$ we have furthermore that  
\begin{equation}\label{Commu}
\left\| \left[ \hepsiop_{\rm mol} , P_j^\epsi \right] \,{\bf 1}_{E}(\hepsiop_{\rm mol}) \right\|_{\mathcal{L}(\Hi_{\rm mol}, D_{\rm mol}^n  )} \leq C_n \, \epsi^3
\end{equation}
and 
\begin{equation}\label{CommuRel}
\left\| \left[ \hepsiop_{\rm mol} , P_j^\epsi \right]   \right\|_{\mathcal{L}(D_{\rm mol}^{n+1} ,D_{\rm mol}^{n })}   \leq C_n \, \epsi \,.
\end{equation}
  Given two isolated bands $E_i$ and $E_j$ the projections $P_i^\epsi$ and $P_j^\epsi$ satisfy  
 \begin{equation}\label{Portho}
 P_j^\epsi P_i^\epsi {\bf 1}_{E+\frac{1}{2}}(\hepsiop_{\rm mol})= \Or(\epsi^2) 
 \end{equation}
 in $\mathcal{L}( \Hi_{\rm mol}, D _{\rm mol}^n )$. 
\end{proposition}

  Note that it is possible to go to even higher orders or exponentially small errors of the form $\E^{-\frac{c}{\epsi}}$, but this requires some highly technical pseudo-differential calculus, cf.\ \cite{MaSo2}, and is not needed for our analysis. On the other hand, the statements  (\ref{PepsinaheP}) and (\ref{CommuRel}) without energy cutoffs and also the 
  fact that the errors are small even in the norm of $D_{\rm mol}^n$ for $n>0$ 
  are essential to our analysis.  Since they do not follow in any straightforward way from the known results, we give the proof of Proposition~\ref{BOprop} in Section~\ref{BOpropproof}.

The next natural steps would be to improve on the error in Proposition~\ref{TheoBOnoField} for states with high energies cut off and to determine the asymptotic expansion of $P^\epsi_j \hepsiop_{\rm mol} P^\epsi_j$. However, for our purpose it suffices to consider the leading order effective Hamiltonian and we will need Proposition~\ref{TheoBOnoField} without a cutoff. Higher order results can be found e.g.\ in~\cite{MaSo1,PST2}. 
  
Now we investigate the lifted projectors $P_j^\epsi\otimes1\in\mathcal{L}(\Hi)$. Modulo some technicalities concerning the different graph norms, the following corollary is a simple consequence of Proposition~\ref{BOprop}.
\begin{corollary}\label{BOpropField}
Let  $P_j^\epsi$     be the operator  defined in Proposition~\ref{BOprop}.
Then for any $n\in\N_0$
\begin{equation}\label{Pepsibound}
\left\| \,P^\epsi_j\otimes 1 \,\right\|_{\mathcal{L}(    D_0^n  )} = \Or(1)\,,
\end{equation}
\begin{equation}\label{PepsinahePField}
\left\| \,P^\epsi_j\otimes 1  - P_j\otimes 1 \,\right\|_{\mathcal{L}(    D_0^n )} = \Or(\epsi)\,,
\end{equation}
\begin{equation}\label{CommuField}
\left\| \left[ \hepsiop_0 , P_j^\epsi \otimes 1 \right]    \, {\bf 1}_E(\hepsiop_0) \right\|_{\mathcal{L}(\Hi,  D_0^n )} = \Or(\epsi^3) \,,
\end{equation}
and for $n\geq 1$
\begin{equation}\label{CommuFieldRel}
\left\| \left[ \hepsiop_0 , P_j^\epsi \otimes 1 \right]   \right\|_{\mathcal{L}(D_0^n,    D_0^{n-1} )} =\Or(\epsi)  \,.
\end{equation}
 \end{corollary}
\begin{proof}
 In view of Lemma \ref{TensorLemma}, (\ref{Pepsibound}) and (\ref{PepsinahePField}) follow from the corresponding statements for $P_j^\epsi$, and 
 (\ref{CommuFieldRel}) follows from (\ref{CommuRel}) and
\[
 \left[ \hepsiop_0 , P_j^\epsi \otimes 1 \right]  =  \left[ \hepsiop_{\rm mol} , P_j^\epsi \right]   \otimes 1 \,.
\]
Since $H_f$ is nonnegative and $\hepsiop_{\rm mol}\otimes 1$ and $1\otimes H_{\rm f} $   commute,  we have that
\begin{equation}\label{prodest2}
({\bf 1}_E(\hepsiop_{\rm mol})\otimes 1)\,{\bf 1}_E(\hepsiop_{0}) = {\bf 1}_E(\hepsiop_{0})\,.
\end{equation}
Now (\ref{CommuField}) is a direct consequence of (\ref{Commu}) and the simple computation
\begin{eqnarray*}
  \| [\hepsiop_0, P^\epsi_j \otimes 1] \,{\bf 1}_E(\hepsiop_0)\|_{\mathcal{L}(\Hi, D^n_0)} &=& \|
[\hepsiop_{\rm mol}\otimes 1, P^\epsi_j\otimes 1 ] \,{\bf 1}_E(\hepsiop_0)\|_{\mathcal{L}(\Hi, D^n_0)} \\
&  \stackrel{(\ref{prodest2})}{=}& \| ([\hepsiop_{\rm mol}, P^\epsi_j] \otimes 1) ({\bf 1}_E(\hepsiop_{\rm mol})\otimes 1) {\bf 1}_E(\hepsiop_0)\|_{\mathcal{L}(\Hi, D^n_0)} \\
&  = & \|( [\hepsiop_{\rm mol}, P^\epsi_j]{\bf 1}_E(\hepsiop_{\rm mol})\otimes 1) {\bf 1}_E (\hepsiop_0)\|_{\mathcal{L}(\Hi, D^n_0)}\\
&   \leq& \| [\hepsiop_{\rm mol}, P^\epsi_j]{\bf 1}_E(\hepsiop_{\rm mol})\otimes 1\|_{\mathcal{L}( D^n_0)}  \,  \| {\bf 1}_E (\hepsiop_0)\|_{\mathcal{L}(\Hi, D^n_0)}.
\end{eqnarray*}
\end{proof}
In the following, we just write $P_j^\epsi$ for $P_j^\epsi\otimes 1$ to shorten notation.

\subsection{Spontaneous emission: transition operator}\label{spontemi}

According to Corollary~\ref{TheoBOField} there are no transitions between different electronic states 
at leading order even on microscopic times of order $\epsi^{-1}$. Naive perturbation theory used in the proof suggests that the leading order transitions through coupling to the field are of order $\epsi^{\frac{3}{2}\beta-1}$ on this time scale. However, we will see that the transitions through spontaneous emission are actually smaller, namely of order $\epsi^{\frac{3}{2}\beta-\frac{1}{2}}$. Since we aim at a leading order expression for these transitions, we need to control the full time evolution up to errors which are smaller than $\epsi^{\frac{3}{2}\beta-\frac{1}{2}}$. To this end we construct superadiabatic subspaces  corresponding 
to  definite dressed electronic states containing so-called virtual but no free photons. Note that the adiabatic subspaces of Corollary~\ref{BOpropField} correspond to definite   electronic states with arbitrary  state of the field.  
 
For the sake of clarity we formulate our results about the dressed projector $P^\epsi_{j,\rm vac}$ with an additional small parameter $\delta>0$. Later on we will choose a $\delta$ that is fixed by $\epsi$ and $\beta$.

\begin{proposition}\label{BOpropVacField}
Fix an arbitrary energy cutoff $E<\infty$ and let $E_j$   be  an isolated electronic energy band.  
There is $\epsi_0>0$ such that for $\epsi<\epsi_0$ and any $\delta\geq\epsi^{\frac{1}{2}}$ there are operators  $P_{j,{\rm vac}}^\epsi(\delta)\in \mathcal{L}(\Hi)\cap \mathcal{L}(D_0)$     with the following properties:\\
$P_{j,{\rm vac}}^\epsi$ are orthogonal projections with
\begin{equation}\label{CommuVacFieldRel}
 \left\| \, \left[ \hepsiop , P_{j,{\rm vac}}^\epsi \right]  \right\|_{\mathcal{L}(   D_0,\Hi)} = \Or (\epsi   )\,,  \end{equation}
\begin{equation}\label{CommuVacField}
\left\| \, \left[ \hepsiop , P_{j,{\rm vac}}^\epsi \right] \,{\bf 1}_{E+1} (\hepsiop_0 ) \,\right\|_{\mathcal{L}( \Hi, D_0)} =\Or\big(  \epsi^{\frac{3}{2}\beta   }  \delta^{ \frac{1}{2}} \big)
\,,
\end{equation}
and, as a consequence,  
\begin{eqnarray}\label{CommuChiVacField}
\left\| \left[\tilde\chi(\hepsiop ),  P_{j,{\rm vac}}^\epsi   \right] \right\|_{\mathcal{L}( \Hi, D_0)} &=& \Or\big(\epsi^{\frac{3}{2}\beta   } \delta^{ \frac{1}{2}} \big) \,, 
\end{eqnarray}
for any smooth $\tilde \chi$ with compact support in $(-\infty,E+1)$.\\
Moreover, for $n=0,1$
\begin{equation}\label{PepsinahePVacField}
  \left\| P_{j,{\rm vac}}^\epsi - P_j^\epsi\otimes Q_0 \right\|_{\mathcal{L}( D_0^n)} =\Or( \epsi^{\frac{3}{2}\beta} \delta^{-\frac{1}{2}}  )\,,
\end{equation}
where $P_j^\epsi$ is the superadiabatic projection from Proposition~\ref{BOprop} and $Q_0$ is the vacuum projection in Fock space. 
\end{proposition}

The importance of the dressed vacuum projection $P_{j,\rm vac}^\epsi$ lies in the fact that its range is invariant under the full dynamics with a smaller error than~$P_j^\epsi\otimes Q_0$. This is because it contains also the dressing of the electrons by virtual photons, which
  is crucial for computing the leading order transitions through spontaneous emission.  
The transitions from and into the subspace $P^\epsi_{j,\rm vac}$ are generated by the commutator $[P^\epsi_{j,\rm vac},\hepsiop]$, which we can compute at leading order.
\begin{proposition}\label{BOpropVacField2}
For $\epsi^{\frac{1}{2}}\leq\delta\leq1$ and $\frac{2}{3}<\beta\leq\frac{4}{3}$
\begin{equation}\label{Pvacschlangecom0}
[   P_{j,\rm vac}^\epsi , \hepsiop ] \,{\bf 1}_{E+\frac{1}{2}}(\hepsiop_0 )= \epsi^{\frac{3}{2}\beta }\delta^{\frac{1}{2}} \mathcal{T}_j    \,{\bf 1}_{E+\frac{1}{2}}(\hepsiop_0 ) + \Or\big( \epsi^{\frac{3}{2}\beta+1 }(1+\epsi\delta^{-5/2}) \big)
\end{equation}
in $\mathcal{L}(\Hi)$ with
\[
\mathcal{T}_j =    \I  \delta^\frac{1}{2} (T_j+T_j^*) + 2 \, \delta^{-\frac{1}{2}} \, \epsi (\nabla T_j+\nabla T_j^*) \cdot \epsi\nabla_x   
\]
and 
\[
T_j =  -( H_{\rm f} + H_{\rm el} - E_j+ \I \delta)^{-1} \hepsiop_1 (P_j\otimes Q_0)\,.
\]
For $n=0,1$ it holds 
\begin{equation}\label{Tjnorm}
  \|\mathcal{T}_j\|_{\mathcal{L}(D_0^{n+1},D_0^n)}=\Or(1).
 \end{equation}
\end{proposition}

We can now compute the leading order expression for the piece of the evolution that makes a transition  from 
 $P^\epsi_j\otimes Q_0$ to $P_i^\epsi$, which is the first of our two main results. 
 For the sake of brevity we consider   two non-degenerate levels $E_i(x)$ and $E_j(x)$ with associated normalized eigenfunctions $\varphi_i(x)\in$ Ran$P_i(x)$ and $\varphi_j(x)\in$ Ran$P_j(x)$. The results can be generalized to degenerate bands in a straightforward manner. 
 
 \begin{theorem}[\bf Leading order spontaneous emission]\label{theoTji}
 
Fix an arbitrary energy cutoff $E<\infty$ and let $E_j$ and $E_i$  be isolated electronic energy bands with $E_j-E_i>0$ and let  $\frac{5}{6}<\beta\leq\frac{4}{3}$. Then
 \begin{equation} \label{MainEst}
\lim_{\epsi\to 0} \left(  \epsi^{ \frac{1}{2}- \frac{3}{2}\beta} \, P_i^\epsi  \E^{-\I \frac{t}{\epsi} \hepsiop }-
  \I   \int_0^t   \E^{- \I \frac{t-s}{\epsi} \hepsiop_{i,\rm field  }}   \, \mathcal{   T}_{j\to i}   \,\, \E^{- \I \frac{s}{\epsi} \hepsiop_{j,\rm field} } \,\D s\right) (P^\epsi_j\otimes Q_0)  {\bf 1}_E(\hepsiop)  \;\;=\;\; 0\, 
   \end{equation}
 in the norm of bounded operators and uniformly on bounded time intervals. Here 
 \begin{equation}\label{Tji}
 \mathcal{T}_{j\to i}  \;\;=\;\;  \frac{\I\delta}{2\pi\epsi^\frac{1}{2}}\sum_{\lambda=1,2}\; \int\limits_{|k|<\Lambda_0} \frac{    e_\lambda(k)\cdot D_{ij}(x)\,\Delta_E(x)  }{ \sqrt{|k|}(  |k| - \Delta_E(x) +\I \delta)}   \;
 |\varphi_i(x)\rangle  \langle \varphi_j(x) |\otimes a^*(k,\lambda)   \,\D k
  \end{equation}
  with $\delta=\epsi^{\frac{1}{2}-(\beta-\frac{5}{6})/5}$,  $\Delta_E(x):=E_j(x)-E_i(x)$ and $D_{ij}(x):=\sum_{\ell=1}^r \langle \varphi_i(x) | y_\ell \varphi_j(x) \rangle_{\Hi_{\rm el}}$.
 \end{theorem}
 
 In Theorem~\ref{decayformula} we will show that the integral in (\ref{MainEst}) is   of order one (with a norm independent of the precise choice for $\delta(\epsi)$), although this is not obvious from the prefactor and the norm of $\mathcal{T}_{j\to i}$.  However, this observation justifies the interpretation of  the integral as the leading order piece of the wave function that makes a transition from level $j$ to level $i$.

 \begin{proof} ({\em of Theorem~\ref{theoTji}})
 In order to interchange energy cutoffs with unitary groups it turns out useful, that
  replacing energy cutoffs in terms of $\hepsiop_0$ by cutoffs in terms of $\hepsiop$  and vice versa only adds an error of order $\epsi^{\frac{3}{2}\beta}$. More precisely, the graph norms induced by $\hepsiop_0$ and $\hepsiop$ are equivalent and we have that 
  \begin{equation}\label{ChiDiff}
  \left\|  \tilde \chi(\hepsiop ) -  \tilde \chi(\hepsiop_0 ) \right\|_{\mathcal{L} (\Hi , D_0)}=\Or (\epsi^{\frac{3}{2}\beta}) . 
  \end{equation}
Both claims follow from the following lemma.
\begin{lemma}\label{ChiLemma}
Let $(H_0,D(H_0))$ be self-adjoint and equip $D :=  D(H_0)$ with the graph norm $\|\cdot \|_{D_0}$. 
If $A\in \mathcal{L}(D,\Hi)$ satisfies
\[
\| A\|_{\mathcal{L} ( D ,\Hi)} \leq \delta <1\,,
\]
then $H= H_0 + A$ is self-adjoint on $D(H)  = D $   and 
the   graph-norm $\|\cdot \|_{D_H}$ induced  by $H$ on $D$ is equivalent to $\|\cdot\|_{D_0}$. More precisely,  for $\psi\in D$
\[
(1-\delta)\,  \|\psi \|_{D_0}  \;\leq \; \|\psi \|_{D_H}  \;\leq\;  (1+\delta)\, \|\psi\|_{D_0} \,.
\]
Moreover,  for any $\tilde \chi \in C^\infty_0(\R)$ there is a constant $C<\infty$ (depending only on $\tilde \chi$, but not on $H$, $A$ or $\delta$) such that
\[
\left\| \tilde \chi (H_0) - \tilde \chi(H)\right\|_{\mathcal{L} ( \Hi,D)} \leq C\delta \,.
\]
\end{lemma} 

As another useful consequence of this lemma we note that the full unitary group is uniformly bounded in $\mathcal{L}(D_0)$, i.e.\ 
\begin{equation}\label{groupbound}
\big\| \E^{-\I \frac{t}{\epsi}H^\epsi} \big\|_{\mathcal{L}(D_0)} \leq C \big\| \E^{-\I \frac{t}{\epsi}H^\epsi }\big\|_{\mathcal{L}(D_{H^\epsi} )} = C\,,
\end{equation}
where $D_{H^\epsi} $ denotes $D_0$ equipped with the graph norm of $H^\epsi$.

 We now fix some $\tilde \chi\in C_0^\infty$ with $\tilde \chi {\bf 1}_E = {\bf 1}_E$ and $\tilde \chi {\bf 1}_{E+\frac{1}{2}} = \tilde \chi $ and abbreviate   $P^\epsi_{j,\rm v}:=P^\epsi_{j,\rm vac}$.
Then (using Lemma~2 from now on implicitly)
 \begin{eqnarray*}
  P_i^\epsi    P^\epsi_{j,\rm v}   \tilde\chi(\hepsiop )& \stackrel{(\ref{PepsinahePVacField})}{=} & P_i^\epsi   P_j^\epsi \otimes Q_0\,  \tilde\chi(\hepsiop ) + \Or(\epsi^{\frac{3}{2}\beta}\delta^{-\frac{1}{2}} )\\&\stackrel{(\ref{ChiDiff})}{=} & P_i^\epsi   P_j^\epsi \otimes Q_0\,  \tilde\chi(\hepsiop_0 ) + \Or\big(\epsi^{\frac{3}{2}\beta}(1+\delta^{-\frac{1}{2}}) \big)\stackrel{(\ref{Portho})}{=} \Or(\epsi^{\frac{3}{2}\beta}\delta^{-\frac{1}{2}} )
  \end{eqnarray*}
  and, hence, 
 \begin{equation}\label{hilf1}
 P_i^\epsi    P^\epsi_{j,\rm v} \E^{-\I \frac{t}{\epsi} \hepsiop  }    P^\epsi_{j,\rm v} \tilde\chi(\hepsiop )\stackrel{(\ref{CommuChiVacField})}{=} P_i^\epsi    P^\epsi_{j,\rm v}   \tilde\chi(\hepsiop )    \E^{-\I \frac{t}{\epsi} \hepsiop  }    P^\epsi_{j,\rm v} + \Or(\epsi^{\frac{3}{2}\beta}\delta^{\frac{1}{2}}) = \Or(\epsi^{\frac{3}{2}\beta}\delta^{-\frac{1}{2}} )\,.
 \end{equation}
 In the following computation we make explicit the size of the error term in the line where it first appears and collect all of them only in the final line. We use (\ref{Tjnorm}), (\ref{groupbound}) and $P_{j,\rm v}^\epsi\in\mathcal{L}(D_0)$ throughout without noting it explicitly. 
 
\begin{eqnarray}\label{EstiBig}
\lefteqn{P_i^\epsi  \E^{-\I \frac{t}{\epsi} \hepsiop } P^\epsi_j\otimes Q_0\, {\bf 1}_E(\hepsiop) \ \,\stackrel{(\ref{PepsinahePVacField})}{=}\ \, 
P_i^\epsi  \E^{-\I \frac{t}{\epsi} \hepsiop }  P^\epsi_{j,\rm v}\, \tilde\chi(\hepsiop ) {\bf 1}_E(\hepsiop) \;+\; \Or\left(\epsi^{\frac{3}{2}\beta}\delta^{-\frac{1}{2}} \right)} \\
&\stackrel{(\ref{hilf1})}{=}&  
P_i^\epsi  \left[ \E^{-\I \frac{t}{\epsi} \hepsiop  } , P^\epsi_{j,\rm v}\right]  P^\epsi_{j,\rm v} \tilde\chi(\hepsiop ) {\bf 1}_E(\hepsiop) \;+\; \Or\left(\epsi^{\frac{3}{2}\beta}\delta^{-\frac{1}{2}} \right)
\nonumber \\
 &\stackrel{(\ref{CommuChiVacField})}{=}& 
P_i^\epsi    \left[ \E^{-\I \frac{t}{\epsi} \hepsiop  } , P^\epsi_{j,\rm v }\right]\tilde\chi(\hepsiop )  P^\epsi_{j,\rm v } {\bf 1}_E(\hepsiop) + \Or\left(\epsi^{\frac{3}{2}\beta   } \delta^{ \frac{1}{2}} \right) \nonumber \\
 &=& -\tfrac{\I}{\epsi}  \int_0^t  P_i^\epsi   \E^{-\I \frac{t-s}{\epsi} \hepsiop  }  [ \hepsiop , P^\epsi_{j,\rm v }] \tilde\chi(\hepsiop )     \E^{ -\I \frac{s}{\epsi} \hepsiop  }P^\epsi_{j,\rm v } \,\D s\,{\bf 1}_E(\hepsiop) \nonumber \\
 &\stackrel{(\ref{CommuVacFieldRel}),(\ref{ChiDiff})}{=}& -\tfrac{\I}{\epsi}  \int_0^t  P_i^\epsi   \E^{-\I \frac{t-s}{\epsi} \hepsiop  }  [ \hepsiop , P^\epsi_{j,\rm v }] \tilde\chi(\hepsiop_0 )     \E^{ -\I \frac{s}{\epsi} \hepsiop  }P^\epsi_{j,\rm v } \,\D s\,{\bf 1}_E(\hepsiop) + \Or\left( \epsi^{\frac{3}{2}\beta} |t|\right)\nonumber  \\
&\stackrel{(\ref{Pvacschlangecom0})}{=}& \I \epsi^{\frac{3}{2}\beta-1}\delta^{\frac{1}{2}}   \int_0^t   P_i^\epsi\, \E^{-\I \frac{t-s}{\epsi} \hepsiop  }\, \mathcal{   T}_j \, \tilde\chi(\hepsiop_0)  \, \E^{- \I \frac{s}{\epsi} \hepsiop }P^\epsi_{j,\rm v } \,\D s\,{\bf 1}_E(\hepsiop) + \Or\left(\epsi^{\frac{3}{2}\beta}(1+\epsi\delta^{-\frac{5}{2}})|t|\right)\nonumber \\
&\stackrel{(\ref{PepsinahePField})}{=}& \I \epsi^{\frac{3}{2}\beta-1}\delta^{\frac{1}{2}}   \int_0^t   P_i\, \E^{-\I \frac{t-s}{\epsi} \hepsiop  }\, \mathcal{   T}_j \, \tilde\chi(\hepsiop_0)  \, \E^{- \I \frac{s}{\epsi} \hepsiop }P^\epsi_{j,\rm v } \,\D s\,{\bf 1}_E(\hepsiop)+ \Or\left(\epsi^{\frac{3}{2}\beta}\delta^{\frac{1}{2}}|t|\right) \nonumber \\
  &\stackrel{(\ref{AdiFP})}{=}& \I \epsi^{\frac{3}{2}\beta-1}\delta^{\frac{1}{2}}    \int_0^t  P_i\, \E^{-\I \frac{t-s}{\epsi} \hepsiop_{i,\rm field}   }\,  \mathcal{   T}_j \,\tilde\chi(\hepsiop_0 )  \,  \E^{- \I \frac{s}{\epsi} \hepsiop  }P^\epsi_{j,\rm v } \,\D s\,{\bf 1}_E(\hepsiop)\nonumber  \\
&& \quad+\, \Or\left(\epsi^{3\beta-2}\delta^{\frac{1}{2}}|t|^2+\epsi^{\frac{3}{2}\beta}\delta^{\frac{1}{2}}|t|\right)\nonumber \\
&\stackrel{(\ref{ChiDiff}),(\ref{CommuChiVacField})}{=}& \I \epsi^{\frac{3}{2}\beta-1}\delta^{\frac{1}{2}}    \int_0^t   P_i\,\E^{-\I \frac{t-s}{\epsi} \hepsiop_{i,\rm field}  }   \, \mathcal{   T}_j   \,  \E^{- \I \frac{s}{\epsi} \hepsiop  }P^\epsi_{j,\rm v }\tilde\chi(\hepsiop_0 )\,\D s\,{\bf 1}_E(\hepsiop) + \Or\left(\epsi^{3\beta-1}\delta^{\frac{1}{2}}|t|\right)\nonumber \\
  &\stackrel{(\ref{PepsinahePVacField}),(\ref{PepsinahePField})}{=}& \I \epsi^{\frac{3}{2}\beta-1}\delta^{\frac{1}{2}}   \int_0^t  P_i\, \E^{- \I \frac{t-s}{\epsi} \hepsiop_{i,\rm field}  } \, \mathcal{   T}_j   \,  \E^{- \I \frac{s}{\epsi} \hepsiop  } (P_j\otimes Q_0) \, \tilde\chi(\hepsiop_0)\,\D s\,{\bf 1}_E(\hepsiop)\nonumber \\
&& \quad+\, \Or\left( \epsi^{3\beta-1} |t|+\epsi^{\frac{3}{2}\beta}\delta^{\frac{1}{2}}|t|\right)\nonumber \\ 
&\stackrel{(\ref{AdiFP})}{=}& \I \epsi^{\frac{3}{2}\beta-1}\delta^{\frac{1}{2}}    \int_0^t  P_i\, \E^{-\I \frac{t-s}{\epsi} \hepsiop_{i,\rm field}  } \, \mathcal{   T}_j   \,  \E^{- \I \frac{s}{\epsi} \hepsiop_{j,\rm field}  }(P_j\otimes Q_0)\tilde\chi(\hepsiop_0 )\,\D s\,{\bf 1}_E(\hepsiop)\nonumber \\
&& \quad+\, \Or\left(\epsi^{3\beta-2}\delta^{\frac{1}{2}}|t|^2+\epsi^{\frac{3}{2}\beta}\delta^{\frac{1}{2}}|t|\right)\nonumber \\
  &=& \I \epsi^{\frac{3}{2}\beta-1}\delta^{\frac{1}{2}}  \int_0^t   \E^{ -\I \frac{t-s}{\epsi} \hepsiop_{i,\rm field}  } \, P_i \mathcal{   T}_j (P_j\otimes Q_0) \,     \E^{- \I \frac{s}{\epsi} \hepsiop_{j,\rm field}  }  \,\D s\,{\bf 1}_E(\hepsiop) \nonumber \\
  &&\quad+\,\Or\left(\epsi^{3\beta-2}\delta^{\frac{1}{2}}|t|^2  + \epsi^{\frac{3}{2}\beta}(1+\epsi\delta^{-\frac{5}{2}})|t| + \epsi^{\frac{3}{2}\beta}\delta^{ -\frac{1}{2}} \right)      \, .\nonumber 
  \end{eqnarray}
 For our choice of  $\delta=\epsi^{\frac{1}{2}-(\beta-\frac{5}{6})/5}$ the error term is  $o( \epsi^{ \frac{3}{2}\beta- \frac{1}{2}})$ and we can neglect it in the following.
  
By using that $\| [\epsi\nabla_x,P_j]\|=\Or(\epsi)$, $P_iP_j=0$,  $T_j^* Q_0=0$, $\nabla T_j^* Q_0=0$  and $\|\nabla T_j\| =\Or(\delta^{-\frac{3}{2}})$, we see that 
\begin{eqnarray*}
\lefteqn{
 \epsi^{-\frac{1}{2}}\delta^{\frac{1}{2}}  P_i  \mathcal{   T}_j (P_j\otimes Q_0) =}\\
 &=&    P_i  \left(  \I  \epsi^{-\frac{1}{2}} \delta   (T_j+T_j^*) + 2      \epsi^{ \frac{1}{2}}(\nabla T_j+\nabla T_j^*) \cdot \epsi\nabla_x   \right) (P_j\otimes Q_0) \\
 &=& \I  \epsi^{-\frac{1}{2}}  \delta   P_i     T_j (P_j\otimes Q_0) + 2     \epsi^{ \frac{1}{2}}   P_i    \nabla T_j  (P_j\otimes Q_0)\cdot \epsi\nabla_x     + \Or(\epsi^\frac{3}{2}\delta^{-\frac{3}{2}}) \,.
\end{eqnarray*}
Thus the corresponding replacement in the integrand contributes an error of order $\Or(  \epsi^{\frac{3}{2}\beta+\frac{1}{2}}\delta^{-\frac{3}{2}} |t|)$.
Since $\nabla_xP_i$ and $\nabla_x P_j\otimes Q_0$ are bounded independently of $\epsi$ and, since
\[
 \left\|(H_{\rm f} + H_{\rm el}(x) - E_j(x)+\I\delta)^{-1}\right\|\leq  \delta^{-1}  \,,
\]
we have
\begin{eqnarray*}\lefteqn{
\epsi^\frac{1}{2}  P_i \nabla T_j (P_j\otimes Q_0) }\\&=& -\epsi^\frac{1}{2}  P_i \nabla   \left(
(H_{\rm f} + H_{\rm el}(x) - E_j(x)+\I\delta)^{-1} H_1^\epsi (P_j\otimes Q_0)
\right) (P_j\otimes Q_0)\\
&=& -\epsi^\frac{1}{2}  \nabla   \left( P_i(x)
(H_{\rm f} + H_{\rm el}(x) - E_j(x)+\I\delta)^{-1} H_1^\epsi  
\right) (P_j\otimes Q_0) + \Or\left(\epsi^\frac{1}{2} \delta^{-1}\right)\\
&=& -\epsi^\frac{1}{2} P_i  \nabla   \left( 
(H_{\rm f} + E_i(x) - E_j(x)+\I\delta)^{-1} H_1^\epsi  
\right) (P_j\otimes Q_0) + \Or\left(\epsi^\frac{1}{2} \delta^{-1}\right)\\
&=&\epsi^\frac{1}{2} P_i   (\nabla E_i(x) - \nabla E_j(x))
(H_{\rm f} + E_i(x) - E_j(x)+\I\delta)^{-2} H_1^\epsi  
(P_j\otimes Q_0)\; +\; \Or\left(\epsi^\frac{1}{2} \delta^{-1}\right)
\end{eqnarray*}
in the norm of bounded operators. Recall that 
\[
H_1^\epsi = -  4\sqrt{\pi} \sum_{\ell=1}^{r}A (\epsi^\beta y_{\ell})\cdot p_{\ell, y} 
\]
and thus on the one particle sector of Fock space, abbreviating $\hat\rho(k) := (2\pi)^{-\frac{3}{2}}{\bf 1}_{[0,\Lambda_0]}(k)$,
\begin{eqnarray*}\lefteqn{
P_i H_1^\epsi (P_j\otimes Q_0) (x,k,\lambda) =}\\&=&-  4\sqrt{\pi}  \sum_{\ell=1}^{r} \frac{\hat\rho(k)}{\sqrt{2|k|}} |\varphi_i(x)\rangle \langle \varphi_i(x), \E^{-\I \epsi^\beta k\cdot y_\ell} e_\lambda(k)\cdot\nabla_{y_\ell}\varphi_j(x)\rangle_{\Hi_{\rm el}}\langle\varphi_j(x)|\\
&=&-  4\sqrt{\pi}  \sum_{\ell=1}^{r} \frac{\hat\rho(k)}{\sqrt{2|k|}} |\varphi_i(x)\rangle \langle \varphi_i(x), e_\lambda(k)\cdot\nabla_{y_\ell}\varphi_j(x)\rangle_{\Hi_{\rm el}}\langle\varphi_j(x)| +\Or(\epsi^\beta)\\
&=&2\sqrt{\pi}  \sum_{\ell=1}^{r} \frac{\hat\rho(k)}{\sqrt{2|k|}} |\varphi_i(x)\rangle e_\lambda(k) \cdot\langle \varphi_i(x),  [H_{\rm el}(x),y_\ell ] \varphi_j(x)\rangle_{\Hi_{\rm el}}\langle\varphi_j(x)| +\Or(\epsi^\beta)\\
&=&2\sqrt{\pi} \left(E_i(x)-E_j(x)\right) \frac{\hat\rho(k)}{\sqrt{2|k|}} e_\lambda(k) \cdot\sum_{\ell=1}^{r} \langle \varphi_i(x),  y_\ell \varphi_j(x)\rangle_{\Hi_{\rm el}}    |\varphi_i(x)\rangle\langle\varphi_j(x)| +\Or(\epsi^\beta)\\
&=:&\sqrt{2\pi} \left(E_i(x)-E_j(x)\right) \frac{\hat\rho(k)}{\sqrt{|k|}} e_\lambda(k) \cdot D_{ij}(x)    |\varphi_i(x)\rangle\langle\varphi_j(x)| +\Or(\epsi^\beta)\,.
\end{eqnarray*}
Collecting the previous observations, we showed (\ref{MainEst}) with a
  transition operator   given by
 \begin{eqnarray}
  \mathcal{\tilde T}_{j\to i}  &=& 
   \frac{\I\delta}{2\pi\epsi^\frac{1}{2}}\sum_{\lambda=1,2}\; \int\limits_{|k|<\Lambda_0} \frac{    e_\lambda(k)\cdot D_{ij}(x)\,\Delta_E(x)  }{ \sqrt{|k|}(  |k| - \Delta_E(x) +\I \delta)} \times\nonumber \\&&\qquad \left( 1 +
    \frac{2 \I \epsi \nabla_x\Delta_E(x) }{\delta\sqrt{|k|}(|k| - \Delta_E(x)+\I \delta) }\cdot (-\I\epsi\nabla_x)\right) \times\nonumber  \\ &&\qquad\qquad
 |\varphi_i(x)\rangle  \langle \varphi_j(x) |\otimes a^*(k)   \,\D k + \Or\left(\epsi^{\beta-\frac{1}{2}}+\epsi^{\frac{1}{2}}\delta^{-1  }\right) \nonumber\\
      &=:& t_1\,+\,t_2 \,+  \Or\left(\epsi^{\beta-\frac{1}{2}}+\epsi^{\frac{1}{2}}\delta^{-1  }\right)\,.\label{t1t2}
 \end{eqnarray}
 The error term is $o(1)$ and we can neglect it.
 While  the norm of $t_2$ is of order $(\frac{\epsi}{\delta^3})^\frac{1}{2}$ and thus slightly smaller than that of $t_1$, which is  of order $(\frac{\delta}{\epsi})^\frac{1}{2}$, both grow as $\epsi\to 0$. 
 In order to show that the contribution of $t_1$ to the integral in  (\ref{MainEst}) is of order one and that of $t_2$ is strictly smaller, one has to perform the time integration and use cancellations due to oscillations. 
In order not to duplicate the corresponding arguments, we skip the proof that the contribution of $t_2$ is negligible at this point and comment on it instead after the proof of Theorem~\ref{decayformula}. At that point we will have introduced the necessary machinery in order to explain the argument. 
\end{proof}

\subsection{Spontaneous emission: rate of decay}\label{decayrate}

In order to obtain 
 an explicit formula for the leading order rate of spontaneous decay from Ran$P^\epsi_j$ to Ran$P^\epsi_i$ we evaluate the norm of the leading order   wave function in Theorem~\ref{theoTji}.

\begin{theorem}[\bf Probability for spontaneous decay]\label{decayformula}
Under the same hypotheses as in Theorem~\ref{theoTji}
 for
\[
\Psi=\psi\otimes\Omega\in (P^\epsi_j\otimes Q_0)\, {\bf 1}_E(\hepsiop)\Hi
\]
it holds that
\begin{equation}\label{MR}
 \lim_{\epsi\to0}\left(\epsi^{1-3\beta}\,\| P_i^\epsi  \E^{-\I \frac{t}{\epsi} \hepsiop } \,\Psi\|^2_\Hi -  \int_0^{t}  \tfrac{4}{3}\,\big\||D_{ij}| \Delta_E^{3/2} \,
\E^{-\I\frac{s}{\epsi} \hepsiop_{j}}P_j\psi\big\|^2_{\Hi_{\rm mol}}\,\D s \right)\,=\,0 
\end{equation}
uniformly on compact time intervals and uniformly in $\psi$.
\end{theorem}
\noindent Before we come to the proof, we collect some remarks on the result:
\begin{enumerate}
\item
 Note that the subtracted term in (\ref{MR}) is of order $1$ and therefore the same is true for $\epsi^{1-3\beta}\,\| P_i^\epsi  \E^{-\I \frac{t}{\epsi} \hepsiop } \, \Psi\|^2$. 
Now recall that $\epsi^\beta$ is equal to the coupling constant $\alpha$ by choice of~$\beta$. So $\| P_i^\epsi  \E^{-\I \frac{t}{\epsi} \hepsiop }  \,\Psi\|^2$ is proportional to $\alpha^3$ and grows linearly in time. 
Observing that in our units time is scaled with~$\alpha^2$, we see that Theorem \ref{decayformula} is the generalization to molecules of the physics textbook result that for atoms the decay rate is $\tfrac{4}{3}\,\alpha^5|D_{ij}|^2 \Delta_E^3$.
\item For $\beta<1$ we may replace $P_i^\epsi$ and $P_j^\epsi$ by $P_i$ and $P_j$ respectively in the theorem because $P_k^\epsi-P_k=\Or(\epsi)$ for $k=i,j$ and $\epsi^{\frac{3}{2}\beta -\frac{1}{2}}>\epsi$. 
\item Practically the result (\ref{MR}) means that the decay rate can be computed at leading order within the Born-Oppenheimer approximation: while $\psi\in \Hi_{\rm mol}$ still contains the electronic degrees of freedom, 
$ P_j\psi$ is of the form $\phi(x) \varphi_j(x,y)$  and therefore
\[
(\E^{-\I\frac{s}{\epsi} \hepsiop_{j}}P_j\psi)(x,y) = (\E^{-\I\frac{s}{\epsi} h^\epsi_{j} } \phi)(x)\varphi_j(x,y)\,,
\]
where the effective Born-Oppenheimer Hamiltonian   is just 
\[
h^\epsi_{\rm el} = \epsi^2\left(-\I \nabla_x + \mathcal{A}_j(x)
\right)^2 + E_j(x) +\Or(\epsi^2)\,,
\]
where $\mathcal{A}_j(x) := \I\langle \varphi_j(x), \nabla\varphi_j(x)\rangle_{\Hi_{\rm el}}$ the connection coefficient of the Berry-connec\-tion $\nabla_x^j:= P_j \nabla_x P_j$.
Hence the decay rate can be written as
\[
\tfrac{4}{3}\,\big\||D_{ij}| \Delta_E^{3/2} \,
\E^{-\I\frac{s}{\epsi} \hepsiop_{j}}P_j\psi\big\|^2_{\Hi_{\rm mol}} =
\tfrac{4}{3}\,\big\||D_{ij}| \Delta_E^{3/2} \,
\E^{-\I\frac{s}{\epsi} h^\epsi_{j}}\phi\big\|^2_{\Hi_{\rm nuc}}
\]
and one only needs to solve an effective Schr\"odinger equation for the nuclei.
\end{enumerate}
 
\begin{proof}({\em of Theorem~\ref{decayformula})}
According to Theorem~\ref{theoTji}  we have that
 \[
 \lim_{\epsi\to 0} \left(  \epsi^{1-3\beta   }  \| P_i^\epsi  \E^{-\I \frac{t}{\epsi} \hepsiop } \, \Psi\|^2 - \Theta_{ij}(t)
 \right) =0
 \]
 with
  \begin{equation} \label{Thetaij0}
  \Theta_{ij}(t)  :=    \textstyle
  \left\| \int_0^t \D r  \,
  \E^{\I\frac{r }{\epsi} H^\epsi_{i,\rm f} }    
  \mathcal{T}_{j\to i} 
  \E^{-\I\frac{r }{\epsi} H^\epsi_{j,\rm f}   }\Psi \right\|^2  
  \end{equation}
  and we   need to  show that 
$
 \lim_{\epsi\to0}\left( \Theta_{ij}(t) \,-\,  \int_0^{t}  \tfrac{4}{3}\,\big\||D_{ij}| \Delta_E^{3/2} \,
\E^{-\I\frac{s}{\epsi} H^\epsi_j}\psi\big\|^2\,\D s \right)\;=\;0. 
$
First note that because of the explicit form (\ref{Tji}) of $\mathcal{T}_{j\to i}$ this state lives only  the one-particle sector of Fock space. Writing $k=\omega|k|$, the only dependence on the angular variable $\omega$  
appears in the polarization vectors $e_\lambda(\omega)$. Using 
  that for any symmetric $3\times 3$\,-matrix $A$ it holds that
\[
\textstyle \sum_{\lambda =1,2} \int_{S^2} \D \omega\, \langle e_\lambda( \omega), A e_\lambda( \omega) \rangle  =
\frac{8\pi}{3} {\rm tr}A\,,
\]
one can perform the angular integration in (\ref{Thetaij0}) and obtains a factor $\frac{8\pi}{3}$. However, to not overburden notation, we will make this explicit only later on. 
First we rewrite (\ref{Thetaij0}) as
 \begin{equation} \label{Thetaij}
  \Theta_{ij}(t)   =  \epsi^2\int_0^{\frac{t}{\epsi}} \D s\int_{-a(s)}^{a(s)}\hspace{-2mm} \D s' \underbrace{ \Big\langle \Psi,\,   \E^{\I(s+\frac{s'}{2}) H^\epsi_{j,\rm f} } \mathcal{T}_{j\to i}^* \E^{-\I s' H^\epsi_{i,\rm f} }    \mathcal{T}_{j\to i} \E^{-\I(s-\frac{s'}{2}) H^\epsi_{j,\rm f} }   
\Psi\Big \rangle }_{=: \tilde I(s,s')} 
 \end{equation}
where $s= (r+r')/(2\epsi)$, $s'=( r-r')/\epsi$, $a(s):=\min\{ 2s,2(\frac{t}{\epsi}-s)\} $. Let
\[
\tilde H_j:=   P_j  (-\epsi^2\Delta_x)  P_j + P_i  (-\epsi^2\Delta_x)  P_i + E_j 
\]
and 
\[
\tilde H_i :=   P_j  (-\epsi^2\Delta_x)  P_j + P_i  (-\epsi^2\Delta_x)  P_i + E_i  
\]
which satisfy
\[
P_j \tilde H_j   =   \tilde H_j P_j = P_j H^\epsi_j P_j\,,\quad P_i \tilde H_i   =  \tilde H_i P_i = P_i H^\epsi_i P_i
\quad\mbox{and}\quad
\tilde H_j - \tilde H_i = \Delta\,.
\]
%
Then with 
  $\mathcal{T}_{j\to i}=P_i\mathcal{T}_{j\to i}P_j$ and $\Psi=\psi\otimes\Omega$ we have that
\begin{eqnarray*}
\tilde I(s,s') &=& \Big\langle  \Psi,\,   \E^{\I(s+\frac{s'}{2}) H^\epsi_{j } } P_j \mathcal{T}_{j\to i}^* P_i  \E^{-\I s' H^\epsi_{i,\rm f} }    P_i \mathcal{T}_{j\to i}P_j  \E^{-\I(s-\frac{s'}{2}) H^\epsi_{j } }   
\Psi\Big \rangle  \\
&=& \Big\langle  \Psi,\,   \E^{\I(s+\frac{s'}{2}) \tilde H _{j } } P_j \mathcal{T}_{j\to i}^* P_i  \E^{-\I s' (\tilde H _{i } + H_{\rm f}) }    P_i \mathcal{T}_{j\to i}P_j  \E^{-\I(s-\frac{s'}{2}) \tilde H_{j } }   
\Psi\Big \rangle  \,.
\end{eqnarray*}
Making the angular integration explicit,  we can thus replace $\tilde I(s,s')$ in (\ref{Thetaij}) by 
 \begin{eqnarray*} 
\lefteqn{I(s,s') \ \,=\ \,\frac{8\pi}{3} \frac{1}{(2\pi)^2} \frac{\delta^2  }{\epsi}  \int_0^{\Lambda_0} \D R\,R \, \E^{-\I s'R}  \left\langle \psi,  \E^{\I(s+\frac{s'}{2}) \tilde H_j } \frac{D^*(x)\Delta(x)}{R-\Delta(x)-\I\delta} \E^{-\I(s+\frac{s'}{2}) \tilde H_j } \right.} \nonumber \\
&& \left. \times\, \underbrace{ \E^{\I(s+\frac{s'}{2}) \tilde H_j } \E^{-\I s' (\tilde H_j-\Delta(x)) }  \E^{-\I(s-\frac{s'}{2}) \tilde H_j }}_{=:U(s,s')}  \, \E^{\I(s-\frac{s'}{2}) \tilde H_j } \frac{D(x)\Delta(x)}{R-\Delta(x)+\I\delta}\E^{-\I(s-\frac{s'}{2}) \tilde H_j }\,\psi\right\rangle  \nonumber \\
&=&   \tfrac{2\delta^2 }{3\pi\epsi} \int_{0}^{\Lambda_0}\D R\,R\,   \E^{-\I s'R}   \left\langle\psi,  \frac{D^*(s+\tfrac{s'}{2})\Delta(s+\tfrac{s'}{2})}{R -\Delta(s+\tfrac{s'}{2})-\I\delta}  U(s,s')     \frac{D(s-\tfrac{s'}{2})\Delta(s-\tfrac{s'}{2})}{ R-\Delta(s-\tfrac{s'}{2})+\I\delta}\,\psi\right\rangle  \nonumber \,.
 \end{eqnarray*}
 Here and in the following we denote for any operator $O$   the Heisenberg operator $\E^{\I s \tilde H_{j } }O\,\E^{-\I s \tilde H_{j } }$ by~$O(s)$ and for better readability we abbreviate $R:= |k|$ and $D:= D_{ij}$.  
  Moreover,  we will still write out $\delta$ in the expressions and in most remainder estimates, but keep in mind that in the end we put $\delta=\epsi^{\frac{1}{2}-(\beta-\frac{5}{6})/5}$.

  Next we show by a stationary phase argument that $I$ is small for large~$s'$.  
  Integration by parts in
 \[ 
 I(s,s') =  \tfrac{2\delta^2 }{3\pi\epsi} \int\limits_{0}^{\Lambda_0}\D R\,  R\,\Big(\left(\tfrac{\I}{s'}\tfrac{\D}{\D R}\right)^l  \E^{-\I s'R}\Big) \,    \frac{D^*(s+\tfrac{s'}{2})\Delta(s+\tfrac{s'}{2})}{R -\Delta(s+\tfrac{s'}{2})-\I\delta}     \,U(s,s')  \,   \frac{D(s-\tfrac{s'}{2})\Delta(s-\tfrac{s'}{2})}{ R-\Delta(s-\tfrac{s'}{2})+\I\delta}, \nonumber  
\] 
 shows that
 \begin{equation}\label{error2}
 I =\Or(\delta^{ -l}\epsi^{-1}\tau^{-l}+\delta^2\epsi^{-1}\tau^{-1})\quad \mbox{ for }\quad |s'|\geq\tau\mbox{ and } \tau\geq\delta^{-1}\,.
 \end{equation}
 Instead of giving the detailed computation we just mention that the boundary terms contain the operators
 \[
 \frac{1}{  -\Delta(s\pm\tfrac{s'}{2})\pm\I\delta}  \quad\mbox{ and } \quad \frac{1}{  \Lambda_0-\Delta(s\pm\tfrac{s'}{2})\pm\I\delta} \,,
 \]
 which are all uniformly bounded, since $0 < \Delta (x) < \Lambda_0$ uniformly in $x$. So the first boundary term, which is of order  $\delta^2\epsi^{-1}\tau^{-1}$, is indeed the worst.
 
Next, for $|s'|\leq\tau$ we expand the operators around $s' =0$.
   Clearly
\begin{eqnarray}\label{DeltaExp}
 \Delta(s\pm\tfrac{s'}{2})  &=& \Delta(s)\,\pm\,\tfrac{\I}{2}\int
 _0^{s'} [H^\epsi_j,\Delta (s\pm\tfrac{s''}{2})]\,\D s''\ \,=\ \, \Delta(s)\,+\,\Or(\epsi\tau)
\end{eqnarray}
in $\mathcal{L}( D_0,\Hi)$ because the gradient of $\Delta$ is bounded independently of $\epsi$. The same is true for $\Delta$ replaced by $D$ because $\varphi_i$ and $\varphi_j$ as well as their derivatives with respect to $x$ decay exponentially in $y$ (the proof in \cite{WaTe} is easily adapted to unbounded potentials whose derivatives with respect to $x$ are bounded). 
By using the so-called Strang splitting (see \cite{JaLu}) we see that
\begin{eqnarray*}
 U(s,s')  &=& \E^{\I(s+\frac{s'}{2}) H^\epsi_j } \E^{-\I s' (H^\epsi_j-\Delta(x)) }  \E^{-\I(s-\frac{s'}{2}) H^\epsi_j }\\
&=& \E^{\I(s+\frac{s'}{2}) H^\epsi_j } \E^{-\I \frac{s'}{2} H^\epsi_j}\E^{\I s'\Delta(x) }  \E^{-\I \frac{s'}{2} H^\epsi_j} \E^{-\I(s-\frac{s'}{2}) H^\epsi_j } \,+\,\Or\big(\tau^3\|[\Delta,[H^\epsi_j,\Delta]]\|\big)\\
&=& \E^{\I s' \Delta(s)}\,+\,\Or(\epsi^2\tau^3)
\end{eqnarray*}
in $\mathcal{L}(D_0,\Hi)$. Plugging these expansions of $\Delta$, $D$, and $U$ into $I$, which is allowed because $\Psi\in D_0$ due to the energy cutoff and all operators involved are in $\mathcal{L}(D_0)$ with a norm bounded independently of $\epsi$, $\delta$, $s$ and $s'$,
we find that for $|s'|\leq \tau$
 \begin{eqnarray}\label{smalls}\nonumber
 I (s,s') &=&  \tfrac{2\delta^2 }{3\pi\epsi} \int_{0}^{\Lambda_0}\D R\,  R\,  \E^{-\I s'R}  \underbrace{\frac{D^*(s )\Delta(s )}{R -\Delta(s+\tfrac{s'}{2})-\I\delta}  \,\E^{\I s' \Delta(s)}  \,   \frac{D(s )\Delta(s )}{ R-\Delta(s-\tfrac{s'}{2})+\I\delta}}_{=: J(s,s',R,\delta)}\\
 &&+\, \Or( \tau+ \epsi\tau^3)  \,.
 \end{eqnarray} 
 While this is not small, note that the contribution of the error term to $\Theta_{ij}$
 is of order 
 \[
 \epsi^2\int_0^{\frac{t}{\epsi}} \D s\int_{-\tau}^{\tau} \D s' \Or( \tau+ \epsi\tau^3) = \Or( \epsi\tau^2+ \epsi^2\tau^4)\,,
 \] 
 which is indeed small for our choice of $\tau=\epsi^{-\frac{1}{2}+(\beta-\frac{5}{6})/10}$.
 
When expanding the denominators in (\ref{smalls}), we have to be more careful, since e.g.
\begin{eqnarray}\label{resodiff}\lefteqn{\hspace{-10mm}
\frac{1}{R -\Delta(s+\tfrac{s'}{2})-\I\delta}- \frac{1}{R -\Delta(s )-\I\delta} }\\&=&
 \frac{1}{R -\Delta(s+\tfrac{s'}{2})-\I\delta} \left(\Delta(s+\tfrac{s'}{2}) - \Delta(s)\right) \frac{1}{R -\Delta(s )-\I\delta}\nonumber
\end{eqnarray}
is only $\Or(\frac{\epsi\tau}{\delta^2})$ when naively estimating the norm, which yields a term of order 
$\Or(\frac{\epsi \tau^2}{\delta })$ to $\Theta_{ij}$. This will be large for our choice of $\delta(\epsi)$ and $\tau(\epsi)$ and thus we need a better estimate. 
For this we have to evaluate the integral explicitly. In order to prepare for the residue calculus, we first show that we can extend the $R$-integration to all of $\R$ with a  negligible error. Adding the integral to $+\infty$ yields
\begin{eqnarray}\lefteqn{ \hspace{-15mm}
\tfrac{2\delta^2 }{3\pi\epsi}
 \int_{-\tau}^{\tau} \epsi \D s'  \lim_{\rho\to\infty} \int_{\Lambda_0}^{\rho}\D R\,R\,   \E^{-\I s'R}   J(s,s',R,\delta) }\nonumber\\
 &=&  \tfrac{2\delta^2 }{3\pi\epsi}  \lim_{\rho\to\infty} \int_{-\tau}^{\tau}\epsi \D s'  \int_{\Lambda_0}^{\rho}\D R\,\left( \I \tfrac{\D}{\D s'}   \E^{-\I s'R}  \right)  J(s,s',R,\delta)\nonumber\\
  &=&-\, \tfrac{2\I\delta^2 }{3\pi\epsi} \lim_{\rho\to\infty}
 \int_{-\tau}^{\tau}\epsi \D s'   \int_{\Lambda_0}^{\rho}\D R\,  \E^{-\I s'R}    \tfrac{\D}{\D s'}  J(s,s',R,\delta)\nonumber \\&& + \, \tfrac{2\I\delta^2 }{3\pi\epsi}
 \epsi     \lim_{\rho\to\infty}   \int_{\Lambda_0}^{\rho}\D R\, \left(  \E^{-\I \tau R}       J(s,\tau,R,\delta)   -  \E^{ \I \tau R}       J(s,-\tau,R,\delta) \right)\nonumber\\
 &=&
 \Or\left( \delta^2\tau+\delta^2\right)\,,\label{error3}
 \end{eqnarray}
  since $0 < \Delta (x) < \Lambda_0$ uniformly in $x$ implies that $   J(s,s',R,\delta)$ and  $\tfrac{\D}{\D s'}  J(s,s',R,\delta)$ are uniformly bounded by $\frac{1}{R^2}$ in $\mathcal{L}(D_0,\Hi)$ for $R\geq \Lambda_0$ and analogously for $R\leq 0$. Thus we can integrate from $-\infty$ to $\infty$ in $R$ while adding an error of order $ {\delta^2}\tau $ to $\Theta_{ij}$.

Now we want to replace $J(s,s',R,\delta)$ in   (\ref{smalls}) by 
\[
\tilde J(s,s',R,\delta) := \frac{D^*(s )\Delta(s )}{R -\Delta(s )-\I\delta}  \,\E^{\I s' \Delta(s)}  \,   \frac{D(s )\Delta(s )}{ R-\Delta(s )+\I\delta}\,.
\]
One of the two terms appearing in the difference is  
\begin{equation}\label{rest}
 \tfrac{2\delta^2 }{3\pi\epsi} \lim_{\rho\to\infty}\int_{-\rho}^{\rho}\D R\,  R\,  \E^{-\I s'R}   \left(  \frac{1}{R -\Delta_+ -\I\delta}  -  \frac{1}{R -\Delta_0-\I\delta} \right) f(s)  \frac{1}{ R-\Delta_-+\I\delta} \,,
\end{equation}
where we abbreviate 
\[
\Delta_\pm := \Delta(s\pm \tfrac{s'}{2}) \,,\quad \Delta_0 := \Delta(s) \quad\mbox{and}\quad f(s ):= D^*(s)D(s)\Delta(s)^2 \E^{\I s'\Delta(s)}\,.
\]
To show that this term (and analogously the other one)  gives only a negligible contribution to $\Theta_{ij}$, 
we   use the residue calculus. For $s'<0$ we need to close the contour in the upper complex plane. Writing the spectral representation of the self-adjoint operators $\Delta_\pm$ resp.\ $\Delta_0$  with spectrum contained in $(0,\Lambda_0)$ as
\[
\Delta_{\pm,0}  =: \int_0^{\Lambda_0} \lambda \,\D P^{\pm,0}_\lambda\,,
\] 
the residue theorem yields for $ s' <0$
\begin{eqnarray*}
\lefteqn{\hspace{-0cm}
 \tfrac{2\delta^2 }{3\pi\epsi} \lim\limits_{\rho\to\infty}\int\limits_{-\rho}^{\rho}\D R\,  R\,  \E^{-\I s'R}   \left(  \frac{1}{R -\Delta_+ -\I\delta}  -  \frac{1}{R -\Delta_0-\I\delta} \right) f(s)  \frac{1}{ R-\Delta_-+\I\delta} }\\
 &=&\tfrac{2\delta^2 }{3\pi\epsi} \lim_{\rho\to\infty}  \int\limits_0^{\Lambda_0}\hspace{-1pt}\int\limits_{-\rho}^{\rho}\hspace{-1pt}\D R\,  R\,  \E^{-\I s'R}   \left(  \frac{\D P^+_\lambda}{R -\lambda -\I\delta}  -  \frac{\D P^0_{ \lambda}}{R - \lambda-\I\delta}  \right) f(s)  \frac{1}{ R-\Delta_-+\I\delta} \\
 &=&  \tfrac{2\delta^2 }{3\pi\epsi}  \int\limits_0^{\Lambda_0}  \lambda\E^{-\I s'(\lambda+\I\delta)} \left(\D P^+_\lambda -\D P^0_{ \lambda} \right) f(s) \frac{1}{\lambda-\Delta_-+2\I \delta} \;+\;\Or(\tfrac{\delta^2}{\epsi} )\;\;= \;\; (*)\,.
    \end{eqnarray*}
  According to (\ref{resodiff}) we can replace $\Delta_-$ by $\Delta_0$ in the resolvent at the price of a term of order $\frac{\delta^2}{\epsi}\cdot \frac{\epsi\tau}{\delta^2} = \tau$. Then one can commute the resolvent with $f(s)$ and afterwards replace by the same reasoning $\Delta_0$ by $\Delta_+$ for the first summand. Now one can integrate the spectral measures explicitly again and obtains 
  \begin{eqnarray}
  (*)& =& \tfrac{2\delta^2 }{3\pi\epsi}   \int\limits_0^{\Lambda_0} \lambda\E^{-\I s'(\lambda+\I\delta)} \left(\frac{\D P^+_\lambda}{\lambda-\Delta_++2\I \delta} 
 -\frac{\D P^0_{ \lambda}}{\lambda-\Delta_0+2\I \delta} 
 \right) f(s)  \;+\;\Or\big(  \tfrac{\delta^2}{\epsi} + \tau  \big) \nonumber\\
 &=& \tfrac{2\delta^2 }{3\pi\epsi}  \E^{-|s'|\delta} \left(\frac{\E^{-\I s' \Delta_+}\Delta_+}{2\I\delta} - \frac{\E^{-\I s' \Delta_0}\Delta_0}{2\I \delta}\right) f(s)  \;+\;\Or\big(  \tfrac{\delta^2}{\epsi} + \tau  \big)\nonumber\\
 &=& \Or\big(  \tau^2\delta+\tfrac{\delta^2}{\epsi} + \tau  \big) \,,\label{error4} 
       \end{eqnarray}
  where we used that by exactly the same reasoning as in (\ref{DeltaExp}) we have
  \[
 \E^{-\I s' \Delta_+}\Delta_+ -  \E^{-\I s' \Delta_0}\Delta_0 = \Or(\epsi\tau^2)
  \]     
       for $|s'|\leq\tau$. The additional factor of $\tau$ comes from the fact that derivatives of $\E^{\I s' \Delta(x)}$ are of order $|s'|\leq \tau$.
   After integration over $s$ and $s'$ this adds an error of order $ \epsi\tau^3\delta+\tau\delta^2 + \epsi\tau^2$ to $\Theta_{ij}$. 
 To estimate (\ref{rest}) for $s'>0$,   one closes the contour in the lower complex plane and proceeds along the same lines as above.
     
     Let $ a_\tau(s) := \min\{a(s),\tau\}$, then 
  collecting once more all the   estimates we obtain
\begin{eqnarray*} \lefteqn{\hspace{-1.2cm}
 \epsi^2\int_0^{\frac{t}{\epsi}} \D s\int_{-\frac{t}{\epsi}}^{\frac{t}{\epsi}} \D s'\, I (s,s')  
\; -\;  \frac{2\delta^2\epsi}{3\pi}\int_0^{\frac{t}{\epsi}} \D s\int_{-a_\tau(s)}^{a_\tau(s)} \D s'\,  \lim_{\rho\to\infty}
\int_{-\rho}^\rho \D R\, R\, \E^{-\I s'R}\,\tilde J(s,s',R,\delta) } \\
&\stackrel{(\ref{error2})}{=} & \epsi^2\int_0^{\frac{t}{\epsi}} \D s\int_{-\frac{t}{\epsi}}^{\frac{t}{\epsi}} \D s'\,\Or\left((\delta\tau)^{ -l}\epsi^{-1}+\delta^2\epsi^{-1}\tau^{-1}
\right) \\
 &\stackrel{(\ref{smalls})}{}&+\, \, \epsi^2\int_0^{\frac{t}{\epsi}} \D s\int_{-\tau}^{\tau} \D s'\,\Or\left(  \tau+ \epsi\tau^3   \right)\\ 
&\stackrel{(\ref{error3})}{}&+\,\,  \epsi \int_0^{\frac{t}{\epsi}} \D s\,\Or(\delta^2\tau + \delta^2)\\
&\stackrel{ (\ref{error4})}{}&+\, \, \epsi^2\int_0^{\frac{t}{\epsi}} \D s\int_{-\tau}^{\tau} \D s'\,\Or\left(    \tau^2\delta+ \tfrac{\delta^2}{\epsi} +  \tau  \right)\\ 
&=& 
  \Or\left((\delta\tau)^{ -l}\epsi^{-1} +\delta^2\epsi^{-1}\tau^{-1}+  \delta^2\tau + \delta^2
+ \epsi\tau^2 + \epsi^2\tau^4+ \epsi\tau^3\delta
\right)\,.
\end{eqnarray*}
Recall that $\delta=\epsi^{\frac{1}{2}-(\beta-\frac{5}{6})/5}$ and $\tau=\epsi^{-\frac{1}{2}+(\beta-\frac{5}{6})/10}$ and thus $\delta\tau = \epsi^{- (\beta-\frac{5}{6})/10}\gg 1$ for $\frac{5}{6}<\beta\leq\frac{4}{3}$. Then for $l$ big enough, the error is~$o(1)$ for all $\beta$ with $\frac{5}{6}<\beta\leq\frac{4}{3}$.

Finally, we compute the main term using again the residue calculus.  We close the integral depending on the sign of $ s'$ and get
\begin{eqnarray*}\lefteqn{ 
   \tfrac{2\delta^2\epsi}{3\pi}\int_0^{\frac{t}{\epsi}} \D s\int_{-a_\tau(s)}^{a_\tau(s)} \D s'\,  \lim_{\rho\to\infty}
\int_{-\rho}^\rho \D R\, R\, \E^{-\I s'R}\,\tilde J(s,s',R,\delta) =}\\
&=&  \tfrac{2\delta^2\epsi}{3\pi}\int_0^{\frac{t}{\epsi}} \D s\int_{-a_\tau(s)}^{a_\tau(s)} \D s'\,  \lim_{\rho\to\infty}
\int_{-\rho}^\rho \D R\, R\, \E^{-\I s'R}\,\frac{D^*(s )\Delta(s )}{R -\Delta(s )-\I\delta}  \,\E^{\I s' \Delta(s)}  \,   \frac{D(s )\Delta(s )}{ R-\Delta(s )+\I\delta}\\
&=& \tfrac{2\delta\epsi }{3}\int_0^{\frac{t}{\epsi}} \D s\int_{-a_\tau(s)}^{a_\tau(s)} \D s'\,  \E^{- |s'|\delta} \,   |D|^2(s)\Delta^3(s)   \,+\, \Or(\delta^2\tau)\\
&=& \tfrac{2\delta\epsi }{3}\int_0^{\frac{t}{\epsi}} \D s\int_{-\infty}^{\infty} \D s' \, \E^{- |s'|\delta} \,   |D|^2(s)\Delta^3(s) \,+\,\Or(\delta^2\tau+\E^{-\tau\delta} +\epsi\tau )\\
&=& \tfrac{4\epsi }{3}\int_0^{\frac{t}{\epsi}} \D s \,   |D|^2(s)\Delta^3(s)\,+\,\Or(\delta^2\tau+\E^{-\tau\delta}+\epsi\tau  )\,.
 \end{eqnarray*}
So we end up with
 \begin{eqnarray*}
 \Theta_{ij}(t) &=&  \epsi^{3\beta -1} \left\langle \Psi ,\, \epsi^2 \int_0^{\frac{t}{\epsi}} \D s\int_{-a(s)}^{a(s)} \D s' I (s,s') \,
\Psi\right\rangle\\
&=& \epsi^{3\beta}   \int_0^{t/\epsi}  \tfrac{4}{3}\left\langle \Psi,\,   |D|^2(s) \Delta^3(s) \,
\Psi\right\rangle\,\D s \,+\,o(\epsi^{3\beta -1})\\
&=& \epsi^{3\beta-1}   \int_0^{t}  \tfrac{4}{3}\,\big\||D_{ij}| \Delta_E^{3/2} \,
\E^{-\I \frac{s}{\epsi}H^\epsi_j}P_j\psi \big\|^2\,\D s \,+\,o(\epsi^{3\beta -1}).
 \end{eqnarray*}
\end{proof}

We still need to show that the contribution of $t_2$ from (\ref{t1t2}) to the transitions is negligible at leading order. More precisely, we need to show that 
\begin{eqnarray*}
 \tilde \Theta_{ij}(t) &:=&   
  \left\langle \Psi,\int_0^t \D r\int_0^t \D r'\,  \E^{\I\frac{r}{\epsi} H^\epsi_{j,\rm f} }\, t_2^*\, \E^{-\I\frac{r}{\epsi} H^\epsi_{i,\rm f}} \E^{\I\frac{r'}{\epsi} H^\epsi_{i,\rm f} }  \,  t_2\, \E^{-\I\frac{r'}{\epsi} H^\epsi_{j,\rm f} }\Psi\right\rangle\\  
&=&  \Big\langle \Psi,\,  \epsi^2\int_0^{\frac{t}{\epsi}} \D s\int_{-a(s)}^{a(s)} \D s'  \E^{\I(s+\frac{s'}{2}) H^\epsi_{j,\rm f} }\, t_2^*\, \E^{-\I s' H^\epsi_{i,\rm f} }   \, t_2\,\E^{-\I(s-\frac{s'}{2}) H^\epsi_{j,\rm f} }  
\Psi\Big\rangle,
 \end{eqnarray*}
is $o(\epsi^{3\beta-1})$.
With the same type of arguments as in the previous proof one can now show that the main contribution to 
$\tilde \Theta_{ij}(t) $ comes from the integral 
\[
\tfrac{8 \epsi^3}{3\pi}\int_0^{\frac{t}{\epsi}} \D s \left\langle   \Psi_\ell(s), \int_{-\tau}^{\tau} \D s'\,  \lim_{\rho\to\infty}
\int_{-\rho}^\rho \D R\,   \E^{-\I s'R}\,  \frac{\partial_\ell\Delta(s )D^*(s )\Delta(s )}{(R -\Delta(s )-\I\delta)^2}  \,\E^{\I s' \Delta(s)}  \,   \frac{\partial_i\Delta(s )D(s )\Delta(s )}{( R-\Delta(s )+\I\delta)^2 } \Psi_i(s) \right\rangle
\]
which is easily seen to be $\Or\left(\epsi^2 \left(\frac{\tau}{\delta^3}+ \frac{\tau^2}{\delta^2} \right)\right)$ after performing the $R$ integration. Here $\Psi_\ell(s) :=  \E^{\I s H^\epsi_{j } }(-\I\epsi \partial_{x_\ell}) \E^{-\I s H^\epsi_{j } } \Psi\in D_0$. This concludes to proof of Theorem~\ref{theoTji}.

\section{Proofs of the main propositions}\label{Props}

Before giving the details of the proofs let us shortly comment on the relation and differences 
between the Propositions~\ref{TheoBOnoField}, \ref{BOprop} and 3.  In some sense they are all ``adiabatic theorems'', however, of slightly different spirit. 
In Proposition~\ref{TheoBOnoField} we adapt and simplify arguments from \cite{SpTe}, which in turn were motivated by Kato's proof of the adiabatic theorem of quantum mechanics for Hamiltonians slowly depending on time. The basic idea is to show that the   transitions between adiabatic subspaces are small even for long times by explicitly evaluating an oscillatory integral. In Proposition~\ref{BOprop} we use the idea of superadiabatic perturbation theory: the adiabatic subspaces are replaced by slightly tilted superadiabatic subspaces. The coupling between the superadiabatic subspaces is so small that the transitions between them can be estimated even for long times by a crude norm-estimate of the integrand. 
The technical reason that forces us to include the weaker statement of Proposition~\ref{TheoBOnoField} is that it can be easily proven without energy cutoffs. This is crucial when replacing the full time-evolution by its adiabatic approximation in the computation (\ref{EstiBig}) in the proof of Theorem~\ref{theoTji}.

 An essential input for adiabatic decoupling and thus for all proofs in this section is the fact that
  the smoothness of $H_{\rm el}(x)$ and the gap assumption imply the smoothness of the map 
  $P_j: \R^{3l}\to\mathcal{L} ( \Hi_{\rm el}   )$, $x\mapsto P_j(x)$.

\begin{lemma}\label{derivPj}
Let $E_j$ be an isolated energy band and $P_j$ the corresponding band projection. Then  $P_j\in C^\infty_{\rm b}\big(\R^{3l},\mathcal{L} ( \Hi_{\rm el}  )\big)$ and $E_j\in C^\infty_{\rm b}\big(\R^{3l}\big)$. 
Moreover, for any $\alpha\in \N_0^{3l}$ and any $n\in\N_0$ one has $\partial^\alpha_x P_j(x) \in \mathcal{L}\big(\Hi_{\rm el}, D(H_{\rm el}^{n})\big)$.
\end{lemma} 

Apart from $P_j(x)$ there will appear numerous operator-valued multiplication operators of this type and in addition also differential operators $\partial^\alpha_x$ with operator-valued coefficients. The following lemma will turn out useful when working with these kind of operators.

\begin{lemma}\label{TechLem1}
Let $A_\alpha: \R^{3l} \to \mathcal{L}( \Hi_{\rm el} )$ be bounded, smooth and with bounded derivatives, i.e.\ $ A_\alpha\in C^\infty_{\rm b}\big(\R^{3l},\mathcal{L} ( \Hi_{\rm el}  )\big)$.
Then $A_\alpha=\int^\oplus A_\alpha(x)\,\D x$ defines a bounded operator on $\Hi_{\rm mol}= L^2(\R^{3l}; \Hi_{\rm el})$ and we call the differential operator
\[
A^\epsi  = \sum_{|\alpha|=0}^{n} A_\alpha \epsi^\alpha \partial^\alpha_x
\]
an admissible operator of order $n$.
\brom
\item 
An admissible operator of order $n$ is   bounded  in  $\mathcal{L} ( D_{\rm mol}^m , \Hi_{\rm mol} )$ for $m= \lceil n/2\rceil$ uniformly in $\epsi>0$. 
\item
If all coefficients $A_\alpha$ of an admissible operator $A^\epsi$ of order $n$ have the property   that $[ (\hepsiop_{\rm mol})^k , A_\alpha ]$ is   an admissible operator of order $2k-1$, then $A^\epsi$ is uniformly bounded in $\mathcal{L} ( D_{\rm mol}^{k+m}, D_{\rm mol}^{k }  )$ for $m= \lceil n/2\rceil$.
\erom
\end{lemma}

As a first simple application we note the following corollary.
\begin{corollary}\label{PCOR}
 $\partial^\beta_x P_j$ is uniformly bounded in $\mathcal{L}(D_{\rm mol}^n)$  for all $n\in\N_0$ and $\beta\in \N_0^{3l}$.
\end{corollary}
\begin{proof}
 According to Lemma~\ref{derivPj}, $\partial^\beta_x P_j$ is an admissible  operator    of order $0$ for any $\beta\in\N_0^{3l}$.  Statement (ii)  of Lemma~\ref{TechLem1}  implies the   claim of the corollary once we show    that $[ (\hepsiop_{\rm mol})^n, \partial^\beta_x P_j] $  are admissible operators of order $2n-1$ for any $n\in \N$. This in turn follows from direct computation and  the fact that according to Lemma~\ref{derivPj} we have $\partial^\alpha_x P_j(x) \in \mathcal{L}\big(\Hi_{\rm el}, D(H_{\rm el}^{n})\big)$.
\end{proof}

\subsection{Proof of Proposition~\ref{TheoBOnoField}}\label{proofBOprop}

Since $\epsi\partial_{x_i}$ has norm one in $\mathcal{L}(D_{\rm mol}^{n+1},D_{\rm mol}^n)$,   Corollary~\ref{PCOR} implies that the commutator
\begin{equation}\label{commPj}
  [\hepsiop_{\rm mol},P_j] \;=\; [-\epsi^2\Delta_x,P_j] \;=\;  -\epsi^2(\Delta_x P_j)\,-\,2\epsi\nabla_x P_j\cdot\epsi\nabla_x \;\;=\;\;\Or(\epsi)
 \end{equation}
is of order $ \epsi $ in $\mathcal{L}(D_{\rm mol}^{n+1},D_{\rm mol}^n)$ for all  $n\in\N_0$. Set $P_j^\perp(x):=1-P_j(x)$. Since 
\[
 \hepsiop_{\rm mol}\,-\,\hepsiop_{j } \;=\; P_j\hepsiop_{\rm mol}P_j^\perp\,+\,P_j^\perp\hepsiop_{\rm mol} P_j \;=\; (1-2P_j) [\hepsiop_{\rm mol},P_j],
\]
the self-adjointness of $(\hepsiop_{j },D_{\rm mol})$ for $\epsi$ small enough follows from  Lemma \ref{ChiLemma}.

We notice that $\|\E^{-\I \frac{t}{\epsi}\hepsiop_{\rm mol}}\|_{\mathcal{L}(D_{\rm mol} )}=1$ and $\|\E^{-\I \frac{t}{\epsi}\hepsiop_{j }}\|_{\mathcal{L}(D(\hepsiop_{j }))}=1$ for all $t\in\R$, when $D(\hepsiop_{j })$ is equipped with the graph norm. Then the equivalence of the graph norms due to Lemma \ref{ChiLemma} implies that $\|\E^{-\I \frac{t}{\epsi}\hepsiop_{j }}\|_{\mathcal{L}(D_{\rm mol})}$ is bounded independently of~$\epsi$.
 
Now set $R_j (x):=P_j^\perp(x)(H_{\rm el}(x)-E_j(x))^{-1}P_j^\perp(x)$ and
\begin{eqnarray*}
K_j(x)&:=& R_j (x)  \hepsiop_{\rm mol} P_j(x)\,+\,P_j(x) \hepsiop_{\rm mol} R_j (x)
\\ &=& R_j (x) [\hepsiop_{\rm mol},P_j(x)]P_j(x)\,-\,P_j(x)[\hepsiop_{\rm mol},P_j(x)]R_j (x)\,.
\end{eqnarray*}
Due to (\ref{commPj}) we have that $[ \hepsiop_{\rm mol},P_j]K_j$, $[E_j,K_j]$, and $[-\epsi^2\Delta_x,K_j]$ are of order $\epsi^2$ in $\mathcal{L}(D_{\rm mol}^{n+1},D_{\rm mol}^n)$. 
Thus it holds that
\begin{equation}\label{COMMKJ}
[\hepsiop_{\rm mol}, K_j] = [H_{\rm el}, K_j] +\Or(\epsi^2) = P^\perp_j \hepsiop_{\rm mol} P_j + P_j \hepsiop_{\rm mol} P_j^\perp +\Or(\epsi^2)
\end{equation}
and therefore
\begin{eqnarray*}
 \lefteqn{\frac{\D}{\D s}\left(\E^{-\I \frac{t-s}{\epsi} \hepsiop_{j }}\,K_j\,\E^{-\I \frac{s}{\epsi} \hepsiop_{\rm mol}}\right) }\\
&=& \tfrac{\I}{\epsi}\,\E^{-\I \frac{t-s}{\epsi} \hepsiop_{j }}\Big(\hepsiop_{j }\,K_j\,-\,K_j\,\hepsiop_{\rm mol}\Big)\E^{-\I \frac{s}{\epsi} \hepsiop_{\rm mol}}\\
&=& \tfrac{\I}{\epsi}\,\E^{-\I \frac{t-s}{\epsi} \hepsiop_{j }}\Big((\hepsiop_{\rm mol}+(2P_j-1)[\hepsiop_{\rm mol},P_j])\,K_j\,-\,K_j\,\hepsiop_{\rm mol}\Big)\E^{-\I \frac{s}{\epsi} \hepsiop_{\rm mol}}\\
&=& \tfrac{\I}{\epsi}\,\E^{-\I \frac{t-s}{\epsi} \hepsiop_{j }}\Big([\hepsiop_{\rm mol},K_j]\Big)\E^{-\I \frac{s}{\epsi} \hepsiop_{\rm mol}}\,+\,\Or(\epsi)\\
&=& \tfrac{\I}{\epsi}\,\E^{-\I \frac{t-s}{\epsi} \hepsiop_{j }}\Big(P^\perp_j \hepsiop_{\rm mol} P_j + P_j \hepsiop_{\rm mol} P_j^\perp\Big)\E^{-\I \frac{s}{\epsi} \hepsiop_{\rm mol}}\,+\,\Or(\epsi)
\end{eqnarray*}
in $\mathcal{L}(D_{\rm mol}^{n+1},D_{\rm mol}^n)$.  Hence the difference in the unitary groups is
\begin{eqnarray*} 
  \E^{-\I \frac{t}{\epsi} \hepsiop_{j }} -\E^{-\I \frac{t}{\epsi} \hepsiop_{\rm mol}}  &=&  \tfrac{\I}{\epsi} \int_0^t \E^{ -\I \frac{t-s}{\epsi} \hepsiop_{j }} \left( \hepsiop_{\rm mol}-\hepsiop_{j } \right)   \E^{ -\I \frac{s}{\epsi} \hepsiop_{\rm mol}}\,\D s
\\
&=& \tfrac{\I}{\epsi} \int_0^t \E^{ -\I \frac{t-s}{\epsi} \hepsiop_{j }} \left(P^\perp_j \hepsiop_{\rm mol} P_j + P_j \hepsiop_{\rm mol} P_j^\perp \right)    \,\E^{ -\I \frac{s}{\epsi} \hepsiop_{\rm mol}}\,\D s\\
&=& \int_0^t\frac{\D}{\D s}\left(\E^{-\I \frac{t-s}{\epsi} \hepsiop_{j }}\,K_j\,\E^{-\I \frac{s}{\epsi} \hepsiop_{\rm mol}}\right)\,\D s\,+\,\Or(\epsi|t|)\\
&=& K_j\,\E^{-\I \frac{t}{\epsi} \hepsiop_{\rm mol}}\,-\,\E^{-\I \frac{t}{\epsi} \hepsiop_{j }}\,K_j\,+\,\Or(\epsi|t|).
\end{eqnarray*}
Since $K_j$ is of order $\epsi$ in $\mathcal{L}(D_{\rm mol}^{n+1},D_{\rm mol}^n)$, we obtain (\ref{leadingBO}).

\subsection{Proof of Proposition~\ref{BOprop}}\label{BOpropproof}

This construction has been done in different places using different techniques. For the most general treatment of the Born-Oppenheimer approximation allowing even nuclei that are point charges we refer to the recent work of Martinez and Sordoni \cite{MaSo2} based on a twisted pseudo-differential calculus. Since the precise statements we need for treating the coupling to the field do not follow from their results, we give a more elementary proof for the case of smeared nuclei here. It is partly  an adaption of the arguments used in~\cite{WaTe} in a different context.

For better readability we now drop the index $j$ and write  $P_0 := P_j$ and $E_* :=E_j$. To have some margin to play with we 
use first the characteristic function ${\bf 1}_{E+1}$ on $(-\infty, E+1]$ and recall that with $e$ denoting the infimum of the spectrum of $\hepsiop_{\rm mol}$ we have that 
${\bf 1}_{E+1}(\hepsiop_{\rm mol}) = {\bf 1}_{[e,E+1]}(\hepsiop_{\rm mol})$.

Starting from the orthogonal projection $P_0$ we want to construct a self-adjoint operator  $P^\epsi\in\mathcal{L}(\Hi_{\rm mol})$ with
\[ 
P^\epsi P^\epsi = P^\epsi \quad\mbox{ and }\quad \left[ \hepsiop_{\rm mol} , P^\epsi \right] {\bf 1}_{E+1}(\hepsiop_{\rm mol}) = \Or(\epsi^3)\,.
\]
The first statement just means that $P^\epsi$ is a projection. 
 The basic
idea for constructing $P^\epsi$ is to determine first the coefficients in an asymptotic expansion
\begin{equation}\label{expansion}
P^\epsi = P_0 + \epsi P_1 +\epsi^2 P_2 + \Or(\epsi^3)\,,
\end{equation}
where we recall that  according to (\ref{commPj}) 
the commutator $[\hepsiop_{\rm mol}, P_0]$ with the choice $P_0=P_j$ is of order $\epsi$ as an operator in $\mathcal{L}(D_{\rm mol}^{n+1},D_{\rm mol}^n)$.
As shown in many instances, the requirements that $P^{(2)} := P_0 + \epsi P_1 +\epsi^2 P_2$ satisfies
\[
P^{(2)}P^{(2)}-P^{(2)} = \Or(\epsi^3) 
\quad\mbox{ and }\quad 
[ P^{(2)} , \hepsiop_{\rm mol} ] = \Or(\epsi^3) 
\]
fix  $P^{(2)}$ uniquely modulo terms of order $\epsi^3$. 
We will not repeat the construction here, but only give the result:
Let 
\[
[ P_0 ] := \tfrac{1}{\epsi}  \left[\hepsiop_{\rm mol} , P_0\right] \,,
\]
then $[P_0]$ is, according to  (\ref{commPj}),
  a uniformly bounded operator in $\mathcal{L}(D_{\rm mol}^{n+1},D_{\rm mol}^n)$.
We put 
\[
S_1 := P_0 [ P_0] R \,,
\]
 with the reduced resolvent $R(x)  =  P_0(x)^\perp ( H_{\rm el}(x) - E_*(x) )^{-1}  P_0(x)^\perp$. 
Since $\partial_x^\alpha H_{\rm el}(x)$ and thus also $\partial_x^\alpha R (x)$  are bounded operator on~$\Hi_{\rm el}$ for any $\alpha\in \N^{3l}$,
 $S_1$ is an admissible operator of order one in the sense of Lemma~\ref{TechLem1}. This is where smearing out the nuclear charge distribution is essential. By the same reasoning as in (\ref{COMMKJ}) this choice makes 
 \[
 [P_0] + [H_{\rm el}, S_1+S_1^* ] = \Or(\epsi)\,.
 \]
Now let 
\[
P_1 := S_1 + S_1^* \quad \mbox{ and }\quad P^{(1)} := P_0 + \epsi P_1\,,
\]
then 
\begin{eqnarray*}
P^{(1)}P^{(1)} - P^{(1)} &=& \epsi ( P_0 P_1 + P_1 P_0  - P_1 ) + \epsi^2 P_1P_1\\
&=& \epsi ( S_1 + S_1^* - P_1 ) + \epsi^2 (S_1S_1^* + S_1^* S_1)\;=\; \epsi^2 (S_1S_1^* + S_1^* S_1)
\end{eqnarray*}
and 
\[
[   \hepsiop_{\rm mol}, P^{(1)} ] \;=\; \epsi [P_0]  + \epsi [ H_{\rm el} , S_1 +S_1^*] + \epsi^2 [ -\epsi\Delta, P_1] =\Or(\epsi^2)
\]
 in $\mathcal{L}(D_{\rm mol}^{n+1},D_{\rm mol}^n)$.
 Now we simply iterate this construction:  first we modify $P^{(1)}$ in order to make it a projection to higher order by putting
 \[
 \tilde P^{(1)} := P^{(1)} + \epsi^2 (S_1^*   S_1 - S_1S_1^*)\,.
 \]
 This gives 
 \[
 \tilde P^{(1)} \tilde P^{(1)} - \tilde P^{(1)} =   P^{(1)}   P^{(1)} - P^{(1)}  - 2 \epsi^2    S_1S_1^*     - \epsi^2  (S_1^*   S_1 - S_1S_1^*) + \Or(\epsi^3) = \Or(\epsi^3)\,.
 \]
 Then we put
 \[
 [\tilde P^{(1)}] :={\textstyle \frac{1}{\epsi^2} }  [   \hepsiop_{\rm mol}, \tilde P^{(1)} ]
 \]
 and
 \[
 S_2 = P_0 [\tilde P^{(1)}] R \,,
 \]
 which makes 
 \[
 [\tilde P^{(1)}] + [H_{\rm el}, S_2+S_2^* ] = \Or(\epsi)\,.
 \]
 Defining 
 \[
 P_2 = S_2+S_2^* +S_1^*   S_1 - S_1S_1^*
 \]
 we find that 
 \[
[   \hepsiop_{\rm mol}, P^{(2)} ] \;=\; \epsi^2[\tilde P^{(1)}]  + \epsi^2 [ H_{\rm el} , S_2+S_2^*  ]   + \epsi^3 [ -\epsi\Delta, S_2+S_2^*] \;=:\; \epsi^3 R^\epsi_1    \,,
\]
  and still
\[
P^{(2)}P^{(2)} - P^{(2)} \;=\; \tilde P^{(1)}\tilde P^{(1)} + \epsi^2(  S_2 + S_2^*) - (\tilde P^{(1)}  + \epsi^2(  S_2 + S_2^*)) +\Or(\epsi^3) \;=:\;   \epsi^3 R^\epsi_2    \,.
\]
 Note that   $R^\epsi_1$  is an admissible operator of order three and $R^\epsi_2$ is an admissible operator of order four. 
The following lemma shows that all the operators appearing in the construction can be bounded by appropriate powers of $\hepsiop_{\rm mol}$.

\begin{lemma}\label{TechLem} The operators $P_0$, $P_1$, $P_2$, $R^\epsi_1$  and $R^\epsi_2$ are admissible  operators   of order $0$, $1$, $2$, $3$, and $4$ respectively.
Their coefficients have commutators with $(\hepsiop_{\rm mol})^k$ that are admissible operators of order $2k-1$ for any $k\in \N$. Thus they are uniformly bounded operators from $D_{\rm mol}^{n+m}$ to $D_{\rm mol}^{n}$ for any $n\in\N_0$ and $m=0$,   $1$, $2$, $3$, and $4$ respectively.
 \end{lemma}

In order to make sense of (\ref{expansion}) as a bounded operator in $\mathcal{L}(\Hi_{\rm mol})$ and to get uniform bounds on $[\hepsiop_{\rm mol}, P^\epsi]$  we thus  need to cut off 
large  energies. To do so we fix $E<\infty$ and choose $\chi_{E }\in C_0^\infty(\R,[0,1])$ 
such that $\chi_{E }|_{[e-1,E+1]}=1$ and supp$\chi_{E }\subset (e-2,E+2)$. Then we define
\[
\tilde P^\epsi := \epsi P_1 + \epsi^2 P_2
\]
and 
\[
P^\epsi_{\chi_E} := P_0 + \epsi P_1 + \epsi^2 P_2  - (1-\chi_E(\hepsiop_{\rm mol}))\, \tilde P^\epsi \, (1-\chi_E(\hepsiop_{\rm mol}))\,,
\]
i.e.\ we cut off the corrections to $P_0$ at high energies. To see that $P^\epsi_{\chi_E} $ is indeed a bounded operator in $\mathcal{L}( D _{\rm mol} ^n  )$ for any $n\in\N_0$, note that
\[
P^\epsi_{\chi_E}  = P_0  + \tilde P^\epsi  \chi_E(\hepsiop_{\rm mol}) + \chi_E(\hepsiop_{\rm mol}) \tilde P^\epsi (1- \chi_E(\hepsiop_{\rm mol}))
\]
and that $P_0$ and $\tilde P^\epsi  \chi_E(\hepsiop_{\rm mol})$ are bounded independently of $\epsi$ in $\mathcal{L}( D _{\rm mol} ^n  )$  by Lemma~\ref{TechLem} and the fact that $\chi_E(\hepsiop_{\rm mol})\in\mathcal{L}(\Hi_{\rm mol}, D _{\rm mol} ^n  )$ with norm bounded independently of $\epsi$.
In particular we have also 
\[
 \left\| P^\epsi_{\chi_E} - P_0\right\|_{\mathcal{L}( D _{\rm mol} ^n )} = \Or(\epsi)\,.
\]
We first proof that the operator $P^\epsi_{\chi_E}$ has all the properties claimed in the proposition modulo the fact that it is not a projection. In a second step we turn it into a projection without loosing the desired properties.

Now by Lemma~\ref{TechLem} it follows that 
\[
[ \hepsiop_{\rm mol} , P^\epsi_{\chi_E} ] = [ \hepsiop_{\rm mol} , P_0 ] +\Or(\epsi) = \Or(\epsi) 
\]
as a bounded operator from $D_{\rm mol} ^{n+1} $ to  $D _{\rm mol} ^{n} $.
With
$\chi_E (\hepsiop_{\rm mol})\,{\bf 1}_{E+1}(\hepsiop_{\rm mol}) ={\bf 1}_{E+1} (\hepsiop_{\rm mol})$ this implies \begin{eqnarray*}
 \left[ \hepsiop_{\rm mol} , P_{\chi_E}^\epsi \right] \,{\bf 1}_{E+1}(\hepsiop_{\rm mol})  &=&\left[ \hepsiop_{\rm mol} ,  P_0 + \epsi P_1 + \epsi^2 P_2 \right] \,{\bf 1}_{E+1}(\hepsiop_{\rm mol})\\
 &=& \epsi^3 R^\epsi_1 {\bf 1}_{E+1}(\hepsiop_{\rm mol}) =  \Or(\epsi^3) 
\end{eqnarray*}
as an operator from $\Hi_{\rm mol}$ to $D _{\rm mol} ^n $.
Note that, by taking adjoints, this implies  that  
\[
\| {\bf 1}_{E+1}(\hepsiop_{\rm mol})  \left[ \hepsiop_{\rm mol} , P_{\chi_E}^\epsi \right] \|_{\mathcal{L}(\Hi_{\rm mol})} = \Or(\epsi^3)
\]
and with $\| {\bf 1}_{E+1}(\hepsiop_{\rm mol}) \|_{\mathcal{L}(\Hi_{\rm mol}, D_{\rm mol}^n)} = \Or(1)$ also
\[
\| {\bf 1}_{E+1}(\hepsiop_{\rm mol})  \left[ \hepsiop_{\rm mol} , P_{\chi_E}^\epsi \right] \|_{\mathcal{L}(\Hi_{\rm mol}, D_{\rm mol}^n)} = \Or(\epsi^3)\,.
\]
For later use we also show that this implies the smallness of the commutator of $P_{\chi_E}^\epsi$ with a smooth energy cutoff $\tilde\chi$ supported in $(e-\frac{3}{4},E+\frac{3}{4})$
\begin{equation}\label{Pchicommu}
\left\|  \left[ \tilde \chi(\hepsiop_{\rm mol})  , P_{\chi_E}^\epsi \right] \right\|_{\mathcal{L}(\Hi_{\rm mol}, D _{\rm mol}^n )}= \Or(\epsi^3)\,.
\end{equation}
 Since the argument will be used several times in the remainder of the paper, we formulate it as a lemma.

\begin{lemma}\label{CommuLemma}
Let  $I\subset \R$ be a compact interval, $\tilde I\subset I$ another interval with different endpoints and $\tilde \chi\in C_0^\infty(\R)$ with supp$\tilde\chi\subset \tilde I$. Then for any $n\in\N_0$ there exists $C<\infty$ depending only on   $n$  and $\tilde \chi$ with the following property:
Let  $(H,D(H))$ be self-adjoint and $A\in \mathcal{L}(\Hi)$ be  bounded and self-adjoint. Then
 \[
 \| [H, A] {\bf 1}_I(H)  \|_{\mathcal{L}(\Hi, D(H^n))} \leq \delta 
 \]
       implies that
\[
\left\| \left[ \tilde \chi(H), A\right]\right\|_{\mathcal{L}(\Hi, D(H^n))} \leq C\delta \,.
\]
\end{lemma}

Now we need to turn the ``almost projection'' $P^\epsi_{\chi_E}$ into a true projection.
Since we will use this trick as well several times, we formulate it again as a lemma.

\begin{lemma}\label{ProjectorLemma}
There are constants $C_n<\infty$, $n\in\N$, such that the following holds:\\
Let $(H , D(H ) )$ be a   self-adjoint operator  and let $D^n:= D(H^n)$ be equipped with the~norm
\[
\| \psi \|_{D^n} := \sum_{i=0}^n \| H^i \psi \| \,.
\]
For some $N\in\N$  let $\tilde Q$ be an operator that is bounded in  $\mathcal{L}(D^n)$  for all $0\leq n\leq N$ and   self-adjoint in $\mathcal{L}(\Hi)$   with the following properties:
\begin{equation}\label{Property4}
 \| \tilde Q  \tilde Q  - \tilde Q   \|_{\mathcal{L}( D^n   )} \leq \delta 
\end{equation}
for all $0 \leq n\leq N$ and some $\delta <\frac{1}{4}$ and
\begin{equation}\label{Property0}
 \| [ H  ,  \tilde Q ]     \|_{\mathcal{L}( D^{n} , D^{n-1}   )} \leq \delta_{n } 
\end{equation}
for all $1\leq n\leq N$ and some   $\delta_{n }<  \frac{1}{2} \frac{1}{2^{2n } }$.
  
 Then there is an orthogonal projection $Q\in\mathcal{L}(\Hi)$ 
 with $\|Q\|_{\mathcal{L}(D^n)}\leq 4^{n+1}$ that satisfies
 \begin{equation}\label{Prop6}
\| Q  - \tilde Q  \| _{\mathcal{L}(  \Hi   )} \leq  \delta  \quad \mbox{and}\quad \| Q  - \tilde Q  \| _{\mathcal{L}(  D^n   )} \leq C_n  \delta   
\end{equation}
and  
   \[ 
 \| [ H  ,    Q ]     \|_{\mathcal{L}( D^{n} ,D^{n-1}  )} \leq C_n\, \delta_{n }
\]
for all $n\leq N$.

Moreover, there is a constant $C_E$ depending only on $E\in \R$ such that we have the following implications: 
\begin{equation}\label{Property3}
 \| [ H  ,  \tilde Q ]   \,{\bf 1}_{E+1}(H )   \|_{\mathcal{L}( \Hi , D^n   )} \leq \beta_1  
\end{equation}
for all $n\leq N $ implies 
\[
 \| [ H  ,    Q ]   \,{\bf 1}_{E+\frac{1}{2}}(H )   \|_{\mathcal{L}( \Hi , D^n )} \leq C_EC_n\,\beta_1
\]
for all $n\leq N $, and 
\begin{equation}\label{Property5}
 \|  (\tilde Q  \tilde Q  - \tilde Q )   \,{\bf 1}_{E+\frac{1}{2}}(H )   \|_{\mathcal{L}( \Hi , D^n   )} \leq   \beta_2
\end{equation} 
for all $n\leq N $ implies
\begin{equation}\label{Property7}
  \| (Q  - \tilde Q )\,{\bf 1}_{E+\frac{1}{2}}(H )  \| _{\mathcal{L}( \Hi, D^n  )} \leq C_EC_n\, \beta_2 
\end{equation}
for all $n\leq N $.
\end{lemma}

We can now apply Lemma~\ref{ProjectorLemma} to the almost projection $P^\epsi_{\chi_E}$ almost commuting with  $\hepsiop_{\rm mol}$, where now 
$\delta$ and $\delta_n$ are of order $\epsi$ and $\beta_1$ of order $\epsi^3$. 
To be able to use also the last implication of Lemma~\ref{ProjectorLemma} with $\beta_2$ of order $\epsi^3$, we still need to show  (\ref{Property5}). 
To this end observe that we have by construction that
\[ 
\|   ( P^{(2)} P^{(2)} - P^{(2)}  ) \,{\bf 1}_{E +\frac{1}{2}}(\hepsiop_{\rm mol}) \|_{\mathcal{L}(\Hi_{\rm mol},D_{\rm mol} ^n )}  =\Or(\epsi^3)\,.
\]
Hence for $\tilde\chi\in C^\infty_0(\R)$ with $\tilde \chi {\bf 1}_{E+\frac{1}{2} } = {\bf 1}_{E+\frac{1}{2} }$ and supp$\tilde \chi
\subset (e-\frac{3}{4},E+\frac{3}{4})$ we have
\begin{eqnarray*}
 ( ( P^\epsi_{\chi_E})^2 - P^\epsi_{\chi_E} ) \,{\bf 1}_{E +\frac{1}{2} }(\hepsiop_{\rm mol}) &=&  ( ( P^\epsi_{\chi_E})^2 - P^\epsi_{\chi_E} ) \tilde\chi(\hepsiop_{\rm mol}) {\bf 1}_{E +\frac{1}{2}}(\hepsiop_{\rm mol})\\
 &\stackrel{(\ref{Pchicommu})}{=}& \tilde\chi(\hepsiop_{\rm mol}) ( ( P^\epsi_{\chi_E})^2 - P^\epsi_{\chi_E} )  {\bf 1}_{E +\frac{1}{2}}(\hepsiop_{\rm mol}) + \Or(\epsi^3)\\
 &=& \tilde\chi(\hepsiop_{\rm mol}) ( ( P^{(2)})^2 - P^{(2)}  )  {\bf 1}_{E +\frac{1}{2} }(\hepsiop_{\rm mol}) + \Or(\epsi^3)\\
 &=&\Or(\epsi^3)\,.
\end{eqnarray*} 
Thus we can use Lemma~\ref{ProjectorLemma} to turn $P^\epsi_{\chi_E}$ into an orthogonal projection $P^\epsi$ with the desired properties.

Next we show that 
 \[ 
 P^{(1)}_j P^{(1)}_i = \Or(\epsi^2)\quad\mbox{for }\, i\not=j\,.
 \]
 To enhance readability we denote $P^{(1)}_j = P^{(1)}$ and $P^{(1)}_i = Q^{(1)}$ etc., i.e.\ we distinguish the different electronic levels by the letters $P$ and $Q$ instead of the indices $j$ and $i$. Then with $Q_0P_0 =0$ 
 \[
 Q^{(1)} P^{(1)} =  (Q_0 + \epsi Q_1   ) (P_0 + \epsi P_1   )= \epsi ( Q_1 P_0 + Q_0 P_1 ) +  \Or(\epsi^2)\,.
 \]
Denoting the $S_1$-operator associated to $Q^{(1)}$ by $R_1$ we find
  \begin{eqnarray*}
   Q_1 P_0 + Q_0 P_1    &=& R_1 P_0 + Q_0 S_1^* \\
   &=& Q_0 [Q_0] R(E_Q) P_0 - Q_0 R(E_P) [P_0]P_0\\
   &=& (E_P-E_Q)^{-1} ( Q_0 [Q_0]   P_0 +  Q_0   [P_0]P_0 ) = 0\,.
   \end{eqnarray*}
 This implies that
\begin{eqnarray*}
P^\epsi Q^\epsi {\bf 1}_{E+\frac{1}{2}}(\hepsiop_{\rm mol}) & = & P^\epsi Q^\epsi  \tilde \chi(\hepsiop_{\rm mol})  {\bf 1}_{E+\frac{1}{2}}(\hepsiop_{\rm mol}) \nonumber \\
&=&  \tilde \chi(\hepsiop_{\rm mol}) P^\epsi Q^\epsi   {\bf 1}_{E+\frac{1}{2}}(\hepsiop_{\rm mol}) +\Or(\epsi^3) \nonumber\\
&=&  \tilde \chi(\hepsiop_{\rm mol}) P^{(1)}  Q^{(1)}   {\bf 1}_{E+\frac{1}{2}}(\hepsiop_{\rm mol}) +\Or(\epsi^2) \nonumber\\
&=& \Or(\epsi^2) 
\end{eqnarray*}
in $\mathcal{L}( \Hi_{\rm mol}, D _{\rm mol} ^n )$.

\subsection{Proof of Propositions~\ref{BOpropVacField}~\&~\ref{BOpropVacField2}}\label{BOpropVacFieldProof}

We first recall the perturbative form of $\hepsiop$ from (\ref{hepsiop}): 
\[
\hepsiop =: \hepsiop_0 + \epsi^{\frac{3}{2}\beta } \hepsiop_1 + \epsi^{\frac{3}{2}\beta+1 } \hepsiop_2\,.
\]
The operators $\hepsiop_1$ and $\hepsiop_2$ 
satisfy 
\[
\| \hepsiop_i \|_{\mathcal{L}(D(\hepsiop_0),\Hi)}\leq C_i
\]
with constants $C_i$   independent of $\epsi$. 
Hence Lemma~\ref{ChiLemma} yields that for $\epsi$ small enough 
$\hepsiop$ is self-adjoint on $D(\hepsiop_0)$ and the graph norms induced by $\hepsiop_0$ and $\hepsiop$ are uniformly equivalent.

We write $P_{\rm vac}^\epsi := P_{j,{\rm vac}}^\epsi$ and $P_{\rm vac}  := P_j \otimes Q_0$, where $Q_0$ is the projection onto the vacuum state in~$\fock$.
As before we first construct an almost projection $\tilde P_{\rm vac}^\epsi$ with the desired properties and then apply Lemma~\ref{ProjectorLemma}. Since the first correction to $\hepsiop_0$ is of order $\epsi^{\frac{3}{2}\beta}$, it is  natural to make the ansatz
\[
\tilde P_{\rm vac}^\epsi  := P^\epsi \otimes Q_0 + \epsi^{\frac{3}{2}\beta} P_{\frac{3}{2}\beta}   \,,
\]
where we assume that $P^\epsi$ is constructed as in Proposition~\ref{BOprop} but with energy cut off at $E+1$. 
Computing the commutator with $\hepsiop$, we find that
\begin{eqnarray*}
\lefteqn{[ \tilde P_{\rm vac}^\epsi, \hepsiop ] \,{\bf 1}_{E+1}(\hepsiop_0)\;\;=} \\ &\stackrel{(\ref{Assu1})}{=}&  [ P^\epsi \otimes Q_0 + \epsi^{\frac{3}{2}\beta} P_{\frac{3}{2}\beta}, \hepsiop_0 + \epsi^{\frac{3}{2}\beta} \hepsiop_1 ] \,{\bf 1}_{E+1}(\hepsiop_0) +\,\Or\big(\epsi^{\frac{3}{2}\beta+1 } \big)\\
&=& 
 [ P^\epsi \otimes Q_0  , \hepsiop_0   ] \,{\bf 1}_{E+1}(\hepsiop_0) \,+ \, \epsi^{3 \beta} \,   [ P_{\frac{3}{2}\beta}, \hepsiop_1   ]     \,{\bf 1}_{E+1}(\hepsiop_0)\\
 && + \, \epsi^{\frac{3}{2}\beta} \,\left(   [ P^\epsi \otimes Q_0  ,   \hepsiop_1 ] + [P_{\frac{3}{2}\beta}, \hepsiop_0 ]    \right) \,{\bf 1}_{E+1}(\hepsiop_0)
  \, +\, \Or\big(\epsi^{\frac{3}{2}\beta+1 } \big)\\
  &\stackrel{(\ref{Assu3}) - (\ref{Assu4})}{=}& \epsi^{\frac{3}{2}\beta} \,\left(   [ P_{\rm vac}  ,   \hepsiop_1 ] + [P_{\frac{3}{2}\beta}, \hepsiop_0 ]    \right) \,{\bf 1}_{E+1}(\hepsiop_0)\, + \,\Or\big(\epsi^{\frac{3}{2}\beta+1}\big) 
\end{eqnarray*}
in $\mathcal{L}(\Hi )$.
In this computation we made the following assumptions, which are clear on a formal level but need to be proved later on:
\begin{equation}\label{Assu1}
[ \tilde P_{\rm vac}^\epsi, \hepsiop_2 ] \,{\bf 1}_{E+1}(\hepsiop_0) =\Or(1)
\end{equation}
\begin{equation}\label{Assu3}
\epsi^{3 \beta} \,   [ P_{\frac{3}{2}\beta}, \hepsiop_1   ]     \,{\bf 1}_{E+1}(\hepsiop_0) = \Or\big(\epsi^{\frac{3}{2}\beta+1} \big)
\end{equation}
\begin{equation}\label{Assu2}
[ P^\epsi \otimes Q_0  , \hepsiop_0   ] \,{\bf 1}_{E+1}(\hepsiop_0) =   \Or(\epsi^{\frac{3}{2}\beta+1})
\end{equation}
\begin{equation}\label{Assu4}
[ (P^\epsi-P_j)\otimes Q_0, \hepsiop_1 ] {\bf 1}_{E+1}(\hepsiop_0) = \Or(\epsi)\,,
\end{equation}
all in $\mathcal{L}(\Hi )$. 
Whether (\ref{Assu1}) and (\ref{Assu3}) are satisfied depends on $P_{\frac{3}{2}\beta}$, which we now construct by the requirement that the commutator is of order $\epsi^{\frac{3}{2}\beta}\delta^\frac{1}{2}$, i.e.\ that
\[
\left(   [ P_{\rm vac}  ,   \hepsiop_1 ] + [P_{\frac{3}{2}\beta}, \hepsiop_0 ]    \right) \,{\bf 1}_{E+1}(\hepsiop_0) = \Or(\delta^\frac{1}{2})\,.
\]
Dropping the energy cutoff for a moment this translates to  
\begin{equation}\label{Commu32}
    [P_{\frac{3}{2}\beta}, \hepsiop_0 ]     =- [ P_{\rm vac}  ,   \hepsiop_1 ] +\Or(\delta^\frac{1}{2})\,.
\end{equation}
To solve this equation for $P_{\frac{3}{2}\beta}$ we cannot proceed  as in adiabatic theory with spectral gap,   since the reduced resolvent $( H_{\rm el} + H_{\rm f} - E_j)^{-1} (1- P_{\rm vac})$ is not bounded without spectral gap. Therefore we proceed as in \cite{Teu1} and shift the resolvent  into the complex plane by a small amount $\delta$,
\begin{eqnarray*}
P_{\frac{3}{2}\beta}^\delta &:=&   -( H_{\rm f} + H_{\rm el} - E_j+ \I \delta)^{-1} \hepsiop_1 P_{\rm vac}  - P_{\rm vac} \hepsiop_1 ( H_{\rm f} + H_{\rm el} - E_j- \I \delta)^{-1}\\
&=:& T_\delta + T_\delta^* \,.
\end{eqnarray*}
Note that  $P_{\frac{3}{2}\beta}^{\delta=0}$ is exactly the first order correction  one would obtain by formally applying standard perturbation theory to the electronic eigenprojection $P_j(x)\otimes Q_0$. 
With this definition  we find that
\begin{eqnarray}\label{commucomp}
    \lefteqn{
    [P_{\frac{3}{2}\beta}^\delta, H_{\rm f} + H_{\rm el}]} \\
    &=&    \hepsiop_1 P_{\rm vac} \;-\;  P_{\rm vac} \hepsiop_1 \;+\;   (E_j-\I\delta) ( H_{\rm f} + H_{\rm el} - E_j+ \I \delta)^{-1}\hepsiop_1P_{\rm vac} \nonumber \\
    & & \;-\;   P_{\rm vac} \hepsiop_1 ( H_{\rm f} + H_{\rm el} - E_j- \I \delta)^{-1}(E_j+\I\delta)
    \nonumber \\
    & & \;-\;   ( H_{\rm f} + H_{\rm el} - E_j+ \I \delta)^{-1} \hepsiop_1 P_{\rm vac} E_j \;+\;     E_jP_{\rm vac} \hepsiop_1 ( H_{\rm f} + H_{\rm el} - E_j- \I \delta)^{-1}  \nonumber \\
    &=& \;-\;[P_{\rm vac} ,\hepsiop_1]  + \I\delta \left( T_\delta + T_\delta^*\right)\nonumber   \,.
\end{eqnarray}
We will show that 
\[
 \I\delta \left( T_\delta + T_\delta^*\right) = \Or( \delta^{1/2})\quad\mbox{and}\quad [T_\delta, \epsi\nabla_x] = \Or \left(\frac{\epsi}{\delta^{3/2}}\right)\,,
\]
which indeed  gives us  (\ref{Commu32}) for $\delta\geq\epsi^{\frac{1}{2}}$. Note for the following that $P_{\frac{3}{2}\beta}^\delta$ is again a fibered operator,
\[
T_\delta(x)  = - ( H_{\rm f}(x) + H_{\rm el} - E_j(x) + \I \delta)^{-1}  \hepsiop_1   P_{\rm vac}(x)       \,.
\]
This is important, since we will need to commute $P_{\frac{3}{2}\beta}^\delta$ through $\hepsiop_{\rm mol}$ and thus to compute derivatives of $P_{\frac{3}{2}\beta}^\delta(x)$ with respect to $x$.
\begin{lemma} For $\delta>0$ small enough we have 
that $\R^{3l}\to \mathcal{L}(\Hi)$, $x\mapsto T_\delta(x)$ is smooth and there is a constant $C<\infty$ not depending on $\delta$ or $\epsi$ such that  
\begin{equation}\label{No1}
\|T_\delta\| 
\leq   \frac{C}{\sqrt{\delta}}\,,
\end{equation}
 \begin{equation}\label{No3}
\| ( H_{\rm f} + H_{\rm el}) T_\delta \|
\leq   \frac{C}{\sqrt{\delta}}  
\end{equation}
and for $|\alpha|\in \N_0^{3l}$
\begin{equation}\label{No2}
\|    \partial_x^\alpha T_\delta \| 
\leq C\,\left(\frac{1}{\delta}
\right)^{|\alpha|+\frac{1}{2}}   \quad \mbox{ and } \quad \| ( H_{\rm f} + H_{\rm el})   \partial_x^\alpha T_\delta \| 
\leq C\,\left(\frac{1}{\delta}
\right)^{|\alpha|+\frac{1}{2}}  \,.
\end{equation}

\begin{proof} Let $ (\varphi_1 (x), \ldots,\varphi_s(x)) $ be an orthonormal basis of    Ran$P_j(x)$ and write
$\Psi\in$ Ran$P_j\otimes Q_0$ as 
\[
\Psi(x,y) =   \sum_{m=1}^s \psi_m (x) \,\varphi_m (x,y)\,.
\]
 Then
\[
\hepsiop_1  P_{\rm vac}\Psi 
=     \I \sum_{i=1}^r A(\epsi^\beta y_i) \cdot \nabla_{y_i}  \Psi 
=\I\sum_{m=1}^s \sum_{i=1}^r      \sum_{\lambda=1}^2 \psi_m (x) \frac{e_\lambda(k)}{\sqrt{2|k|}}\,\hat\rho(k)  \E^{-\I \epsi^\beta k\cdot y_i} \nabla_{y_i}\varphi_m (x,y)\,.
\]
Since the sum is finite, it suffices to estimate the   resolvent acting on each summand.
We split   
\[
1 = {\bf 1}_{[e,\infty)}(H_{\rm el}) =   {\bf 1}_{[e,E_*]}(H_{\rm el})   + {\bf 1}_{(E_*,\infty)}(H_{\rm el})= : P_\leq   +  P_>
\]
and observe that on the range of $P_>$ the resolvent is indeed uniformly bounded also for $\delta=0$ because of the gap condition. 
So it remains to look at the resolvent acting on the range of 
\[
P_\leq := \sum_{\ell=1}^{j } P_\ell \,. 
\]
Using $H_{\rm el}P_\ell = E_\ell P_\ell$   we get that
\begin{eqnarray*}\lefteqn{ 
\left\| ( |k|+ H_{\rm el}(x))  ( |k|+ H_{\rm el}(x) - E_j(x)+ \I \delta)^{-1}  \frac{\hat\rho(k) }{\sqrt{2|k|}}\,  P_\leq(x) \E^{-\I \epsi^\beta k\cdot y_i} \nabla_{y_i}\varphi_m (x,y)\right\|^2 }\\
&=& \left\| \sum_{\ell=1}^{j } ( |k| + E_\ell (x) ) ( E_\ell (x) - E_j(x)+ |k| + \I \delta)^{-1}  \frac{  \hat\rho(k) }{\sqrt{2|k|}}\,  P_\ell(x) \E^{-\I \epsi^\beta k\cdot y_i} \nabla_{y_i}\varphi_m (x,y)\right\|^2\\
&\leq &
C \sum_{\ell=1}^{j }  \int_0^{\Lambda_0} \frac{( E_\ell (x) + |k| )^2 |k| }{ \big(  |k|-(E_j(x)-E_\ell(x)) \big)^2 +  \delta^2}  \, \D |k| \;\leq\; \frac{C}{\delta}\,.
\end{eqnarray*} 
This shows (\ref{No1}) and (\ref{No3}).

 To get the bounds for the derivatives first observe that
whenever a derivative hits a resolvent, we get 
 \begin{eqnarray*} \lefteqn{
 \partial_{x_j}  \left( |k| + H_{\rm el}(x) - E_j(x)+ \I \delta \right )^{-1}  }\\
&=& \left( |k| + H_{\rm el}(x)  - E_j(x)+ \I \delta \right )^{-1}
  \partial_{x_j} (E_j(x)-H_{\rm el}(x))
\left( |k| + H_{\rm el}(x)  - E_j(x)+ \I \delta \right )^{-1}
\end{eqnarray*}
 where $  \partial_{x_j}( E_j(x)-H_{\rm el}(x))$ is uniformly bounded. 
By Lemma \ref{derivPj} derivatives of $P_j$ map into the domain of $H_{\rm el}$ and thus into the domain of $\hepsiop_1$.
Hence,  whenever at least one  derivative hits $P_j$   there will be at most $|\alpha|$ resolvents left and such a term can be estimated by $\delta^{-|\alpha|}$.
When all the derivatives hit the resolvent, the worst term has $|\alpha|+1$ resolvents, which can be estimated by $\delta^{-|\alpha|}$ times the norm of $T_\delta$.
  \end{proof}
 \end{lemma}

\begin{corollary}
Let $\epsi^{\frac{1}{2}}\leq\delta\leq1$. With 
\[
P_{\frac{3}{2}\beta}^\delta \; := \;  T_\delta + T_\delta^*
\]
we have that for $n=0,1$
\begin{equation}\label{P32norm1}
  \| P_{\frac{3}{2}\beta}^\delta \|_{\mathcal{L}(D( (\hepsiop_0)^n))} = \Or(\delta^{-\frac{1}{2}}),  
\end{equation}
and 
\begin{equation}\label{Pvacclose}
 \| \tilde P_{\rm vac}^\epsi  - P_j^\epsi \otimes Q_0 \|_{\mathcal{L}(D( (\hepsiop_0)^n))} =  \Or\big(\epsi^{\frac{3}{2}\beta }\delta^{-\frac{1}{2}} \big),
\end{equation}
\begin{equation}\label{Pvacsquare}
 \| \tilde P_{\rm vac}^\epsi \tilde P_{\rm vac}^\epsi - \tilde P_{\rm vac}^\epsi \|_{\mathcal{L}(D( (\hepsiop_0)^n))} =  \Or\big(\epsi^{3\beta}\delta^{-1}\big)\,.
\end{equation}
Moreover,  
\begin{equation}\label{P32commuRel}
    \left[P_{\frac{3}{2}\epsi}^\delta, \hepsiop_0  \right]  + [P_0\otimes Q_0, \hepsiop_1]   =  \delta^{\frac{1}{2}}\mathcal{T} + R,  
\end{equation}
where
\begin{equation*}
\mathcal{T} =     \I \delta^{\frac{1}{2}} (T_\delta+T_\delta^*) + 2\, \delta^{-\frac{1}{2}}\epsi  \nabla_x (T_\delta+T_\delta^*) \cdot \epsi\nabla_x   = \Or(1)
\end{equation*}
in $\mathcal{L}\big(D( (\hepsiop_0)^{n+1}),D( (\hepsiop_0)^n)\big)$ and
\begin{equation*}
\| R \|_{\mathcal{L}(D( (\hepsiop_0)^n))} = \Or(\epsi^2\delta^{-\frac{5}{2}} )\,.
\end{equation*}
 
\begin{proof} 
We will use $\delta\geq\epsi^{\frac{1}{2}}$ without noting it explicitly. It follows directly from (\ref{No1}) that $\|T_\delta\|=\|T_\delta^*\|=\Or(\delta^{-\frac{1}{2}})$, which yields~(\ref{P32norm1}) for $n=0$.
For $n=1$ notice that (\ref{No3}), (\ref{No2}), and $T_\delta(H_{\rm el} + H_{\rm f}) = T_\delta E_j$ imply that
\[
[ \hepsiop_0 , T_\delta ] = [ -\epsi^2\Delta_x + H_{\rm el} + H_{\rm f} ,T_\delta ] = [ -\epsi^2\Delta_x ,T_\delta ]- \hepsiop_1 P_{\rm vac}+\I\delta T_\delta
\]
is $\Or(1)$ in $\mathcal{L}(D(\hepsiop_0),\Hi)$ and thus by Lemma~\ref{NormDnLemma} also $T_\delta=\Or(\delta^{-\frac{1}{2}})$ in $\mathcal{L}(D(\hepsiop_0) )$.
In a similar way we find that also $T^*_\delta$ is $\Or(\delta^{-\frac{1}{2}} )$    in $\mathcal{L}(D(\hepsiop_0) )$.
 The estimate (\ref{Pvacclose}) follows immediately from (\ref{P32norm1}). For (\ref{Pvacsquare}) we note that
\begin{equation}\label{offdiagonal}
    P_{\rm vac}T_\delta \;=\; 0 \;=\; T_\delta^*P_{\rm vac}
\end{equation}
because $[P_{\rm vac}, H_{\rm el} + H_{\rm f}]=0$ and $\hepsiop_1$ creates a photon when applied to $P_{\rm vac}$.  
Then we have
  \begin{eqnarray*}\lefteqn{
  (P^\epsi \otimes Q_0 + \epsi^{\frac{3}{2}\beta} P^\delta_{\frac{3}{2}}\beta)(P^\epsi \otimes Q_0 + \epsi^{\frac{3}{2}\beta} P^\delta_{\frac{3}{2}}\beta)\;\;=}\\
  &=& P^\epsi \otimes Q_0 + \epsi^{\frac{3}{2}\beta}  \left( P^\delta_{\frac{3}{2}\beta}P^\epsi \otimes Q_0 + P^\epsi \otimes Q_0P^\delta_{\frac{3}{2}\beta}\right) +   \epsi^{3\beta} P^\delta_{\frac{3}{2}\beta}P^\delta_{\frac{3}{2}\beta}\\
  &=&  P^\epsi \otimes Q_0 + \epsi^{\frac{3}{2}\beta}   \left( P^\delta_{\frac{3}{2}\beta}P_0 \otimes Q_0 + P_0 \otimes Q_0P^\delta_{\frac{3}{2}\beta}\right) +  \Or\big(\epsi^{3\beta}\delta^{-1}\big) \\
  &\stackrel{(\ref{offdiagonal})}{=}& P^\epsi \otimes Q_0 + \epsi^{\frac{3}{2}\beta} P^\delta_{\frac{3}{2}\beta}  
  +  \Or\big(\epsi^{3\beta}\delta^{-1}\big)\,.
    \end{eqnarray*}
For    (\ref{P32commuRel}) note that according to (\ref{commucomp})
\[
 \left[P_{\frac{3}{2}\beta}^\delta, H_{\rm el} + H_{\rm f}  \right]  + [P_0\otimes Q_0, \hepsiop_1]   =  \I\delta (T+T^*)\,.
\]
With
\[
 \left[P_{\frac{3}{2}\beta}^\delta, -\epsi^2\Delta_x  \right]    =   2 \epsi (\nabla_x P_{\frac{3}{2}\epsi}^\delta) \cdot \epsi\nabla_x + \epsi^2 (\Delta_x P_{\frac{3}{2}\beta}^\delta)  
 \]
and $\|R\|_{\mathcal{L}(\Hi)} = \| \epsi^2 (\Delta_x P_{\frac{3}{2}\beta}^\delta) \|_{\mathcal{L}(\Hi)} = \Or(\epsi^2\delta^{-\frac{5}{2}} )$ we directly obtain (\ref{P32commuRel}) for $n=0$.
For $n=1$ it suffices to show $\|[\mathcal{   T} ,\hepsiop_0]\|_{\mathcal{L}(D(\hepsiop_0),\Hi)} = \Or(1)$. It holds 
\[
[\mathcal{   T} ,\hepsiop_0]=\I\delta^{\frac{1}{2}}[T_\delta+T_\delta^* ,\hepsiop_0]+\delta^{-\frac{1}{2}}\epsi[\nabla_x(T_\delta+T_\delta^*) ,\hepsiop_0]. 
\]
As shown in the proof of (\ref{P32norm1}) the first term is of order $\delta^{\frac{1}{2}}$. Analogously, it follows that the second term is of order $\epsi\delta^{-2}$. Hence, both are $\Or(1)$ because of $\epsi^{\frac{1}{2}}\leq\delta\leq1$.
Finally $\|R\|_{\mathcal{L}(D( \hepsiop_0)  )} = \| \epsi^2 (\Delta_x P_{\frac{3}{2}\beta}^\delta) \|_{\mathcal{L}(D( \hepsiop_0 ))} = \Or(\epsi^2\delta^{-\frac{5}{2}} + \epsi^2\delta^{-\frac{7}{2}}+ \epsi^4\delta^{-\frac{9}{2}} )$ follows from Lemma~\ref{NormDnLemma} and (\ref{No2}).
 \end{proof}
\end{corollary}

\begin{lemma} It holds that
\begin{equation}\label{PvacschlangecomRel}
\| [ \hepsiop ,  \tilde P_{\rm vac}^\epsi ]   \|_{\mathcal{L}(D(\hepsiop_0), \Hi)} = \Or\big(\epsi + \epsi^{\frac{3}{2}\beta  }\delta^{\frac{1}{2}}\big)
\end{equation}
and
\begin{equation}\label{Pvacschlangecom}
[ \tilde P_{\rm vac}^\epsi , \hepsiop ] \,{\bf 1}_{E+\frac{1}{2}}(\hepsiop_0 )= \epsi^{\frac{3}{2}\beta  } \delta^{\frac{1}{2}}\mathcal{T}    \,{\bf 1}_{E+\frac{1}{2}}(\hepsiop_0 ) + \Or( \epsi^{\frac{3}{2}\beta+2  } \delta^{-\frac{5}{2}}  ) = \Or\big(  \epsi^{\frac{3}{2}\beta   } \delta^{ \frac{1}{2}}\big)
\end{equation}
in $\mathcal{L}(\Hi)$.
As a consequence,
\begin{equation}\label{PvacschlangecomChi}
\left\| \left[\tilde\chi(\hepsiop ),  \tilde P_{{\rm vac}}^\epsi   \right] \right\|_{\mathcal{L}(\Hi)} = \Or(  \epsi^{\frac{3}{2}\beta   } \delta^{ \frac{1}{2}}  )
\end{equation}
for any smooth $\tilde \chi$ with compact support in $(-\infty,E+\frac{1}{2})$.

\begin{proof} 
We first show (\ref{Pvacschlangecom}). Due to (\ref{P32commuRel}) we only need to check (\ref{Assu1})--(\ref{Assu4}). The estimates (\ref{Assu2}) and (\ref{Assu4}) follow from (\ref{CommuField}) and (\ref{PepsinahePField}) respectively. Since $\epsi^{\frac{3}{2}\beta} \delta^{- \frac{1}{2}}<1$ for $\beta>5/6$ and $\delta>\epsi^{1/2}$, $\tilde P_{\rm vac}^\epsi$ is $\Or(1)$ in $\mathcal{L}(\Hi)$ as well as in $\mathcal{L}(D(\hepsiop_0))$ by (\ref{P32norm1}). So (\ref{Assu1}) follows from the fact that $\hepsiop_2$ is uniformly bounded from $D(\hepsiop_0)$ to $\Hi$.
Since $\epsi^{3\beta} \delta^{- \frac{1}{2}}<\epsi^{\frac{3}{2}\beta+1}$ for $\beta>5/6$ and $\delta>\epsi^{1/2}$, (\ref{Assu3}) also follows from (\ref{P32norm1}). Now (\ref{PvacschlangecomChi}) directly follows from Lemma~\ref{ChiLemma}.
   
    For  (\ref{PvacschlangecomRel}) we apply exactly the same reasoning as for (\ref{Pvacschlangecom}), however, with (\ref{CommuFieldRel}) instead of (\ref{CommuField}), which worsens the bound.
 \end{proof}
\end{lemma}

Now Lemma~\ref{ProjectorLemma} applied to $\tilde P^\epsi_{\rm vac}$ with  $n=1$, $\delta  \sim \epsi^{3\beta}\delta^{-1}$, $\delta_1\sim\epsi+\epsi^{\frac{3}{2}\beta   } \delta^{ \frac{1}{2}} $ and $\beta_1\sim \epsi^{\frac{3}{2}\beta   } \delta^{ \frac{1}{2}}$  
yields a projector $P^\epsi_{\rm vac}$   with all the properties claimed in  Propositions~\ref{BOpropVacField}~and~\ref{BOpropVacField2}.

\section{Proofs of Lemmas}\label{Lemmas}

\subsection*{Proof of Lemma~\ref{SALemma}}
The statement about the   potentials is standard. 
Using for examples the estimates contained in Proposition~$1$ and in the proof of Proposition~$2$ of \cite{Ten} we can show easily that
\begin{equation*}
\norm{A ^{(i)}(\alpha y_{j})\cdot p_{j, y, (i)}}_{\mathcal{L}(D_{0}, \Hi)} = \norm{\Phi(v_{\alpha y}^{(i)})\cdot p_{j, y, (i)}}_{\mathcal{L}(D_{0}, \Hi)}\leq C \norm{v_{\alpha y}^{(i)}}_{\omega},
\end{equation*}
where 
\begin{equation*}\begin{split}
v_{\alpha y}^{(i)}(k, \lambda) := \frac{\hat{\varphi}(\mu k)}{\abs{k}^{1/2}}e_{\lambda}^{(i)}(k)\E^{\I k\cdot\alpha y}
\end{split}
\end{equation*}
and, given a function $f\in L^{2}(\field{R}{3}\otimes\field{C}{2})$, 
\begin{equation*}
\norm{f}_{\omega}:= \big(\norm{f\abs{k}^{-1/2}}^{2}_{L^{2}(\field{R}{3}\otimes\field{C}{2})} + \norm{f}^{2}_{L^{2}(\field{R}{3}\otimes\field{C}{2})}\big)^{1/2} \, .
\end{equation*}
Using these explicit expressions we get then
\[
\norm{A ^{(i)}(\alpha y_{j})\cdot p_{j, y, (i)}}_{\mathcal{L}(D_{0}, \Hi)}\leq C \Lambda\mu^{-1} = C\frac{\Lambda}{2\me\alpha^{2}} =   C\Lambda_{0}\,.
\]
In the same way we have
\begin{equation*}
\norm{:A (\alpha y_{j})^{2}:}^{2}_{\mathcal{L}(D_{0}, \Hi)}\leq C\norm{v_{\alpha y}^{(i)}}_{\omega}^{2}\leq C (\Lambda\mu^{-1})^{2} = C\Lambda_{0}^{2}\,.
\end{equation*}
Identical results hold for the coefficients of the Hamiltonian containing the nuclear coordinates, so all the coefficients in $\hepsiop$ can be bounded with an $\epsi$-independent bound in terms of $\hzeroepsiop$ or $\hfree$.

\subsection*{Proof of Lemma~\ref{TensorLemma}}
 
Since $H_f$ is nonnegative and since $\hepsiop_{\rm mol}\otimes 1$ and $1\otimes H_{\rm f} $   commute,  we have that
\begin{equation}\label{prodest1}
(\hepsiop_{\rm mol} \otimes 1)^n \leq (\hepsiop_{\rm mol} \otimes 1  + 1\otimes H_{\rm f})^n = (\hepsiop_0)^n
\end{equation}
for any $n\in \N$. To estimate tensor product operators in $\mathcal{L}(D_0^n)$, the following characterization of this operator norm will be useful.

\begin{lemma} \label{GraphNormLemma}
Let $(H,D(H))$ be self-adjoint and 
\[
D^n := \{\psi \in D(H)\,|\, H^{k}\psi \in D(H)\mbox{ for } k=1,\ldots n-1\}
\]
be equipped with the graph norm
\[
\| \psi \|_{D^n} := \sum_{j=0}^n \| H^j \psi \|\,.
\]
Then $(D^n, \| \cdot\|_{D^n})$ is a Banach space, 
\[
D^n_R := \{ (H+\I)^{-n} \psi\,|\,\psi\in \Hi\} = D^n
\]
and for $A\in \mathcal{L}(D^n,D^m)$ the operator norm $\| A \|_{\mathcal{L}(D^n,D^m)}$ is equivalent to the norm
\[
\| A\|_{R(n,m)} := \sum_{j=0}^m \| H^j A (H+\I)^{-n}\|_{\mathcal{L}(\Hi)}\,.
\]
More precisely, there are constants $C_n$ depending only on $n$ (not on $H$ or $A$), such that
\[
{\textstyle \frac{1}{1+m}} \|A\|_{R(n,m)}\;\leq\; \|A\|_{\mathcal{L}(D^n,D^m)} \;\leq\; C_n \|A\|_{R(n,m)}\,.
\]
\begin{proof}
Since $H$ is self-adjoint, it is closed and therefore  $(D^n, \| \cdot\|_{D^n})$ is a Banach space.
Let $\psi\in D^n$, then $(H+\I)^n \psi\in \Hi$ and thus $\psi\in D^n_R$. Let conversely $\psi = (H+\I)^{-n}\phi \in D^n_R$, then
$H^k\psi\in \Hi$ for $k\leq n$ since $H(H+\I)^{-1} \in \mathcal{L}(\Hi)$.
For the norms observe that
\begin{eqnarray*}
\| A\psi \|_{D^m} &=& \sum_{j=0}^m \| H^j A\psi\| \;=\;  \sum_{j=0}^m \| H^j A  (H+\I)^{-n} (H+\I)^{n} \psi\|\\
&\leq& \sum_{j=0}^m \| H^j A  (H+\I)^{-n}\|_{\mathcal{L}(\Hi)} \| (H+\I)^{n} \psi\| \;\;\leq \;\; C_n \|A\|_{R(n,m)} \|\psi\|_{D^n}
\end{eqnarray*}  
and thus
\[
\|A\|_{\mathcal{L}(D^n,D^m)}\leq C_n \|A\|_{R(n,m)}\,.
\]
Conversely for $j\leq m$
\begin{eqnarray*}
\|H^j A (H+\I)^{-n} \psi \| &\leq& \|A(H+\I)^{-n}\psi \|_{D^m} \leq \|A\|_{\mathcal{L}(D^n,D^m)}  \| (H+\I)^{-n}\psi \|_{D^n}\\
&\leq &  \|A\|_{\mathcal{L}(D^n,D^m)} \|\psi\|\,,
\end{eqnarray*}  
where we use $\|(H+\I)^{-1}\|_{\mathcal{L}(D^j, D^{j+1})} = 1$. Thus
\[
\|A\|_{R(n,m)} \leq (m+1) \|A\|_{\mathcal{L}(D^n,D^m)}\,.
\]
\end{proof}
\end{lemma}

So the $\mathcal{L}( D^n_0 , D^m_0  )$-norm of an operator $B\otimes 1$ for $m\leq n$ is estimated by
\[
\| B\otimes 1 \|_{\mathcal{L}( D^n_0,D^m_0 )} \leq C_n \sum_{j=0}^m \| (\hepsiop_0)^j \,(B\otimes 1) \, (\hepsiop_0+\I)^{-n}\|\,.
\]
This will turn out useful, since for bounded operators on Hilbert spaces $\|A\otimes B\|= \|A\|\cdot \|B\|$.
Let's look at a single  term in the sum more closely,
\begin{eqnarray*}\lefteqn{\hspace{-.5cm}
(\hepsiop_0)^j \,(B\otimes 1) \, (\hepsiop_0+\I)^{-n}  \;=\; (\hepsiop_{\rm mol}\otimes 1 + 1\otimes H_{\rm f} )^j \,(B\otimes 1) \, (\hepsiop_0+\I)^{-n}}\\
&=& \sum_{\ell = 0}^j  {\textstyle { j \choose \ell} }\left( (\hepsiop_{\rm mol})^\ell \otimes 1\right) \left( 1 \otimes H_f^{j-\ell}\right) (B\otimes 1) \, (\hepsiop_0+\I)^{-n}\\
&=& \sum_{\ell = 0}^j  {\textstyle { j \choose \ell} }\left( (\hepsiop_{\rm mol})^\ell B \otimes H_f^{j-\ell}\right)  \left(  (\hepsiop_{\rm mol}+\I)^{-(\ell+n-j)}\otimes (H_{\rm f} +\I)^{-(j-\ell)}    \right)  \,\times\\
&& \hspace{3cm}\times  \left(  (\hepsiop_{\rm mol}+\I)^{\ell+n-j}\otimes (H_{\rm f} +\I)^{j-\ell}    \right)   (\hepsiop_0+\I)^{-n}\\
&=&\sum_{\ell = 0}^j  {\textstyle { j \choose \ell} }\left( (\hepsiop_{\rm mol})^\ell B (\hepsiop_{\rm mol}+\I)^{-(\ell+n-j)} \otimes H_f^{j-\ell}(H_{\rm f} +\I)^{-(j-\ell)}\right)   \,\times\\
&& \hspace{3cm}\times  \left(  (\hepsiop_{\rm mol}+\I)^{\ell+n-j}\otimes (H_{\rm f} +\I)^{j-\ell}    \right)   (\hepsiop_0+\I)^{-n}\,.
\end{eqnarray*}
Since 
\[
 \left(  (\hepsiop_{\rm mol}+\I)^{\ell+n-j}\otimes (H_{\rm f} +\I)^{j-\ell}    \right)   (\hepsiop_0+\I)^{-n}\quad\mbox{and} \quad H_{\rm f}^{j-\ell}(H_{\rm f} +\I)^{-(j-\ell)}
\]
are bounded uniformly in $\epsi$ due to (\ref{prodest1}), it suffices to control terms of the form
\[
 (\hepsiop_{\rm mol})^\ell B (\hepsiop_{\rm mol}+\I)^{-(\ell+n-j)}\,.
\]
By Lemma~\ref{GraphNormLemma} these are controlled again in terms of $\|B\|_{\mathcal{L}(D_{\rm mol}^{\ell+n-j}, D_{\rm mol}^\ell)  }$.

\subsection*{Proof of Lemma~\ref{ChiLemma}}
 The assumption $\| A\|_{\mathcal{L} ( D_0 ,\Hi)}   \leq \delta <1$ implies that for $\psi\in D$
 \[
 \|A\psi \| \leq \delta \|\psi\|_{D_0} = \delta ( \| H_0\psi\|+\|\psi\|)
 \]
 and thus $A$ is $H_0$-bounded with relative bound smaller than $1$.
 The equivalence of the norms follows from
 \[
 \|\psi \|_{D_H} = \| H\psi\| + \|\psi\| \leq \|H_0\psi\| + \|\psi\| + \|A\psi\| \leq \|\psi\|_{D_0} (1+\delta)
 \]
 and
 \[
  \|\psi \|_{D_0} = \| H_0 \psi\| + \|\psi\| \leq \|H \psi\| + \|\psi\| + \|A\psi\| \leq \|\psi\|_{D_H} + \delta \|\psi\|_{D_0}\,.
 \]
 The last claim  follows from the Helffer-Sj\"ostrand formula 
 \[
 \tilde \chi(H) = \frac{1}{\pi}\int_\C \partial_{\bar z} \hat \chi(z)\, (H-z)^{-1}\,\D z\,,
 \]
 where $ \hat \chi $ is an appropriate almost-analytic extension of $\tilde \chi$, 
 and the resolvent formula
 \begin{eqnarray*} 
 \left\| (H_0 -z)^{-1} - (H- z)^{-1} \right\|_{\mathcal{L} ( \Hi,D)} & =& \left\| (H- z)^{-1}  A \,(H_0 -z)^{-1}\right\|_{\mathcal{L} ( \Hi,D)} \\
 &\leq &  \left\| (H- z)^{-1}\right\|_{\mathcal{L} ( \Hi,D)} \left\|  A \right\|_{\mathcal{L} ( D,\Hi )}  \left\| (H_0 -z)^{-1}\right\|_{\mathcal{L} ( \Hi,D)}\\ &
 \leq & \delta (1+\delta) \left( 1 +  \frac{ |z|+1}{|{\rm Im}z| }\right)^2 \,.
 \end{eqnarray*}

 \subsection*{Proof of Lemma~\ref{derivPj}}
Due to the smearing of the nuclear charge it holds $\vrm{nn},V_{\rm en}\in C^\infty_{\rm b}(\R^{3l},C^\infty_{\rm b}(\R^{3r}))$. Note that
\[
\big[\nabla_x,\big(H_{\rm el}(x)-z\big)^{-1}\big] \;=\; \big(H_{\rm el}(x)-z\big)^{-1}\big(\nabla_x\vrm{nn}(x)+\nabla_xV_{\rm en}(x)\big)\big(H_{\rm el}(x)-z\big)^{-1}.
\]
Thus the mapping $x\mapsto (H_{\rm el}(x)-z)^{-1}$ is in $C^1_{\rm b}\big(\R^{3l},\mathcal{L}(\Hi_{\rm el})\big)$. 
Since $E_j$ is separated by a gap, the projection $P_j(x)$ associated to $E_j(x)$ is given via the Riesz formula: 
\begin{eqnarray*}
 P_j(x) &=& \frac{\I}{2\pi}\oint_{\gamma(x)} \big(H_{\rm el}(x)-z\big)^{-1}\,\D z,
 \end{eqnarray*} 
 where $\gamma(x)$ is positively oriented closed curve encircling $E_j(x)$ once. It can be chosen independent of $x$ locally because the gap condition is uniform.
 Therefore $(H_{\rm el}(\cdot)-z)^{-1}\in C^1_{\rm b}\big(\R^{3l},\mathcal{L}(\Hi_{\rm el}) \big)$ entails that $P_j\in C^1_{\rm b}\big(\R^{3l},\mathcal{L}(\Hi_{\rm el}) \big)$. By
 \begin{eqnarray*}
 E_j(x)P_j(x) \ \;=\ \;H_{\rm el}(x)P_j(x) &=& \frac{\I}{2\pi}\oint_{\gamma(x)} z\big(H_{\rm el}(x)-z\big)^{-1}\,\D z
 \end{eqnarray*} 
 we obtain $E_jP_j\in C^1_{\rm b}(\R^{3l},\mathcal{L}(\Hi_{\rm el}))$. Then $E_j={\rm tr}_{L^2(\R^{3s})}\big(E_jP_j\big)\in C^1_{\rm b}(\R^{3l})$. For it holds
 \begin{eqnarray*}
 \nabla_x\, {\rm tr}\big(E_jP_j\big) &=& \nabla_x\, {\rm tr}\big((E_jP_j)P_j\big)\;\;=\;\;  {\rm tr}\big((\nabla_x E_jP_j)P_j\,+\, (E_jP_j)\nabla_x P_j\big)\\
 &=&  {\rm tr}\big((\nabla_x E_jP_j)P_j\big) \,+\, {\rm tr}\big((E_jP_j)\nabla_x P_j\big)  \ \;<\ \; \infty
 \end{eqnarray*} 
 because $P_j$ and $E_j P_j$ are trace-class operators and the product of a trace-class operator and a bounded operator is again a trace-class operator (see e.g.\ \cite{ReSi1}, Theorem VI.19). The argument for higher derivatives goes along the same lines.
 
  For the last claim we  observe that $H_{\rm el}^n \partial_x^\alpha P_j$ is bounded for any $\alpha\in \N_0^{3l}$ and any $n\in\N_0$. Since $H_{\rm el}$, $E_j$ and  $P_j$  have bounded and smooth derivatives, this can be easily seen inductively  by differentiating the identity
 \[
 0 = (H_{\rm el} - E_j )^n P_j\,.
 \]

\subsection*{Proof of Lemma~\ref{TechLem1}}
 We proceed by induction. 
For $m=1$, i.e.\ $|\alpha|\leq 2$,  we have by standard elliptic estimates that $\epsi^{|\alpha|}\partial^\alpha_x$ is relatively bounded by $-\epsi^2\Delta_x$.
   Now $\hepsiop_{\rm mol}$ has the form
 \[
 \hepsiop_{\rm mol} = -\epsi^2 \Delta_x \otimes 1 + H_{\rm el}(x) = -\epsi^2 \Delta_x \otimes 1 \,+\, 1\otimes \underbrace{(-\Delta_y + \vrm{ee}(y))}_{=: H_{\rm e,0}\geq 0} \,+\, V_{\rm en} (x,y),
 \]
 where $V _{\rm en}$ is bounded with bounded derivatives. Hence
 \[
 \| -\epsi^2\Delta_x \psi \| \leq \|  \hepsiop_{\rm mol} \psi \| + \|V_{\rm en}\|_\infty \|\psi\|
 \]
and thus 
\[
\| A_\alpha \epsi^\alpha \partial^\alpha_x \psi \| \leq \|A_\alpha\| \,\|   \epsi^\alpha \partial^\alpha_x \psi \|\leq C  ( \|  \hepsiop_{\rm mol} \psi \| + \|\psi\|) = C \|\psi\|_{D(\hepsiop_{\rm mol})}
\]
 with a constant $C$ independent of $\epsi$.

 Now assume that we proved the assertion for operators of order $n-1$ and let $|\alpha|=n$ and  $m =  \lceil n/2\rceil  $. Then $A(x) \epsi^{|\alpha|}\partial^\alpha_x$
 is relatively bounded by $(-\epsi^2 \Delta_x)^m$ again by standard elliptic estimates. Using the induction hypothesis we find that
 \begin{eqnarray*}
 \| (\epsi^2 \Delta_x)^m \psi \|& \leq & C (\|  \hepsiop_{\rm mol} (\epsi^2 \Delta_x)^{m-1} \psi \|+\|   (\epsi^2 \Delta_x)^{m-1} \psi \|) \\&\leq &
 C (
  \|  (\epsi^2 \Delta_x)^{m-1} \hepsiop_{\rm mol}  \psi \| + \| [ \hepsiop_{\rm mol} ,(\epsi^2 \Delta_x)^{m-1}] \psi \|+\|   (\epsi^2 \Delta_x)^{m-1} \psi \|) \\
  &\leq & C \| \psi \|_{D( (\hepsiop_{\rm mol})^m)}\,,
 \end{eqnarray*}
   since $ (\epsi^2 \Delta_x)^{m-1}$ and $[ \hepsiop_{\rm mol} ,(\epsi^2 \Delta_x)^{m-1}] $ are both differential operators of order at most $2m-2\leq n-1$.
 
 For the second claim note that 
 \begin{eqnarray*}\lefteqn{
 \| (\hepsiop_{\rm mol})^k A_\alpha \epsi^\alpha \partial^\alpha_x \psi \|}\\ &\leq &\|  [ (\hepsiop_{\rm mol})^k ,A_\alpha ] \epsi^\alpha \partial^\alpha_x \psi   \|
 +  \| A_\alpha [ (\hepsiop_{\rm mol})^k , \epsi^\alpha \partial^\alpha_x]  \psi \|+  \| A_\alpha    \epsi^\alpha \partial^\alpha_x  (\hepsiop_{\rm mol})^k \psi \|\\
 &\leq & C \| \psi \|_{D( (\hepsiop_{\rm mol})^{m+k})}\,,
 \end{eqnarray*}
 since $ [ (\hepsiop_{\rm mol})^k ,A_\alpha ]\epsi^\alpha \partial^\alpha_x $ and $[ (\hepsiop_{\rm mol})^k , \epsi^\alpha \partial^\alpha_x]  $ are  admissible of order $2k-1+ n\leq 2(k+m)$.

\subsection*{Proof of Lemma~\ref{TechLem}}

All the operators appearing are differential operators with coefficients $A_\alpha$ that are composed of derivatives of $P_0$ and $R$, i.e.\ of $\partial_x^\alpha P_0$ and $\partial_x^\beta R$. So they are all admissible in the sense of Lemma~\ref{TechLem1}. It remains to show that also commutators of the coefficients $A_\alpha$ with $(\hepsiop_{\rm mol})^k$  are admissible, which in turn follows if $H_{\rm e}^k\partial_x^\beta A_\alpha$ is bounded for any $\beta\in\N_0^{3l}$.   
Now according to Lemma~\ref{derivPj} $\partial^\beta_x P_0(x) \in \mathcal{L}\big(\Hi_{\rm el}, D(H_{\rm el}^{n})\big)$ for any $n$ and clearly also $\partial_x^\beta R(x) \in \mathcal{L}\big(  D(H_{\rm el}^{n})\big)$
for any $n$. Since every coefficient $A_\alpha$ appearing in the construction contains at least one factor of the type $\partial^\beta_x P_0$, the claim follows.

\subsection*{Proof of Lemma~\ref{CommuLemma}}

First take any $\phi\in C^\infty_0(\R)$ and $\hat \phi$ an appropriate almost analytic extension. 
Then    the Helffer-Sj\"ostrand formula  implies
 \begin{eqnarray*}
\lefteqn{ \left\| [ \phi(H), A] {\bf 1}_I(H)\right\|_{\mathcal{L}(\Hi, D(H^n))} }\\ 
&& \ \,=\ \, \left\| \frac{1}{\pi}\int_\C \partial_{\bar z} \hat \phi(z)\, (H-z)^{-1} [ A,H]  {\bf 1}_I(H) \, (H-z)^{-1}\,\D z\right\|_{\mathcal{L}(\Hi, D(H^n))}\\
 && \ \,\leq \ \,  \frac{\delta}{\pi}\int_\C | \partial_{\bar z}  \hat \phi(z) |  \| (H-z)^{-1} \|_{\mathcal{L}(  D(H^n))}  \, \| (H-z)^{-1}\|_{\mathcal{L}(\Hi)}  \,\D z\\
 && \ \,\leq \ \,  \frac{\delta}{\pi}\int_\C | \partial_{\bar z}  \hat \phi(z) |  \frac{1}{|{\rm Im}(z)|^2}  \,\D z\;\leq\; C_\phi \delta \,.
 \end{eqnarray*}
 Taking the adjoint shows that also $ \| {\bf 1}_I(H) [ \phi(H), A] \|_{\mathcal{L}(\Hi)}\leq C_\phi\delta$. With the bound $\| {\bf 1}_I(H) \|_{\mathcal{L}(\Hi, D(H^n))}\leq C_n $ we get also 
 \[
  \| {\bf 1}_I(H) [ \phi(H), A] \|_{\mathcal{L}(\Hi, D(H^n))}  \leq C_nC_\phi\delta\,.
  \]
 Now choose $\chi\in C^\infty_0(\R)$ with supp$\chi \subset I$ and $\chi|_{\tilde I} =1$. This implies 
 $\tilde \chi =\tilde \chi {\bf 1}_I(H)$,  $  \chi =  \chi {\bf 1}_I(H)$ and $\tilde \chi   \chi =\tilde\chi$.
Using the above estimate for $\phi=\chi$ and $\phi=\tilde \chi$, we get
\begin{eqnarray*}
  \| \tilde \chi(H)  A - A\tilde\chi(H) \| &=&   \| \tilde \chi(H)  A - A\chi(H)\tilde\chi(H) \| 
 \leq   \| \tilde \chi(H)  A - \chi (H)A \tilde\chi(H) \| + C_\chi\delta
 \\ &\leq &( C_\chi + C_nC_{\tilde \chi} )\delta \,.
\end{eqnarray*}

\subsection*{Proof of Lemma~\ref{ProjectorLemma}}

 We first state another lemma that will be used in the proof.
 
  \begin{lemma}\label{NormDnLemma}
Let $(H , D(H ) )$ be a   self-adjoint operator  and equip the domains $D^n:= D(H^n)$ with the graph norms
$\| \psi \|_{D^n} := \sum_{i=0}^n \| H^i \psi \|$ and let $N\in\N$. If $A\in\mathcal{L}(\Hi)$ satisfies 
\[
\| [A,H] \|_{\mathcal{L}(D^{n}, D^{n-1})} \leq \delta_{n}
\]
for all $1\leq n\leq N $,
then
\begin{equation}\label{ANormDn}
\|A\|_{\mathcal{L}(  D^n)} \leq n \|A\| + \sum_{i=1}^{n } \delta_i 
\end{equation}
for all $1\leq n\leq N $. 
Moreover, if in addition
\[
\| (A-z)^{-1} \|_{\mathcal{L}(\Hi)} \leq \alpha
\]
and $\delta_n <\frac{1}{2} \frac{1}{2^{2n}\alpha}$, then
\begin{equation}\label{RNormDn}
\| (A-z)^{-1} \|_{\mathcal{L}(D^n)} \leq  2^{2n }    \alpha
\end{equation}
for all $n\leq N$.
 \end{lemma}

\begin{proof}
We proceed by induction.
\begin{eqnarray*}
\| A\psi\|_{D^n} &=& \sum_{i=0}^n \| H^i A\psi \| \;\leq\;  \|A\psi\| + \sum_{i=0}^{n-1 } \| H^{i } AH\psi\| + \sum_{i=0}^{n -1}  \| H^{ i} [H,A] \psi\|\\
&=& \|A\psi\| + \|AH\psi\|_{D^{n-1}} + \| [H,A] \psi\| _{D^{n-1}}\\[1mm]
&\leq & \|A\|\,\|\psi\| + \|A\|_{\mathcal{L}(D^{n-1})} \|\psi\|_{D^n} + \| [H,A] \|_{\mathcal{L}(D^n,D^{n-1}) } \|\psi\|_{D^n}
\end{eqnarray*}
and thus
\[
\|A\|_{\mathcal{L}(D^{n})} \leq  \|A\|_{\mathcal{L}(D^{n-1})} + \|A\| +\delta_n\,.
\]
Since for $n=1$ the computation yields $\|A\|_ {\mathcal{L}(D )} \leq \|A\|+\delta_1$, this implies (\ref{ANormDn}). For (\ref{RNormDn}) we proceed analogously and abbreviate $R:= (A-z)^{-1}$.
\begin{eqnarray*}
\| R\psi\|_{D^n} &=& \sum_{i=0}^n \| H^i R\psi \| \;\leq\;   \|R\psi\| + \sum_{i=0}^{n-1} \| H^{i} RH\psi\| + \sum_{i=0}^{n-1} \| H^{i} [H,R] \psi\|\\
&=& \|R\psi\| + \|RH\psi\|_{D^{n-1}} + \| R[A,H] R\psi\| _{D^{n-1}}\\[1mm]
&\leq & \|R\|\,\|\psi\| + \|R\|_{\mathcal{L}(D^{n-1})} \|\psi\|_{D^n} +  \|R\|_{\mathcal{L}(D^{n-1})}  \delta_n \|R\|_{\mathcal{L}(D^{n})}   \|\psi\|_{D^n}
\end{eqnarray*}
and thus
\[
\|R\|_{\mathcal{L}(D^{n})} \leq \frac{\|R\| +\|R\|_{\mathcal{L}(D^{n-1})} }{1- \delta_n \|R\|_{\mathcal{L}(D^{n-1})}}\,.
\]
For $n=1$ the this yields $\|R\|_{\mathcal{L}(D )} \leq 4 \alpha$ if $1-\delta_1\alpha>\frac{1}{2}$ and by induction one obtains~(\ref{RNormDn}).  
 \end{proof}

 Sine $\tilde Q$ is self-adjoint in $\mathcal{L}(\Hi)$, (\ref{Property4}) implies that the spectrum of $\tilde Q$
 is located in   balls  of radius $  \delta_1$ around $0$ and~$1$. 
  Thus, for $\delta <\frac{1}{2}$   the curve $\gamma : [0,2\pi) \to \C$, $\gamma (\theta) = 1+  \frac{1}{2} \E^{\I \theta}$,
 is contained in the resolvent set of $\tilde Q  \in \mathcal{L}(\Hi)$   and we can define
 \[ 
 Q := \frac{\I}{2\pi}\oint_\gamma (\tilde Q- z)^{-1}\,\D z
 \]
 as a bounded operator in $\mathcal{L}(\Hi)$.  Note that  $Q$ is just the spectral projection of $\tilde Q$  related to the spectrum near $1$.
 For simplicity  write $R(z) :=(\tilde Q- z)^{-1}$ and assume $  \delta  <\frac{1}{4}$.
  Then for $z\in\gamma$ we have 
 $\|R(z)\| \leq 4$ and by Lemma~\ref{NormDnLemma} for $\delta_n  < \frac{1}{2}\frac{1}{2^{2n } }$ also
 \begin{equation}\label{ResoBound}
 \|R(z)\|_{\mathcal{L}(D^n)} \leq 4\cdot 2^{2n}= 4^{n+1}
 \end{equation}
 is uniformly bounded on $\gamma$. Hence $\|Q\|_{\mathcal{L}(D^n)}\leq 4^{n+1}$ for all $n\leq N$.
  The fact that $\tilde Q   - Q$ has spectrum only in a ball around $0$ of size $\delta$ implies  
 \[ 
  \|  \tilde Q  - Q\|_{\mathcal{L}(\Hi)} \leq \delta\,.
 \]
   To estimate the difference also in $\mathcal{L}(D^n)$
 we use Nenciu's formula \cite{Nen}
\[
Q - \tilde Q = \frac{\I}{2\pi} \oint_\gamma \frac{ R(z) -   R(1-z)}{1-z}  \D z\, \left( \tilde Q \tilde Q-\tilde Q\right)\,.
\]
Now (\ref{Property4}) together with (\ref{ResoBound})  implies the second part of (\ref{Prop6}) with $C_n = 8\cdot 4^{n+1}$.
 From 
 \[
 [H, Q ]  = \frac{\I}{2\pi}\oint_\gamma R(z) [H  , \tilde Q  ]  \,R(z) \,\D z\,,
 \]
and (\ref{Property0}) it follows that 
\[
\| [H , Q ]  \|_{\mathcal{L}(D^n, D^{n-1})}  \leq   C_n \,\delta_n 
\]
with $C_n = 4^{2n+1}$.
From now on we will not keep track of the exact value of $C_n$ and   increase it as necessary in the following steps. 
But it should always be noted that it depends only on $n$.

Now pick $\tilde \chi \in C^\infty_0(\R)$ with support in $(e-1,E+1)$ and with
$\tilde \chi {\bf 1}_{E+\frac{1}{2}} =  {\bf 1}_{E+\frac{1}{2}}$.
Then
(\ref{Property3}) together with Lemma~\ref{CommuLemma} implies that there is a constant $C$ depending only on $\tilde \chi$, which in turn can be fixed given $E$, such that 
\[ 
\| [ \tilde\chi(H), \tilde Q] \|_{\mathcal{L}( \Hi , D^n   )} \leq C\,\beta_1\,.
\]
Thus for $z\in\gamma$
 \[
 \left\| \left[ R(z),  \tilde \chi(H ) \right] \right\|_{\mathcal{L}(\Hi,D^n) }  =  \left\|R(z) \left[   \tilde \chi(H),\tilde Q  \right]   R(z)\right\|_{\mathcal{L}(\Hi,D^n) }  \leq C_n\,C\,\beta_1\,,
 \]
 which shows that
 \begin{eqnarray*}
 [H, Q ] \,{\bf 1}_{\rm E+\frac{1}{2} }(H ) &=& \frac{\I}{2\pi}\oint_\gamma R(z) [H  , \tilde Q  ]  \,R(z) \,\D z \,\tilde \chi(H){\bf 1}_{\rm E+\frac{1}{2} }(H ) \\
 &=&\frac{\I}{2\pi}\oint_\gamma R(z) [H  , \tilde Q  ] \,\tilde \chi(H) \,R(z) \,\D z \,{\bf 1}_{\rm E+\frac{1}{2} }(H )\\
 && + \, \frac{\I}{2\pi}\oint_\gamma R(z) [H  , \tilde Q  ]  \,[R(z),\tilde \chi(H)] \,\D z \,{\bf 1}_{\rm E+\frac{1}{2} }(H )
 \end{eqnarray*}
 implies
 \[
  \left\|   [H, Q ] \,{\bf 1}_{\rm E+\frac{1}{2} }(H )\right\|_{\mathcal{L}(\Hi,D^n) } \leq  C_nC_E\,\beta_1\,.
 \]
 Finally we get   (\ref{Property7}) using again  Nenciu's formula and (\ref{Property5}).

\section*{Acknowledgements}
We are grateful to Luca Tenuta for his engagement in the initial phase of this project.
We thank J\"urg Fr\"ohlich, Marcel Griesemer, Christian Hainzl, Michael Sigal, Herbert Spohn,  Hans-Michael Stiepan and Jan-Eric Str\"ang   for helpful remarks and stimulating discussions. This work was supported by the German Science Foundation (DFG) and by the German Israeli Foundation (GIF).

\end{document}